\documentclass[british]{article}
\usepackage{amsmath}
\usepackage{amsthm}
\usepackage{libertineRoman}
\usepackage{biolinum}

\usepackage[libertine]{newtxmath}
\usepackage[T1]{fontenc}
\usepackage[latin9]{inputenc}
\usepackage{geometry}
\usepackage{color}
\usepackage{prettyref}
\usepackage{refstyle}
\usepackage{mathrsfs}
\usepackage{cancel}
\usepackage{graphicx}

\makeatletter

\AtBeginDocument{\providecommand\Figref[1]{\ref{Fig:#1}}}
\AtBeginDocument{\providecommand\Defref[1]{\ref{Def:#1}}}
\AtBeginDocument{\providecommand\Secref[1]{\ref{Sec:#1}}}
\AtBeginDocument{\providecommand\Eqref[1]{\ref{Eq:#1}}}
\AtBeginDocument{\providecommand\Subsecref[1]{\ref{Subsec:#1}}}
\AtBeginDocument{\providecommand\Notaref[1]{\ref{Nota:#1}}}
\AtBeginDocument{\providecommand\Remref[1]{\ref{Rem:#1}}}
\AtBeginDocument{\providecommand\Algref[1]{\ref{Alg:#1}}}
\AtBeginDocument{\providecommand\Lemref[1]{\ref{Lem:#1}}}
\AtBeginDocument{\providecommand\Corref[1]{\ref{Cor:#1}}}
\AtBeginDocument{\providecommand\Thmref[1]{\ref{Thm:#1}}}
\AtBeginDocument{\providecommand\Propref[1]{\ref{Prop:#1}}}
\AtBeginDocument{\providecommand\Claimref[1]{\ref{Claim:#1}}}
\AtBeginDocument{\providecommand\Exaref[1]{\ref{Exa:#1}}}
\RS@ifundefined{subsecref}
  {\newref{subsec}{name = \RSsectxt}}
  {}
\RS@ifundefined{thmref}
  {\def\RSthmtxt{theorem~}\newref{thm}{name = \RSthmtxt}}
  {}
\RS@ifundefined{lemref}
  {\def\RSlemtxt{lemma~}\newref{lem}{name = \RSlemtxt}}
  {}

\theoremstyle{plain}
\newtheorem{thm}{\protect\theoremname}
\theoremstyle{remark}
\newtheorem{notation}[thm]{\protect\notationname}
\theoremstyle{definition}
\newtheorem{defn}[thm]{\protect\definitionname}
\theoremstyle{remark}
\newtheorem{rem}[thm]{\protect\remarkname}
\theoremstyle{plain}
\newtheorem{fact}[thm]{\protect\factname}
\theoremstyle{plain}
\newtheorem{lem}[thm]{\protect\lemmaname}
\theoremstyle{plain}
\newtheorem{cor}[thm]{\protect\corollaryname}
\theoremstyle{plain}
\newtheorem{prop}[thm]{\protect\propositionname}
\theoremstyle{plain}
\newtheorem{lyxalgorithm}[thm]{\protect\algorithmname}
\theoremstyle{definition}
\newtheorem{example}[thm]{\protect\examplename}
\theoremstyle{remark}
\newtheorem{note}[thm]{\protect\notename}
\theoremstyle{remark}
\newtheorem{claim}[thm]{\protect\claimname}
\theoremstyle{plain}
\newtheorem*{prop*}{\protect\propositionname}

\usepackage{tabu}

\usepackage{tikz}
\usetikzlibrary{calc}

\usepackage{enumitem}
\setlistdepth{20}
\renewlist{itemize}{itemize}{20}
\setlist[itemize]{label=$\cdot$}

\setlist[itemize,1]{label=\textbullet}
\setlist[itemize,2]{label=--}
\setlist[itemize,3]{label=*}
\setlist[itemize,4]{label=$\circ$}
\setlist[itemize,5]{label=$\square$}

\usepackage{multicol}

\usepackage{qtree}

\usepackage{hyperref}

\hypersetup{
    colorlinks=true,
    linkcolor=blue,
    filecolor=magenta,      
    urlcolor=cyan,
	citecolor=blue,
    }

\hypersetup{hidelinks,
backref=true,
pagebackref=true,
hyperindex=true,
breaklinks=true,
colorlinks=true,
urlcolor=blue,
bookmarks=true,
bookmarksopen=false}

\usepackage[super]{nth}

\usepackage{braket}

\usepackage{flafter}

\usepackage{psvectorian}

\newref{thm}{name=theorem~,Name=Theorem~,names=theorems~,Names=Theorems~}
\newref{def}{name=definition~,Name=Definition~,names=definitions~,Names=Definitions~}
\newref{cor}{name=corollary~,Name=Corollary~,names=corollaries~,Names=Corollaries~}
\newref{lem}{name=lemma~,Name=Lemma~,names=lemmas~,Names=Lemmas~}
\newref{claim}{name=claim~,Name=Claim~,names=claims~,Names=Claims~}
\newref{sec}{name=section~,Name=Section~,names=sections~,Names=Sections~}
\newref{subsec}{name=subsection~,Name=Subsection~,names=subsections~,Names=Subsections~}
\newref{prop}{name=proposition~,Name=Proposition~,names=propositions~,Names=Propositions~}
\newref{rem}{name=remark~,Name=Remark~,names=remarks~,Names=Remarks~}
\newref{alg}{name=algorithm~,Name=Algorithm~,names=algorithms~,Names=Algorithms~}
\newref{note}{name=note~,Name=Note~,names=notes~,Names=Notes~}
\newref{nota}{name=notation~,Name=Notation~,names=notations~,Names=Notations~}
\newref{exa}{name=example~,Name=Example~,names=examples~,Names=Examples~}
\newref{tab}{name=table~,Name=Table~,names=tables~,Names=Tables~}

\usepackage{color}
\definecolor{purple}{RGB}{120,20,120}
\newcommand\branchcolor[2]{{\color{#1} #2}}

\usepackage{tikzsymbols}

\usepackage[bottom]{footmisc}

\@ifundefined{showcaptionsetup}{}{%
 \PassOptionsToPackage{caption=false}{subfig}}
\usepackage{subfig}
\makeatother

\usepackage{babel}
\usepackage[style=alphabetic]{biblatex}
\providecommand{\algorithmname}{Algorithm}
\providecommand{\claimname}{Claim}
\providecommand{\corollaryname}{Corollary}
\providecommand{\definitionname}{Definition}
\providecommand{\examplename}{Example}
\providecommand{\factname}{Fact}
\providecommand{\lemmaname}{Lemma}
\providecommand{\notationname}{Notation}
\providecommand{\notename}{Note}
\providecommand{\propositionname}{Proposition}
\providecommand{\remarkname}{Remark}
\providecommand{\theoremname}{Theorem}

\begin{document}
\title{Oracle separations of hybrid quantum-classical circuits}
\date{January 2022\textsc{\small{}}}
\author{Atul Singh Arora\thanks{Corresponding author (atul.singh.arora@gmail.com; asarora@caltech.edu).
Institute for Quantum Information and Matter, California Institute
of Technology; Department of Computing and Mathematical Sciences,
California Institute of Technology.}, Alexandru Gheorghiu\thanks{Institute for Theoretical Studies, ETH Z{\"u}rich},
Uttam Singh\thanks{Centre for Theoretical Physics, Polish Academy of Sciences}}

\maketitle
\global\long\def\polylog{{\rm poly}({\rm log}(n))}%
\global\long\def\poly{{\rm poly}(n)}%
\global\long\def\dQC{{\rm QC}_{d}}%
\global\long\def\dCQ{{\rm CQ}_{d}}%
\global\long\def\tr{{\rm tr}}%
\global\long\def\perm#1#2{^{#1}\!P_{#2}}%
\global\long\def\comb#1#2{^{#1}C_{#2}}%
\global\long\def\paths{{\rm paths}}%
\global\long\def\parts{{\rm parts}}%
\global\long\def\mat{{\rm mat}}%
\global\long\def\td{{\rm TD}}%
\global\long\def\f{\mathcal{F}}%
\global\long\def\g{\mathcal{G}}%
\global\long\def\negl{{\rm negl}(n)}%
\global\long\def\CQd{\mathsf{CQ_{d}}}%
\global\long\def\QCd{\mathsf{QC_{d}}}%
\global\long\def\BQP{\mathsf{BQP}}%
\global\long\def\BPP{\mathsf{BPP}}%
\global\long\def\QNC{\mathsf{QNC}}%
\global\long\def\CQ#1{\mathsf{CQ_{#1}}}%
\global\long\def\QC#1{\mathsf{QC_{#1}}}%
\global\long\def\CQdp{\mathsf{CQ_{d'}}}%
\global\long\def\QCdp{\mathsf{QC_{d'}}}%

An important theoretical problem in the study of quantum computation,
that is also practically relevant in the context of near-term quantum
devices, is to understand the computational power of \emph{hybrid
models}, that combine polynomial-time classical computation with short-depth
quantum computation. Here, we consider two such models: $\CQd$ which
captures the scenario of a polynomial-time classical algorithm that
queries a $d$-depth quantum computer many times; and $\QCd$ which
is more analogous to \emph{measurement-based quantum computation}
and captures the scenario of a $d$-depth quantum computer with the
ability to change the sequence of gates being applied depending on
measurement outcomes processed by a classical computation. Chia, Chung
and Lai (STOC 2020) and Coudron and Menda (STOC 2020) showed that
these models (with $d=\log^{\mathcal{O}(1)}(n)$) are strictly weaker
than $\BQP$ (the class of problems solvable by polynomial-time quantum
computation), relative to an oracle, disproving a conjecture of Jozsa
in the relativised world.

In this paper, we show that, despite the similarities between $\CQd$
and $\QCd$, the two models are incomparable, i.e. $\CQd\not\subseteq\QCd$
and $\QCd\not\subseteq\CQd$ relative to an oracle. In other words,
we show that there exist problems that one model can solve but not
the other and vice versa. We do this by considering new oracle problems
that capture the distinctions between the two models and by introducing
the notion of an \emph{intrinsically stochastic oracle}, an oracle
whose responses are inherently randomised, which is used for our second
result. While we leave showing the second separation relative to a
standard oracle as an open problem, we believe the notion of stochastic
oracles could be of independent interest for studying complexity classes
which have resisted separation in the standard oracle model. Our constructions
also yield simpler oracle separations between the hybrid models and
$\BQP$, compared to earlier works.

\pagebreak{}

\tableofcontents{}

\newpage{}

\section{Introduction}

One of the major goals of modern complexity theory is understanding
the relations between various models of efficient quantum computation
and classical computation. It is widely believed that the set of problems
solvable by polynomial-time quantum algorithms, denoted $\BQP$, is
strictly larger than that solvable by polynomial-time classical algorithms,
denoted $\BPP$ \cite{bernstein1997quantum,Shor1997,Arute2019}. This,
so-called Deutsch-Church-Turing thesis, is motivated by the existence
of a number of problems which seem to be classically intractable,
but which admit efficient (polynomial-time) quantum algorithms. Notable
examples include period finding \cite{Simon1997}, integer factorisation
\cite{Shor1997}, the discrete-logarithm problem \cite{Shor1997},
solving Pell's equation \cite{Hallgren02}, approximating the Jones
polynomial \cite{Aharonov2006} and others \cite{Jordan}. In a seminal
paper of Cleve and Watrous \cite{Cleve2000}, it was shown that many
of these problems (period finding, factoring and discrete-logarithm)
can also be solved by quantum circuits of logarithmic-depth together
with classical pre- and post-processing. In the language of complexity
theory, we say that these problems are contained in the complexity
class $\BPP^{\QNC}$. This motivated Jozsa to conjecture that $\BPP^{\QNC}=\BQP$
\cite{Jozsa06}. In other words, it was conjectured that the computational
power of general polynomial-size quantum circuits is fully captured
by quantum circuits of polylogarithmic depth (the class $\QNC$) together
with classical pre- and post-processing.

In fact, Jozsa's conjecture was more broad than simply asserting that
$\BPP^{\QNC}=\BQP$. To explain why, one needs to distinguish between
two \emph{hybrid models} of classical computation and short-depth
quantum computation. These two models capture the idea of interleaving
polynomial-time classical computations with $d$-depth quantum computations.
The first model, denoted $\CQd$, refers to polynomial-time classical
circuits that can call $d$-depth quantum circuits, as a subroutine,
and use their output (see \Figref{CQd}). This is analogous to $\BPP^{\QNC}$,
except the quantum circuits are of depth $d$ rather than $\log^{O(1)}(n)$
(where $n$ is the size of the input). The second model, denoted $\QCd$,
refers to $d$-depth quantum circuits that can, at each circuit layer,
call polynomial-time classical circuits, as a subroutine and use their
output (see \Figref{QCd}). Importantly, the classical circuits \emph{cannot
be invoked coherently}---part of the quantum state produced by the
circuit is measured and the classical subroutine is invoked on the
classical outcome of that measurement.

Though $\CQd$ and $\QCd$ circuits seem similar, they both capture
different aspects of short-depth quantum computation together with
polynomial-time classical computation. $\CQd$ captures the scenario
of a classical algorithm that queries a short-depth quantum computer
many times, while $\QCd$ is more analogous to \emph{measurement-based
quantum computation} in that it captures a short-depth quantum computer
with the ability to change the sequence of gates being applied depending
on measurement outcomes processed by a classical computation.

Jozsa's conjecture is the assertion that the power of $\BQP$ is fully
captured by logarithmic-depth quantum computation interspersed with
polynomial-depth classical computation. This can be interpreted either
as $\BQP=\CQd$ or $\BQP=\QCd$ with $d=\log^{O(1)}(n)$. Aaronson
later proposed as one of his ``ten semi-grand challenges in quantum
complexity theory'' to find an oracle separation between $\BQP$
and the two hybrid models \cite{Aaronson05}. That is, to provide
an oracle $O$ and a problem defined relative to this oracle which
is contained in $\BQP^{O}$ but not in $\QCd^{O}$ or $\CQd^{O}$.
This was resolved in the landmark works of Chia, Chung and Lai \cite{CCL2020}
and, independently, Coudron and Menda \cite{CM2020}. They showed
that indeed such an oracle exists and also proved a hierarchy theorem
for the two hybrid models showing that $\QCd^{O}\subsetneq\QCdp^{O}$
and $\CQd^{O}\subsetneq\CQdp^{O}$, whenever $d'=2d+1$. However,
both works left it as an open problem whether there exists an oracle
separation between $\QCd$ and $\CQd$.

\subsection{Contributions}

In our work, we resolve this question by showing that, indeed, there
exist oracles $O_{1}$ and $O_{2}$ such that $\QCd^{O_{1}}\not\supseteq\CQd^{O_{1}}$
and $\CQd^{O_{2}}\not\supseteq\QCd^{O_{2}}$. In other words, the
two hybrid models are incomparable---there exist problems that one
can solve but not the other and vice versa. As corollaries, we also
obtain simpler oracle constructions for separating the two hybrid
models from $\BQP$. It's important to note that for our results $O_{1}$
is a \emph{standard oracle}, i.e. one that performs the mapping $\ket{x}\ket{z}\overset{O_{1}}{\to}\ket{x}\ket{z\oplus f(x)}$,
for some function $f:X\to Z$. However, the oracle $O_{2}$ is one
which we refer to as an \emph{intrinsically stochastic oracle}, which
is defined for a function $f:X\times Y\to Z$ and a probability distribution
$\mathbb{F}_{Y}$ over $Y$. Classically, for each query $x\in X$,
it samples a new $y\sim\mathbb{F}_{Y}$ and returns $f(x,y)$. Quantumly,
for every application, it samples a new $y\sim\mathbb{F}_{Y}$ and
performs the mapping $\ket{x}\ket{z}\overset{O_{2}}{\to}\ket{x}\ket{z\oplus f(x,y)}$.

\begin{figure}
\begin{centering}
\subfloat[${\rm QNC}_{d}$ scheme; $U_{i}$ are single depth unitaries; $\Pi$
is a projector in the computational basis.\label{fig:QNCd}]{\begin{centering}
\includegraphics[width=3.5cm]{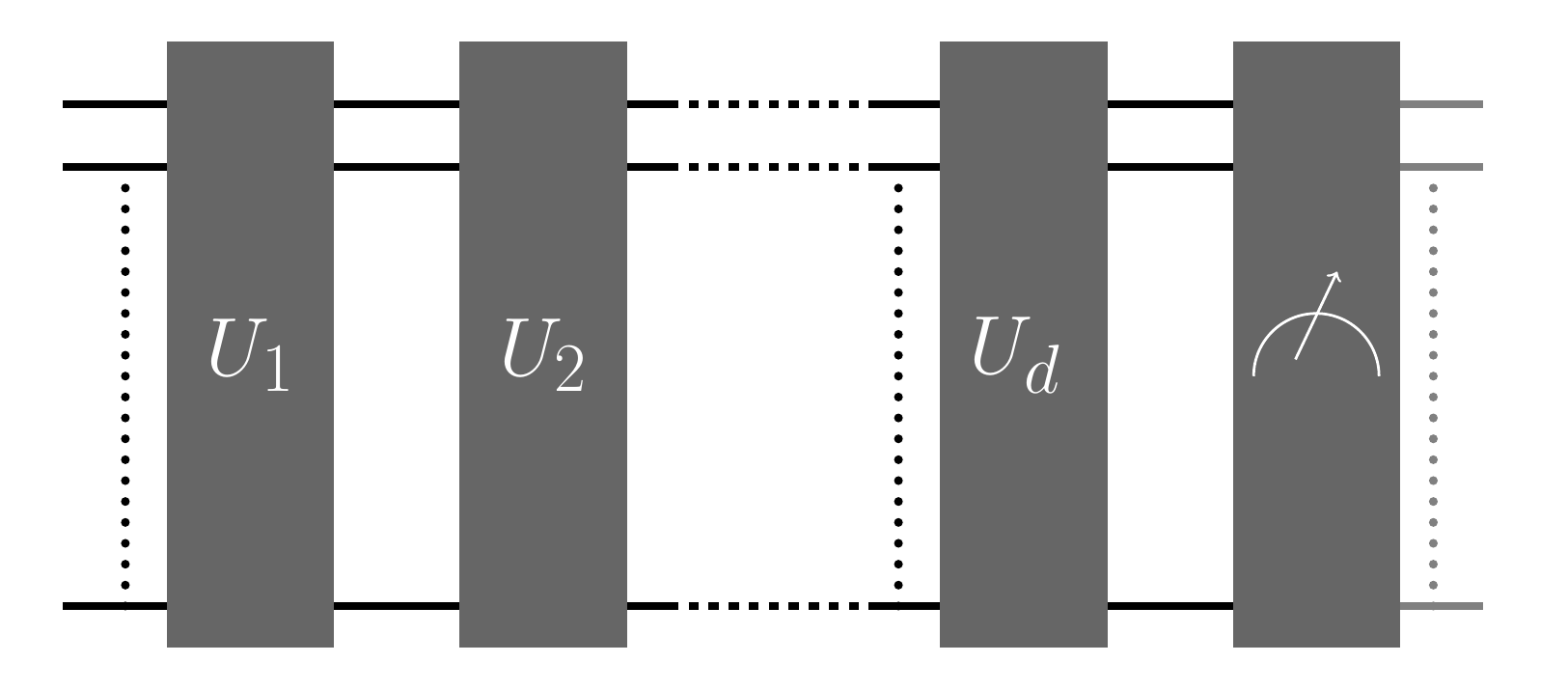}
\par\end{centering}
}
\par\end{centering}
\begin{centering}
\subfloat[${\rm QC}_{d}$ circuit; $U_{i}$ are single layered unitaries, $\mathcal{A}_{c,i}$
are classical poly-sized circuits (in the figure, the subscript for
$\mathcal{A}_{c}$ has been dropped) and the measurements are in the
computational basis. Dark lines denote qubits.\label{fig:QCd}]{\begin{centering}
\includegraphics[width=8cm]{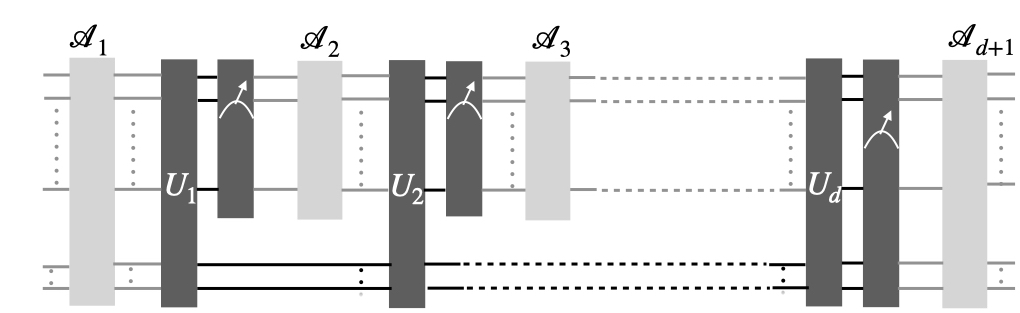}
\par\end{centering}
}\enskip{}\subfloat[${\rm CQ}_{d}$ circuit; for clarity, we dropped the indices in $\mathcal{A}_{c}$
and the second indices in $U_{1,i},U_{2,i}\dots U_{d,i}$. \label{fig:CQd}]{\centering{}\includegraphics[width=8cm]{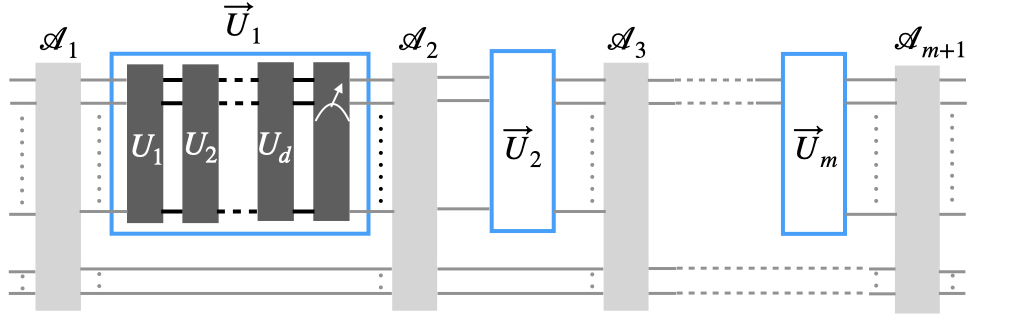}}
\par\end{centering}
\caption{The three circuit models we consider.}
\end{figure}

\begin{table}[h]
    \centering
    \begin{tabu}{ c c c c c } 
        
     \hline
     ~          & $\QCd$ &           $\CQd$    & Oracle Model & ~ \\ 
     \hline
     $d'$-SeS & $d'+1\le d \le 2d'+2 $  &~~~~~$d\le 1$    & Standard & This work \\ 
     $d'$-SCS &             ~~~~$d \le 4$       & $d'+1\le d\le d'+5$ & Stochastic & This work \\ 
     $d'$-SSP & $d'+1\le d \le 2d'+2$   & $d'+1 \le d\le 2d'+1$        & Standard & CCL \\
     \hline
    \end{tabu}
    \caption{Summary of our results: bounds on the smallest depth, $d$, needed to solve $d'$-Serial Simon's ($d'$-SeS) Problem  and $d'$-Shuffled Collisions-to-Simon's ($d'$-SCS) Problem in the two hybrid models of computation compared to those of $d'$-Shuffling Simon's Problem ($d'$-SSP) introduced by CCL. \label{tab:main}}
\end{table}

We summarise our main results in \Tabref{main}. In particular,
we obtain our first result by introducing an oracle problem that we
call the \emph{$d$-Serial Simon's} ($d$-SeS) Problem, which lets
us show the following:
\begin{thm}[informal]
There exists a standard oracle, $O$, such that for all $d>0$, $\QCd^{O}\not\supseteq\CQd^{O}.$

\emph{Intuitively, since $\CQd$ can perform polynomially-many calls
to a $d$-depth quantum circuit, we\textquoteright d like to define
a problem which forces one to sequentially solve $d$ Simon\textquoteright s
problems}.\footnote{Recall that Simon's problem is the following: given oracle access
to a $2$-to-$1$ function $f:\{0,1\}^{n}\to\{0,1\}^{n}$, such that
there exists an $s\in\{0,1\}^{n},s\neq0^{n}$, for which $f(x)=f(y)\iff x=y\oplus s$,
find $s$. The function $f$ is referred to as a Simon function.}\emph{ In other words, the period of each Simon's function can be
viewed as a key which unlocks in the oracle a different Simon function.
The $d$-SeS problem is to find the period of the $d$'th function.
This problem can be solved in $\mathsf{CQ}_{1}$, since the poly-time
algorithm can invoke $d$ separate constant-depth quantum circuits
to solve each Simon problem. However, due to the serial nature of
the problem, we show that no $\QCd$ algorithm can succeed.}
\end{thm}

For our second separation, we consider a problem we call the $d$-Shuffled
Collisions-to-Simon\textquoteright s (\ensuremath{d}-SCS) Problem
and show that:
\begin{thm}[informal]
There exists an intrinsically stochastic oracle, $O$, such that
for all $d\ge4$, $\CQd^{O}\not\supseteq\QCd^{O}$.
\end{thm}

This separation is perhaps more surprising because $\QCd$ has only
$d$ quantum layers while $\CQd$ has polynomially-many quantum layers
(as it can invoke $d$-depth quantum circuits polynomially-many times).
Importantly, however, in $\CQd$ the quantum states prepared when
invoking a $d$-depth quantum circuit are measured entirely. In contrast,
in $\QCd$ some qubits are measured and the outcomes are sent to a
classical subroutine, while other qubits remain unmeasured and maintain
their coherence in between the layers of classical computation. This
observation is the basis for defining $d$-SCS, which can still be
viewed as a version of Simon's problem, in which one has to find a
secret period $s$. To explain the idea, we briefly recap the textbook
quantum algorithm for Simon's problem. One first queries the oracle
to the Simon function in superposition to obtain the state $\sum_{x}\ket{x}\ket{g(x)}$,
where $g:\{0,1\}^{n}\to\{0,1\}^{n}$ is the Simon function, with period
$s\in\{0,1\}^{n},s\neq0^{n}$. The second register is then measured,
yielding 
\begin{equation}
\frac{1}{\sqrt{2}}(\ket{x}+\ket{x\oplus s})\ket{y}\label{eq:superpos}
\end{equation}
with $y=g(x)=g(x\oplus s)$. Finally, by applying Hadamard gates and
measuring the first register, one obtains a string $w\in\{0,1\}^{n}$,
such that $w\cdot s=0$. Repeating this in parallel $O(n)$ times
yields a system of linear equations from which $s$ can be uniquely
recovered. 

Note that in Simon's algorithm, the main component of the quantum
algorithm is to generate superpositions over colliding pre-images
(i.e. pre-images that map to the same image). For $d$-SCS, a first
difference with respect to the standard Simon's problem is that the
oracle doesn\textquoteright t provide direct access to the Simon function
$g$. Instead, consider three functions: a random 2-to-1 function,
$f$, a function $p$ which maps colliding pairs of $f$ to those
of $g$ and the inverse of $p$. Since $f$ and $g$ are both $2$-to-1,
$p$ is a bijection and its inverse is well defined. If one is given
access to all three functions, then evidently, the problem is equivalent
to that of Simon's: the quantum algorithm would create superpositions
over colliding pairs of $f$ and then \emph{coherently} evaluate the
bijection together with its inverse, to obtain superpositions over
colliding pairs of $g$. To obtain the desired separation---a problem
that $\QCd$ can solve but $\CQd$ cannot---we only allow restricted
access to the bijection $p$. This restriction, denoted $p'$, may
be seen as a form of encrypted access to $p$. Specifically, to evaluate
the bijection (or its inverse) $p$ on a colliding pair $(x_{0},x_{1})$
of $f$, i.e. $f(x_{0})=f(x_{1})=y$, $p'$ additionally takes a ``key''
$h(y)$ as input. Access to the function $h,$which produces the key
$h(y)$, is also restricted. It is given via a \emph{shuffler} $\Xi$---an
encoding which ensures that at least depth $d$ is required to evaluate
$h$. 

So far, given access to $f$, $p'$ (which encodes $p$) and $\Xi$
(which encodes $h$), the goal is to find $s$ (the period of $g$).
It is easy to see that a ${\rm QC}_{4}$ algorithm can solve this
problem. The algorithm first creates equal superpositions of colliding
pairs for $f$ as in Simon's algorithm. The measured image register,
which now contains the classical value $f(x_{0})=f(x_{1})=y$, can
be used to evaluate $h(y)$ via $\Xi$ by expending $d$ classical
depth. The ${\rm QC}_{4}$ circuit can perform this computation---it
can perform poly depth classical computation while maintaining quantum
coherence. With $h(y)$ known, in a second quantum layer, the bijection
(and its inverse) $p$ is evaluated, via $p'$, on the superposition
of colliding pairs $\left|x_{0}\right\rangle +\left|x_{1}\right\rangle $
to obtain $\left|p(x_{0})\right\rangle +\left|p(x_{1})\right\rangle $
(up to normalisation). Recalling that $(p(x_{0}),p(x_{1}))$ are pre-images
of the Simon function $g$, it only remains to make Hadamard measurements
to obtain the system of linear equations which yield the secret $s$.

Why can't a $\CQd$ algorithm also solve this problem? By making the
oracle provide access to a generic $2$-to-$1$ function, $f$, instead
of the Simon function, we made it so that superpositions over colliding
pairs are no longer related by the period $s$. Indeed, the only way
to obtain any information about $s$ is to query the bijection. But
in the case of $\CQd$ the quantum subroutines must measure their
states completely before invoking the classical subroutines. This
means that the quantum subroutines essentially obtain no information
about $s$. We also know that only the classical subroutines can obtain
access to the bijection as the shuffler can only be invoked by a circuit
of depth at least $d$. In essence, we've made it so that if a $\CQd$
algorithm could solve the problem with a polynomial number of oracle
queries, then Simon's problem could also be solved classically with
a polynomial number of queries, which we know is impossible.

Formalising the above intuition turns out to be surprisingly involved.
This is due to the following subtlety: the classical subroutines of
$\CQd$ can obtain some information about $g$ by querying the bijection
(essentially obtaining evaluations of $g$) and one would have to
show that this information is insufficient for recovering $s$ even
when invoking short depth quantum subroutines. We overcome this barrier
by using a stochastic oracle in our proof but nevertheless conjecture
that the result should also hold in the standard oracle setting. Thus,
our final modification to Simon's problem, leading to the definition
of $d$-SCS is to make it so that the oracle does not even provide
direct access to $f$ but instead, given a single bit as input, performs
the following mappings:
\begin{equation}
0\stackrel{\mathcal{S}}{\to}(x_{0},f(x_{0})),\quad1\stackrel{\mathcal{S}}{\to}(x_{1},f(x_{1}))\label{eq:S}
\end{equation}
where $(x_{0},x_{1})$ is a randomly chosen colliding pair of $g$.
In other words, the oracle picks a random colliding pair for $g$
and outputs either the first element in the pair or the second, depending
on whether the input was $0$ or $1$. In the quantum setting, this
translates to

\begin{equation}
\ket{0}_{\mathsf{{B}}}\ket{0}_{\mathsf{{X}}}\ket{0}_{\mathsf{{Y}}}\overset{\mathcal{S}}{\to}\ket{0}_{\mathsf{{B}}}\ket{x_{0}}_{\mathsf{{X}}}\ket{f(x_{0})}_{\mathsf{{Y}}}\label{eq:scosoracle1}
\end{equation}

\begin{equation}
\ket{1}_{\mathsf{{B}}}\ket{0}_{\mathsf{{X}}}\ket{0}_{\mathsf{{Y}}}\overset{\mathcal{S}}{\to}\ket{1}_{\mathsf{{B}}}\ket{x_{1}}_{\mathsf{{X}}}\ket{f(x_{1})}_{\mathsf{{Y}}}\label{eq:scosoracle2}
\end{equation}
We can therefore see that if the input qubit is in an equal superposition
of $\ket{0}$and $\ket{1}$ (while all other registers are initialised
as $\ket{0}$) we would obtain the state

\begin{equation}
\frac{1}{\sqrt{2}}(\ket{0}_{\mathsf{{B}}}\ket{x_{0}}_{\mathsf{{X}}}+\ket{1}_{\mathsf{{B}}}\ket{x_{1}}_{\mathsf{{X}}})\ket{f(x_{0})=f(x_{1})}_{\mathsf{{Y}}}\label{eq:superpos-2}
\end{equation}
This is the same type of state as that of Equation \ref{eq:superpos},
except the superposition is over pre-images of $f$ rather than the
Simon function $g$. The oracle will still provide access to the bijection
and the shuffler in the manner explained above, so that the stochastic
access to $f$ does not change the fact that a $\mathsf{{QC}_{4}}$
algorithm is still able to solve the problem (one is still able to
map these states to superpositions over colliding pairs of $g$).
But this stochastic access further restricts what a $\CQd$ algorithm
can do, as now even if the classical subroutines are used to obtain
evaluations of $g$, the stochastic nature of the oracle guarantees
that these evaluations are uniformly random points. With this restriction
in place, one can then show that $\CQd$ cannot solve the problem
unless Simon's problem can be solved classically with polynomially-many
queries. 

\subsection{Overview of the techniques }

We begin with sketching the idea behind the hardness of $d$-SeS for
$\QCd$ circuits. To this end, we first describe the $d$-SeS problem
in some more detail. 

\paragraph{The $d$-Serial Simon's Problem (informal).}

Sample $d+1$ random Simon's functions $\{f_{i}\}_{i=0}^{d}$ with
periods $\{s_{i}\}_{i=0}^{d}$. The problem is to find the period,
$s_{d}$, of the last Simon's function. However, only access to $f_{0}$
is given directly. Access to $f_{i}$, for $i\ge1$, is given via
a function $L_{f_{i}}$ which outputs $f_{i}(x)$ if the input is
$(s_{i-1},x)$ and $\perp$ otherwise, i.e. to access the $i$th Simon's
function, one needs the period of the $(i-1)$th Simon's function
(see \Defref{cSerialSimonsOracle} and \Defref{cSerialSimonsProblem}
for details). 

\paragraph{The lower bound technique.}

\begin{figure}
\begin{centering}
\includegraphics[width=9cm]{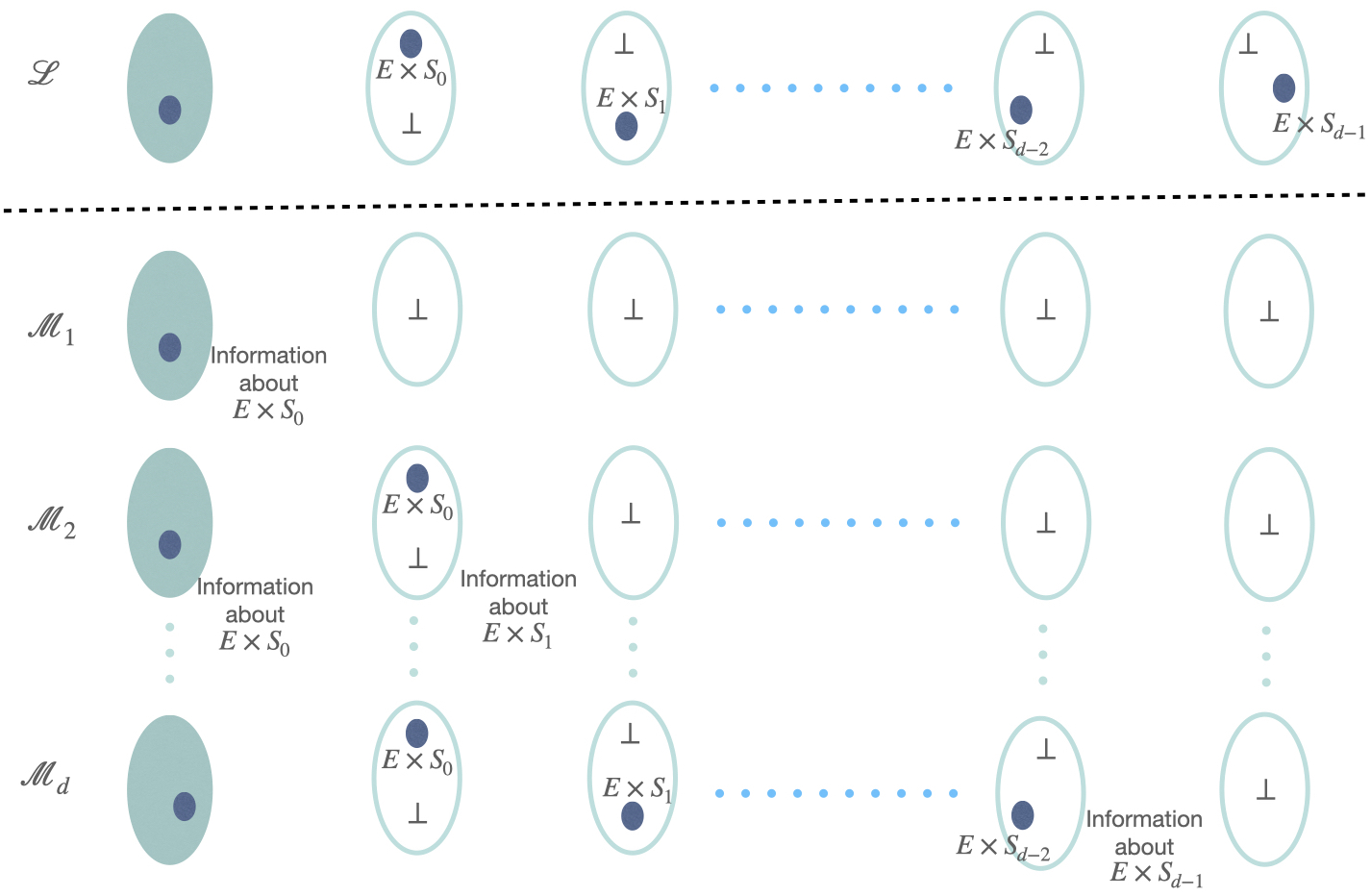}
\par\end{centering}
\caption{\label{fig:ShadowsQNCd}Shadows for ${\rm QNC}_{d}$. The ovals represent
the query domains of the various sub-oracles. Shaded regions denote
a non-$\perp$ response. Let $\bar{S}_{i}$ denote the query domain
at which $\mathcal{L}$ and $\mathcal{M}_{i}$ differ. The construction
ensures that $\mathcal{M}_{1},\dots\mathcal{M}_{i-1}$ contain no
information about $\bar{S}_{i}$. Convention: The left-most sub-oracles
are numbered zero (and they respond with non-$\perp$ on $\{0,1\}^{n}$).}
\end{figure}

We now sketch the idea behind the proof that $\QCd$ circuits solve
$d$-SeS with at most negligible probability (see \Secref{-Serial-Simon's-Problem}).
Denote the initial set of oracles by $\mathcal{L}$. Let $\mathcal{M}_{i}$,
for $1\le i\le d$, be identical to $\mathcal{L}$, except the oracles
$i,\dots d$ output $\perp$ at all inputs (see \Figref{ShadowsQNCd}).
Using a hybrid argument, we show that the $d$-depth quantum circuit
$\mathcal{C}:=U_{d+1}\circ\mathcal{L}\circ U_{d}\dots\mathcal{L}\circ U_{1}$
behaves essentially like $U_{d+1}\circ\mathcal{M}_{d}\circ U_{d}\dots\mathcal{M}_{1}U_{1}$.
The intuition is that since all $s_{i}$ are unknown before the oracles
are invoked, the domain at which the oracles respond non-trivially
(i.e. $\neq\perp$) is hard to find. After $\mathcal{M}_{i-1}$ is
invoked, the ``non-trivial domain'' of the $i$th oracle can be
learnt. It is therefore exposed in $\mathcal{M}_{i}$. Evidently,
since $\mathcal{M}_{1},\dots\mathcal{M}_{d}$ do not contain any information
about $f_{d}$, we conclude the $d$-depth quantum circuit $\mathcal{C}$
cannot solve $d$-SeS with non-negligible probability. The conclusion
continues to hold even when efficient classical computation is allowed
between the unitaries. At a high level, this is because one can condition
on the classical queries yielding $\perp$ and show that this happens
with high probability. Since only polynomially many queries are possible,
the aforementioned analysis goes through largely unchanged. This proof
is quite straightforward and consequently also serves as a considerable
simplification of the proof of the $\QCd$ depth hierarchy theorem
proved by CCL.

\paragraph{~}

We now turn to the ideas behind the hardness of the $d$-SCS problem
for $\CQd$ circuits. We describe the $d$-SCS problem in some detail
and to this end, first introduce the $d$-Shuffler construction\footnote{This is a close variant of the $d$-Shuffling technique introduced
by CCL.} which is used for enforcing statements like ``$k$ is a function
which requires depth larger than $d$ to evaluate''.

\paragraph{The $d$-Shuffler (informal).}

Intuitively, a $d$-Shuffler is an oracle which encodes a function,
say $f:\{0,1\}^{n}\to\{0,1\}^{n}$, in such a way that one needs to
make $d$ sequential calls to it, to access $f$. Consider $d$ random
permutations, $f'_{0},f'_{1}\dots f'_{d-1}$, from $\{0,1\}^{2n}\to\{0,1\}^{2n}$.
Define $f'_{d}$ to be such that $f'_{d}(f'_{d-1}\dots f'_{0}(x))\dots)=f(x)$
for $x\in\{0,1\}^{n}$. Define $f_{i}$, for $0\le i\le d$, to be
$f'_{i}(f'_{i-1}(\dots f'_{0}(x)\dots))$ for $x\in\{0,1\}^{n}$ and
$\perp$ otherwise. We say $(f_{i})_{i=0}^{d}$ is the \emph{$d$-Shuffler}
encoding $f$. We also use $\Xi$ to abstractly refer to it. Using
an argument, similar to the one above, one can show that no $d$-depth
quantum circuit can access $f$ via a $d$-Shuffler, with non-negligible
probability (see \Defref{dShuffler}).

\paragraph{The $d$-Shuffled Collisions-to-Simon's Problem (informal).}

Consider the following functions on $n$-bit functions. Uniformly
sample $f$ from all 2-to-1 functions, $g$ from all Simon's functions,
and $h$ from all $1$-to-$1$ functions. Let $p$ be some canonical
bijection which maps colliding pairs of $f$ to those of $g$ (and
$p_{{\rm inv}}$ be the inverse). 

Let $p'$ be such that $p'(h(f(x)),x)=p(x)$ and $\perp$ if the input
is not of that form ($p'_{{\rm inv}}$ is similarly defined). Let
$\Xi$ be a $d$-Shuffler encoding $h$. The problem is, given access
to $\mathcal{S}$ (a stochastic oracle, encoding $f$ as in \Eqref{S}),
$p'$, $p'_{{\rm inv}}$ and $\Xi$, find the period $s$ of $g$. 

\paragraph{The lower bound technique.}

We briefly outline the broad argument for why $\CQd$ circuits can
solve $d$-SCS with at most negligible probability (see \Secref{ShuffledCollisionsToSimons}).
Consider the first classical circuit: It has vanishing probability
of learning a collision from $\mathcal{S}$ (since $y$ is chosen
stochastically). So, with overwhelming probability, it can learn at
most polynomially many values of $p,p_{{\rm inv}}$, $h$ and non
colliding values of $f$. 

For the subsequent quantum circuit, we expose those values in the
shadows of $p',p'_{{\rm inv}}$ and $\Xi$ (a $d$-Shuffler which
encodes $h$). Here shadows are the analogues of $\mathcal{M}_{i}$
introduced above. For $\mathcal{S}$, we condition on never seeing
the $y$s which already appeared (this will happen with overwhelming
probability). Without access to $h$ for the new $y$ values (since
access to $h$ is via $\Xi$ which requires at least depth $d$),
the circuit cannot distinguish the shadows of $p',p'_{{\rm inv}}$
from the originals; the shadows contain no information about $s$
(informally, they could contain $z$ or $z\oplus s$ but not both).
Finally, since $f$ is a random two-to-one function, and since it
is known that finding collisions even when given direct oracle access
to $f$ is hard, the quantum output cannot contain collisions (with
non-negligible probability). 

The next classical circuit, did not learn of any collisions from the
quantum output. Further, it did not learn anything about the function,
other than those which evaluate to $y$s which appeared in the previous
step. Conditioned on the same $y$s not appearing (which happens with
overwhelming probability), we can repeat the reasoning from the first
step, polynomially many times. 

In the proof, we first show that one can replace the oracles with
their shadows, without changing the outputs of the circuits noticeably.
We then show that if no collisions are revealed, the shadows contain
no information about $s$ and that collisions are revealed with negligible
probability. For the full details of the proofs we refer the reader
to \Secref{ShuffledCollisionsToSimons}.

\subsection{Organisation}

We begin with introducing the precise models of computation we consider
in \Secref{Models-of-Computation}, followed by brief statements of
the basic results in query complexity which we use, in \Secref{Prerequisite-Technical-Results}.
Our first main result---that $d$-Serial Simon's Problem can be solved
using $\CQ 1$ while it is hard for $\QCd$---is established in \Secref{-Serial-Simon's-Problem}.
The second main result, however, requires more ground work. We first
study the simpler aspects of the $d$-Shuffled Simon's problem in
\Secref{warm-up-dSS} where we introduce the $d$-Shuffler (a close
variant of the $d$-Shuffling technique introduced by CCL). In \Secref{Technical-Results-II}
we generalise the so-called ``sampling argument'' as formulated
by CCL, which may be of independent interest. It is a technique used
for establishing the continued utility of $d$-Shufflers despite repeated
applications. We give a new proof of hardness of $d$-SS for $\CQd$
circuits in \Secref{dSS_hard_2} and use this as an intermediate step
for the proof of our second main result which is delineated in \Secref{ShuffledCollisionsToSimons}.
This final section begins with formalising intrinsically stochastic
oracles, defines the $d$-SCS problem precisely, gives the algorithms
for establishing the upper bounds and concludes with the proof hardness
for $\CQd$.

\subsection{Relation to Prior Work}

Four connections with \parencite{CCL2020} are noteworthy. 

First, as was already noted, we recover the depth hierarchy for $\QCd$
with less effort. This was done by circumventing the use of the $d$-Shuffling
and the ``Russian nesting doll'' technique employed by CCL. Our
result is, in fact, slightly stronger too since if we allow Hadamard
measurements in the $\QCd$ model, our hierarchy theorem becomes optimal,
i.e. it separates $d$ and $d+1$ depth. This was one of their open
questions, as their separation was between $d$ and $2d+1$, even
if we allow Hadamard measurements.\footnote{Without the Hadamard measurements, we believe their upper bound becomes
$2d+2$.}

Second, we give a new proof of hardness of $d$-SS for $\CQd$. This
proof is arguably simpler in that while it uses ideas related to $d$-Shuffling,
it does not rely on the ``Russian nesting doll'' technique.\footnote{We say arguably because in proving the basic properties of the $d$-Shuffler,
one needs more notation. The advantage is that, once done, this does
not interfere with how it is used.} Aside from simplicity, the information required to specify the oracles
in the proof by CCL scaled exponentially with $d$ while in our proof,
there is no dependence on $d$. 

Third, the connection between the ``sampling argument'', the $d$-Shuffler
and $d$-Shuffling. We show that the ``sampling argument'' as used
by CCL holds quite generally and in doing so, obtain a proof which
relies on very simple properties of the underlying distribution. We
also show a composition result which allows us to repeatedly use the
argument. We use this generalisation to establish the utility of the
$d$-Shuffler.

Fourth, we also get a hierarchy theorem for $\CQd$ based on the hardness
of the $d$-SCS problem. While this hierarchy theorem holds in the
stochastic oracle setting and is thus not as strong as that obtained
using $d$-SS (or $d$-Shuffling Simon's problem) which holds in the
standard oracle setting, nevertheless, it yields a finer separation---it
separates $d$ from $d+5$.

\subsection{Conclusion and Outlook}

In addition to resolving the open problems left by the works of~\textcite{CCL2020}
and~\textcite{CM2020}, our results provide additional insights into
the relations between different models of hybrid classical-quantum
computation and sharpen our understanding of these models. Specifically,
we find that $\QCd$ and $\CQd$ solve incomparable sets of problems.
Our results relativise and so one should obtain a similar conclusion
about a hierarchy of analogous classes ($\mathsf{QCQC..QC_{d}}$,
$\mathsf{CQC..CQ_{d}}$ etc). While an important open problem is whether
the stochastic oracle in our second construction can be replaced with
a standard oracle, we speculate that other classes which have resisted
separation in the standard oracle model, may also admit separations
via stochastic oracles.

A potentially interesting application that is hinted by our work,
assuming the oracles can be instantiated with cryptographic primitives,
is that of \emph{tests of coherent quantum control}. In other words,
the idea is to design a task (or a test) which can be solved by a
quantum device having coherent control (the ability to adapt the gates
it will perform based on the measurement results) but which cannot
be solved by a device without coherent control (one which performs
a measurement of all its qubits when providing a response). This is
related to the recent notion of \emph{proofs of quantumness}~\parencite{brakerski2018,Brakerski2020}.
This is a test which can be passed by a $\BQP$ machine but not by
a $\BPP$ one, assuming the intractability of some cryptographic task.
In a test of coherent control, the tester (or verifier) has the ability
to check for a more fine-grained notion of ``quantumness'' on the
part of the quantum device (prover) and not just its ability to solve,
say, $\BQP$ problems. 

\subsection{Acknowledgements}

AG is supported by Dr. Max Rössler, the Walter Haefner Foundation
and the ETH Zürich Foundation. ASA acknowledges funding provided by
the Institute for Quantum Information and Matter. We are thankful
to Thomas Vidick and Jérémie Roland for numerous fruitful discussions.
We are also grateful to Ulysse Chabaud for his helpful feedback on
an early draft of this work.

\section{Models of Computation\label{sec:Models-of-Computation}}

\branchcolor{black}{We begin with defining the computational models we study in this work. }
\begin{notation}
A \emph{single layer unitary}, is defined by a set of one and two-qubit
gates which act on disjoint qubits (so that they can all act parallelly
in a single step). The number of layers in a circuit defines its \emph{depth}.
\end{notation}

\begin{notation}
A promise problem $\mathcal{P}$ is denoted by a tuple $(\mathcal{P}_{0},\mathcal{P}_{1})$
where $\mathcal{P}_{0}$ and $\mathcal{P}_{1}$ are subsets of $\{0,1\}^{*}$
satisfying $\mathcal{P}_{0}\cap\mathcal{P}_{1}=\emptyset$. It is
not necessary that $\mathcal{P}_{0}\cup\mathcal{P}_{1}=\{0,1\}^{*}$. 
\end{notation}

\begin{defn}[${\rm QNC}_{d}$ circuits and ${\rm BQNC}_{d}$ languages]
\label{def:QNCd} Denote by ${\rm QNC}_{d}$ the set of $d$-depth
quantum circuits (see \Figref{QNCd}). 

Define ${\rm BQNC}_{d}$ to be the set of all promise problems $\mathcal{P}=(\mathcal{P}_{0},\mathcal{P}_{1})$
which satisfy the following: for each problem $\mathcal{P}\in{\rm BQNC}_{d}$,
there is a circuit family $\{\mathcal{C}_{n}:\mathcal{C}_{n}\in{\rm QNC}_{d}\text{ and acts on }\poly\text{ qubits}\}$
and for all $x\in\mathcal{P}_{1}$, the circuit $\mathcal{C}_{|x|}$
accepts with probability at least $2/3$rds and for all $x\in\mathcal{P}_{0}$,
the circuit $\mathcal{C}_{|x|}$ accepts with probability at most
$1/3$rd. 
\end{defn}

\begin{defn}[${\rm QC}_{d}$ circuits and ${\rm BQNC}_{d}^{{\rm BPP}}$ languages]
\label{def:dQC} Denote by ${\rm QC}_{d}$ the set of all circuits
which, for each $n\in\mathbb{N}$, act on $\poly$ qubits and bits
and can be specified by
\begin{itemize}
\item $d$ single layered unitaries, $U_{1},U_{2}\dots U_{d}$, 
\item $\poly$ sized classical circuits $\mathcal{A}_{c,1}\dots\mathcal{A}_{c,d},\mathcal{A}_{c,d+1}$,
and 
\item $d$ computational basis measurements 
\end{itemize}
that are connected as in \Figref{QCd}. 

Define ${\rm BQNC}_{d}^{{\rm BPP}}$ to be the set of all promise
problems $\mathcal{P}=(\mathcal{P}_{0},\mathcal{P}_{1})$ which satisfy
the following: for each problem $\mathcal{P}\in{\rm BQNC}_{d}^{{\rm BPP}}$,
there exists a circuit family $\{\mathcal{C}_{n}:\mathcal{C}_{n}\in{\rm QC}_{d}\text{ and acts on }\poly\text{ qubits and bits}$\}
and for all $x\in\mathcal{P}_{1}$, the circuit $\mathcal{C}_{\left|x\right|}$
accepts with probability at least $2/3$rds and for all $x\in\mathcal{P}_{0}$,
the circuit $\mathcal{C}_{\left|x\right|}$ accepts with probability
at most $1/3$rd. 

Finally, for $d=\polylog$, denote the set of languages by ${\rm BQNC}^{{\rm BPP}}$.\\
\end{defn}

\begin{defn}[${\rm CQ}_{d}$ circuits and ${\rm BPP}^{{\rm QNC}_{d}}$ languages]
\label{def:dCQ} Denote by ${\rm CQ}_{d}$ the set of all circuits
which, for each $n\in\mathbb{N}$ and $m=\poly$, act on $\poly$
qubits and bits and can be specified by 
\begin{itemize}
\item $m$ tuples of $d$ single layered unitaries $(U_{1,i},U_{2,i}\dots U_{d,i})_{i=1}^{m}$,
\item $m+1$, $\poly$ sized classical circuits $\mathcal{A}_{c,1}\dots\mathcal{A}_{c,m},\mathcal{A}_{c,m+1}$,
and 
\item $m$ computational basis measurements
\end{itemize}
that are connected as in \Figref{CQd}. 

Define, as above, ${\rm BPP}^{{\rm BQNC}_{d}}$ to be the set of all
promise problems $\mathcal{P}=(\mathcal{P}_{0},\mathcal{P}_{1})$
which satisfy the following: for each problem $\mathcal{P}\in{\rm BPP}^{{\rm BQNC}_{d}}$,
there exists a circuit family $\left\{ \mathcal{C}_{n}:\mathcal{C}_{n}\in{\rm CQ}_{d}\text{ and acts on }\poly\text{ qubits and bits}\right\} $
and for all $x\in\mathcal{P}_{1}$, the circuit $\mathcal{C}_{\left|x\right|}$
accepts with probability at least $2/3$rds and for all $x\in\mathcal{P}_{0}$,
the circuit $\mathcal{C}_{\left|x\right|}$ accepts with probability
at most $1/3$rd. 

Finally, for $d=\polylog$, denote the set of languages by ${\rm BPP}^{{\rm QNC}}$.

\end{defn}

\begin{rem}
Connection with the more standard notation: ${\rm QNC}_{d}$ has depth
$d$ and ${\rm QNC}^{m}$ has depth $\log^{m}(n)$, i.e. ${\rm QNC}^{m}={\rm QNC}_{\log^{m}(n)}$. 
\end{rem}

\branchcolor{black}{Later, it would be useful to symbolically represent these three circuit
models but we mention them here for ease of reference. }
\begin{notation}
\label{nota:CompositionNotation}We use the following notation for
probabilities, ${\rm QNC}_{d}$, ${\rm QC}_{d}$ and ${\rm CQ}_{d}$
circuits.
\begin{itemize}
\item Probability: The probability of an event $E$ occurring, when process
$P$ happens, is denoted by $\Pr[E:P]$. In our context, the probability
of a random variable $X$ taking the value $x$ when process $Y$
takes place is denoted by $\Pr[x\leftarrow X:Y]$. When the process
$Y$ is just a sampling of $X$, we drop the $Y$ and use $\Pr[x\leftarrow X]$.
\item ${\rm QNC}_{d}$: We denote a $d$-depth quantum circuit (see \Defref{dQC}
and \Figref{QNCd}) by $\mathcal{A}=U_{d}\circ\dots\circ U_{1}$ and
(by a slight abuse of notation) the probability that running the algorithm
on all zero inputs yields $x$, by $\Pr[x\leftarrow\mathcal{A}]$
while that on some input state $\rho$ by $\Pr[x\leftarrow\mathcal{A}(\rho)]$.
\item ${\rm QC}_{d}$: We denote a ${\rm QC}_{d}$ circuit (see \Defref{dQC}
and \Figref{QCd}) by $\mathcal{B}=\mathcal{A}_{c,d+1}\circ\mathcal{B}_{d}\circ\mathcal{B}_{d-1}\dots\circ\mathcal{B}_{1}$
where $\mathcal{B}_{i}:=\Pi_{i}\circ U_{i}\circ\mathcal{A}_{c,i}$
and ``$\circ$'' implicitly denotes the composition as shown in
\Figref{QCd}. As above, the probability of running the circuit $\mathcal{A}$
on all zero inputs and obtaining output $x$ is denoted by $\Pr[x\leftarrow\mathcal{B}]$
while that on some input state $\rho$ by $\Pr[x\leftarrow\mathcal{B}(\rho)]$. 
\item ${\rm CQ}_{d}$: We denote a ${\rm CQ}_{d}$ circuit (see \Defref{dCQ}
and \Figref{CQd}) by $\mathcal{C}=\mathcal{A}_{c,m+1}\circ\mathcal{C}_{m}\circ\dots\circ\mathcal{C}_{1}$
where $\mathcal{C}_{i}:=\Pi_{i}\circ U_{d,i}\circ\dots\circ U_{1,i}\circ\mathcal{A}_{c,i}$
and ``$\circ$'' implicitly denotes the composition as shown in
\Figref{CQd}. Again, the probability of running the circuit $\mathcal{C}$
on all zero inputs and obtaining output $x$ is denoted by $\Pr[x\leftarrow\mathcal{C}]$
while that on some input state $\rho$ by $\Pr[x\leftarrow\mathcal{C}(\rho)]$. 
\end{itemize}
\end{notation}

\subsection{The Oracle Versions}

\branchcolor{black}{We consider the standard Oracle/query model corresponding to functions---the
oracle returns the value of the function when invoked classically
and its action is extended by linearity when it is accessed quantumly.
The following also extends to stochastic oracles which are introduced
later in \Subsecref{Intrinsically-Stochastic-Oracle} when we need
them.}
\begin{notation}
An oracle $\mathcal{O}_{f}$ corresponding to a function $f$ is given
by its action on ``query'' and ``response'' registers as $\mathcal{O}_{f}\left|x\right\rangle _{Q}\left|a\right\rangle _{R}=\left|x\right\rangle _{Q}\left|a\oplus f(x)\right\rangle _{R}$.
An oracle $\mathcal{O}_{(f_{i})_{i=1}^{k}}$ corresponding to multiple
functions $f_{1},f_{2}\dots f_{k}$ is given by $\mathcal{O}_{(f_{i})_{i=1}^{k}}\left|x_{1},x_{2}\dots x_{k}\right\rangle _{Q}\left|a_{1},a_{2}\dots a_{k}\right\rangle _{R}=\left|x_{1},x_{2}\dots x_{k}\right\rangle _{Q}\left|a_{1}\oplus f_{1}(x_{1}),a_{2}\oplus f_{2}(x_{2}),\dots a_{k}\oplus f_{k}(x_{k})\right\rangle _{R}$.
\\
We overload the notation; when $\mathcal{O}_{f}$ is accessed classically,
we use $\mathcal{O}_{f}(x)$ to mean it returns $f(x)$. \label{nota:functionToOracle}
\end{notation}

\begin{rem}[${\rm QNC}_{d}^{\mathcal{O}}$, ${\rm QC}_{d}^{\mathcal{O}}$, ${\rm CQ}_{d}^{\mathcal{O}}$]
\label{rem:oracleVersionsQNC-CQ-QC} The oracle versions of ${\rm QNC}_{d}$,
${\rm QC}_{d}$ and ${\rm CQ}_{d}$ circuits are as shown in \Figref{QNCdOracle},
\Figref{dQC_oracle} and \Figref{dCQ_oracle}. We allow (polynomially
many) parallel uses of the oracle even though in the figures we represent
these using single oracles. We do make minor changes to the circuit
models, following \cite{CCL2020} when we consider ${\rm QNC}_{d}$
circuits and ${\rm CQ}_{d}$ circuits---an extra single layered unitary
is allowed to process the final oracle call. 
\end{rem}

We end by explicitly augmenting \Notaref{CompositionNotation} to
include oracles.
\begin{notation}
When oracles are introduced, we use the following notation.
\begin{itemize}
\item ${\rm QNC}_{d}^{\mathcal{O}}$: $\mathcal{A}^{\mathcal{O}}=U_{d+1}\circ\mathcal{O}\circ U_{d}\circ\dots\mathcal{O}\circ U_{1}$
(see \Figref{QNCdOracle})
\item ${\rm QC}_{d}^{\mathcal{O}}:$ $\mathcal{B}^{\mathcal{O}}=\ensuremath{\mathcal{A}_{c,d+1}^{\mathcal{O}}}\circ\text{\ensuremath{\mathcal{B}_{d}^{\mathcal{O}}}}\ensuremath{\circ}\dots\ensuremath{\mathcal{B}_{1}^{\mathcal{O}}}$
where $\mathcal{B}_{i}^{\mathcal{O}}=\Pi_{i}\circ\mathcal{O}\circ U_{i}\circ\mathcal{A}_{c,i}^{\mathcal{O}}$
and $\mathcal{A}_{c,i}^{\mathcal{O}}$ can access $\mathcal{O}$ classically
(see \Figref{dQC_oracle}). 
\item ${\rm CQ}_{d}^{\mathcal{O}}$: $\mathcal{C}^{\mathcal{O}}=\mathcal{A}_{m+1}^{\mathcal{O}}\circ\mathcal{C}_{m}^{\mathcal{O}}\circ\dots\mathcal{C}_{1}^{\mathcal{O}}$
where $\mathcal{C}_{i}^{\mathcal{O}}:=\Pi_{i}\circ U_{d+1,i}\circ\mathcal{O}\circ U_{d,i}\circ\dots\circ\mathcal{O}\circ U_{1,i}\circ\mathcal{A}_{c,i}^{\mathcal{O}}$
where $\mathcal{A}_{c,i}^{\mathcal{O}}$ can access $\mathcal{O}$
classically (see \Figref{dCQ_oracle}).
\end{itemize}
The classes $\left({\rm BQNC}_{d}^{{\rm BPP}}\right)^{\mathcal{O}}$
and $\left({\rm BPP}^{{\rm BQNC}_{d}}\right)^{\mathcal{O}}$ are implicitly
defined to be the query analogues of ${\rm BQNC}_{d}^{{\rm BPP}}$
and ${\rm BPP}^{{\rm BQNC}_{d}}$ (resp.), i.e. class of promise problems
solved by ${\rm QC}_{d}^{\mathcal{O}}$ and ${\rm CQ}_{d}^{\mathcal{O}}$
circuits (resp.). 
\end{notation}

\begin{figure}
\begin{centering}
\subfloat[A ${\rm QNC}_{d}$ circuit with access to oracle $\mathcal{O}$. Following
\cite{CCL2020}, in the oracle version of ${\rm QNC}_{d}$, we allow
it to perform one extra single layered unitary to process the output.
\label{fig:QNCdOracle}]{\begin{centering}
\includegraphics[width=5cm]{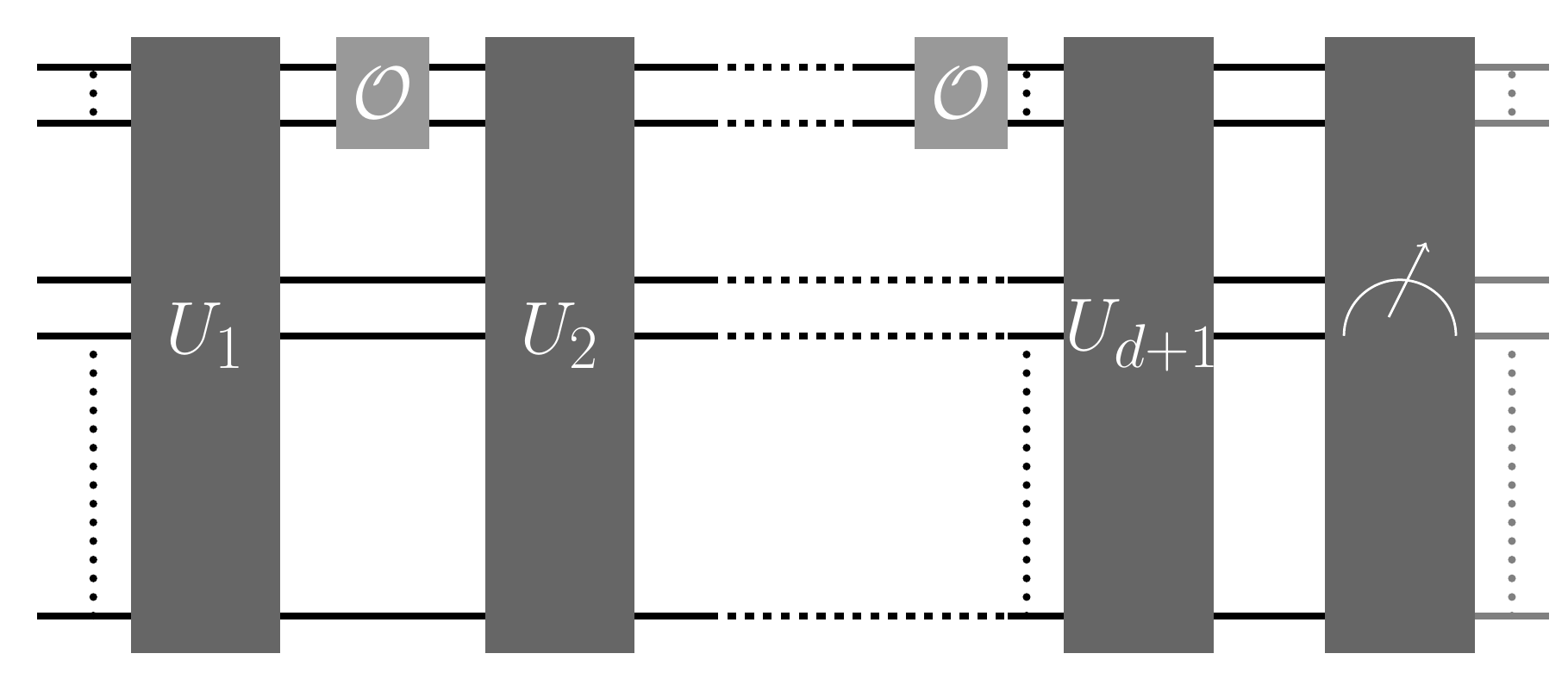}
\par\end{centering}
}
\par\end{centering}
\begin{centering}
\subfloat[A ${\rm QC}_{d}$ circuit with access to an oracle $\mathcal{O}$.
There is no ``extra'' single layered unitary in this model. \label{fig:dQC_oracle}]{\begin{centering}
\includegraphics[width=10cm]{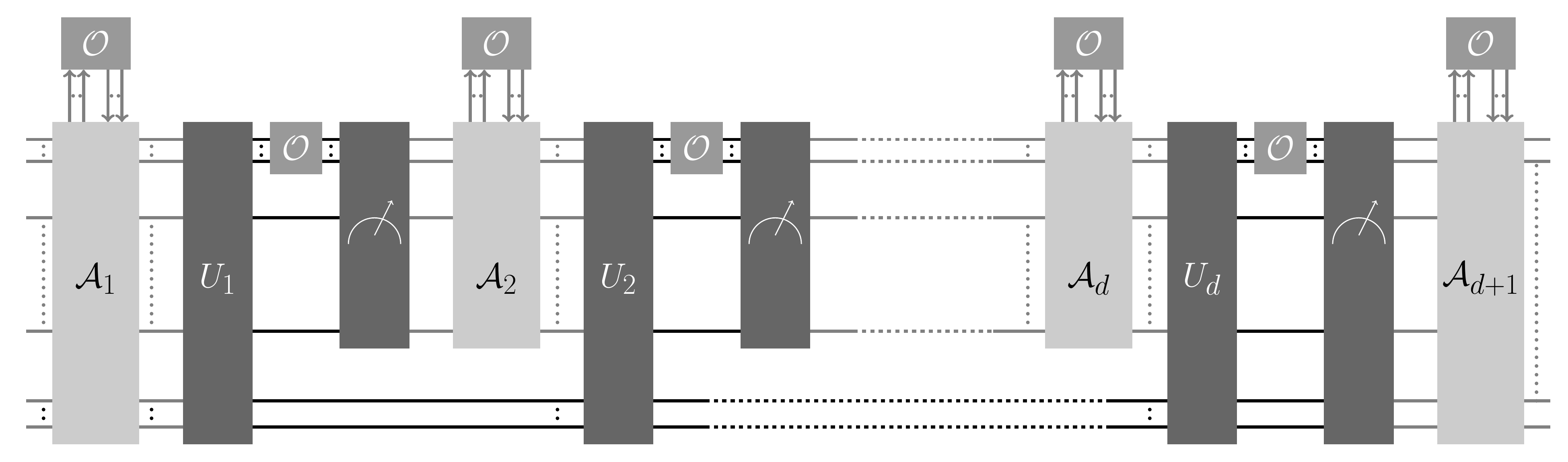}
\par\end{centering}
}\enskip{}\subfloat[A ${\rm CQ}_{d}$ circuit with access to an oracle $\mathcal{O}$.
Again, following \parencite{CCL2020}, we allow an extra single layer
unitary to process the result of the last oracle call.\label{fig:dCQ_oracle}]{\begin{centering}
\includegraphics[width=8cm]{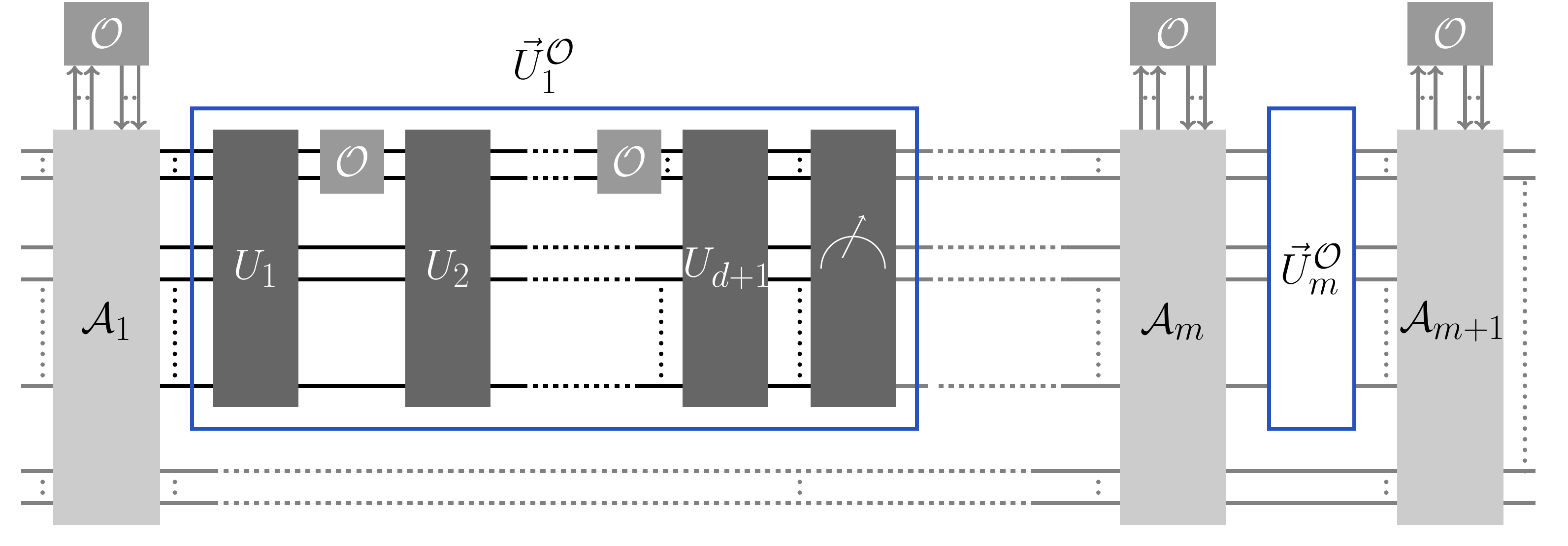}
\par\end{centering}
}
\par\end{centering}
\caption{The same three circuit models, but with oracle access.}

\end{figure}

\section{(Known) Technical Results I \label{sec:Prerequisite-Technical-Results}}

\branchcolor{black}{The basic result we use, following \textcite{CCL2020}, is a simplified
version of the so-called ``one-way to hiding'', O2H lemma, introduced
by \textcite{ambainis_quantum_2018}. Informally, the lemma says the
following: suppose there are two oracles $\mathcal{O}$ and $\mathcal{Q}$
which behave identically on all inputs except some subset $S$ of
their input domain. Let $\mathcal{A}^{\mathcal{O}}$ and $\mathcal{A}^{\mathcal{Q}}$
be identical quantum algorithms, except for their oracle access, which
is to $\mathcal{O}$ and $\mathcal{Q}$ respectively. Then, the probability
that the result of $\mathcal{A}^{\mathcal{O}}$ and $\mathcal{A}^{\mathcal{Q}}$
will be distinct, is bounded by the probability of finding the set
$S$. We suppress the details of the general finding procedure and
only focus on the case of interest for us here.}

\subsection{Standard notions of distances}

\branchcolor{black}{We quickly recall some notions of distances that appear here.}
\begin{defn}
Let $\rho,\rho'$ be two mixed states. Then we define
\begin{itemize}
\item Fidelity: ${\rm F}(\rho,\rho'):=\tr(\sqrt{\sqrt{\rho}\rho'\sqrt{\rho}})$
\item Trace Distance: ${\rm TD}(\rho,\rho'):=\frac{1}{2}\tr\left|\rho-\rho'\right|$
and
\item Bures Distance: ${\rm B}(\rho,\rho'):=\sqrt{2-2F(\rho,\rho')}$.
\end{itemize}
\end{defn}

\begin{fact}
For any string $s$, any two mixed states, $\rho$ and $\rho'$, and
any quantum algorithm $\mathcal{A}$, we have 
\[
\left|\Pr[s\leftarrow\mathcal{A}(\rho)]-\Pr[s\leftarrow\mathcal{A}(\rho')]\right|\le B(\rho,\rho').
\]
\end{fact}

\subsection{The O2H lemma}
\begin{notation}
In the following, we treat $W$ as the workspace register for our
algorithm which is left untouched by the Oracle and recall that $Q$
represents the query register and $R$ represents the response register.

In fact, we wish to allow parallel access to $\mathcal{Q}$ and we
represent this by allowing queries to a tuple of inputs, $\boldsymbol{q}=(q_{1},q_{2}\dots q_{\poly})$
simultaneously. We use boldface to represent such tuples. The query
register in this case is denoted by $\boldsymbol{Q}$. 
\end{notation}

\begin{defn}[$U^{\mathcal{L}\backslash S}$]
\label{def:ULS}Suppose $U$ acts on $\boldsymbol{Q}RW$, $\mathcal{L}$
is an oracle that acts on $\boldsymbol{QR}$ and $S$ is a subset
of the query domain of $\mathcal{L}$. We define 
\[
U^{\mathcal{L}\backslash S}\left|\psi\right\rangle _{\boldsymbol{QR}W}\left|0\right\rangle _{B}:=\mathcal{L}U_{S}U\left|\psi\right\rangle _{\boldsymbol{QR}W}\left|0\right\rangle _{B}
\]
where $B$ is a qubit register, and $U_{S}$ flips qubit $B$ if any
query is made inside the set $S$, i.e. 
\[
U_{S}\left|\boldsymbol{q}\right\rangle _{\boldsymbol{Q}}\left|b\right\rangle _{B}:=\begin{cases}
U_{S}\left|\boldsymbol{q}\right\rangle _{\boldsymbol{Q}}\left|b\right\rangle _{B} & \text{if }\boldsymbol{q}\cap S=\emptyset\\
U_{S}\left|\boldsymbol{q}\right\rangle _{\boldsymbol{Q}}\left|b\oplus1\right\rangle _{B} & \text{otherwise.}
\end{cases}
\]
Here we treat $\boldsymbol{q}$ as a set when we write $\boldsymbol{q}\cap S$. 
\end{defn}

For notational simplicity, in the following, we drop the boldface
for the query and response registers as they do not play an active
role in the discussion.
\begin{defn}[{$\Pr[{\rm find}:U^{\mathcal{L}\backslash S},\rho]$}]
 \label{def:prFind}Let $U^{\mathcal{L}\backslash S}$ be as above
and suppose $\rho\in{\rm D}(QRWB)$. We define 
\[
\Pr[{\rm find}:U^{\mathcal{L}\backslash S},\rho]:={\rm tr}[\mathbb{I}_{QRW}\otimes\left|1\right\rangle \left\langle 1\right|_{B}U^{\mathcal{L}\backslash S}\circ\rho].
\]
This will depend on $\mathcal{L}$ and $S$. When $\mathcal{L}$ and
$S$ are random variables, we additionally take expectation over them.
\end{defn}

\begin{rem}
\label{rem:psiphi0phi1}Let $U^{\mathcal{L}\backslash S}$ be as in
\Defref{ULS} and let $\left|\psi\right\rangle \in QRW$. Note that
we can always write
\[
\mathcal{L}U\left|\psi\right\rangle _{QRW}=\left|\phi_{0}\right\rangle _{QRW}+\left|\phi_{1}\right\rangle _{QRW}
\]
where $\left|\phi_{0}\right\rangle $ and $\left|\phi_{1}\right\rangle $
contains queries outside $S$ and inside $S$ respectively, i.e. $\left\langle \phi_{0}|\phi_{1}\right\rangle =0$.
Further, we can write
\[
U^{\mathcal{L}\backslash S}\left|\psi\right\rangle _{QRW}\left|0\right\rangle _{B}=\left|\phi_{0}\right\rangle _{QRW}\left|0\right\rangle _{B}+\left|\phi_{1}\right\rangle _{QRW}\left|1\right\rangle _{B}.
\]
\end{rem}

\begin{lem}[\parencite{CCL2020,ambainis_quantum_2018} O2H]
 \label{lem:O2H}Let 
\begin{itemize}
\item $\mathcal{L}$ be an oracle which acts on $QR$ and $S$ be a subset
of the query domain of $\mathcal{L}$,
\item $\mathcal{G}$ be a shadow of $\mathcal{L}$ with respect to $S$,
i.e. $\mathcal{G}$ and $\mathcal{L}$ behave identically for all
queries outside $S$,
\item further, suppose that within $S$, $\mathcal{G}$ responds with $\perp$
while (again within $S$), $\mathcal{L}$ does not respond with $\perp$.
Finally, let $\Pi_{t}$ be a measurement in the computational basis,
corresponding to the string $t$. 
\end{itemize}
Then 
\begin{align*}
\left|{\rm tr}[\Pi_{t}\mathcal{L}\circ U\circ\rho]-{\rm tr}[\Pi_{t}\mathcal{G}\circ U\circ\rho]\right| & \le B(\mathcal{L}\circ U\circ\rho,\mathcal{G}\circ U\circ\rho)\\
 & \le\sqrt{2\Pr[\text{find }:U^{\mathcal{L}\backslash S},\rho]}.
\end{align*}
If $\mathcal{L}$ and $S$ are random variables with a joint distribution,
we take the expectation over them in the RHS (see \Defref{prFind}). 
\end{lem}

\branchcolor{black}{\begin{proof}
We begin by assuming that $\mathcal{L}$ and $S$ are fixed (and so
is $\mathcal{G}$). In that case, we can assume $\rho$ is pure. If
not, we can purify it and absorb it in the work register. (The general
case should follow from concavity). From \Remref{psiphi0phi1}, we
have 
\begin{align*}
\left|\psi_{L}\right\rangle  & :=\mathcal{L}U\left|\psi\right\rangle _{Q'}\overset{\prettyref{rem:psiphi0phi1}}{=}\left|\phi_{0}\right\rangle _{Q'}+\left|\phi_{1}\right\rangle _{Q'}.\\
\mathcal{L}U_{S}U\left|\psi\right\rangle _{Q'}\left|0\right\rangle _{B} & =\left|\phi_{0}\right\rangle _{Q'}\left|0\right\rangle _{B}+\left|\phi_{1}\right\rangle _{Q'}\left|1\right\rangle _{B}
\end{align*}
where $Q'$ is a shorthand for $QRW$. Similarly let 
\begin{align*}
\left|\psi_{G}\right\rangle  & :=\mathcal{G}U\left|\psi\right\rangle _{Q'}=\left|\phi_{0}\right\rangle _{Q'}+\left|\phi_{1}^{\perp}\right\rangle _{Q'}
\end{align*}
where note that 
\begin{equation}
\left\langle \phi_{1}|\phi_{1}^{\perp}\right\rangle _{QRW}=0\label{eq:phisPerp}
\end{equation}
 because $\left|\phi_{1}\right\rangle $ and $\left|\phi_{1}^{\perp}\right\rangle $
are the states where the queries were made on $S$, and on $S$ $\mathcal{G}$
responds with $\perp$ while $\mathcal{L}$ does not. Further, we
analogously have 
\[
\mathcal{G}U_{S}U\left|\psi\right\rangle _{Q'}\left|0\right\rangle _{B}=\left|\phi_{0}\right\rangle _{Q'}\left|0\right\rangle _{B}+\left|\phi_{1}^{\perp}\right\rangle _{Q'}\left|1\right\rangle _{B}.
\]
 We show that the difference between $\left|\psi_{L}\right\rangle $
and $\left|\psi_{G}\right\rangle $ is bounded by $P_{{\rm find}}(\mathcal{L},S):=\Pr[{\rm find}:U^{\mathcal{L}\backslash S},\rho]$,
which in turn can be used to bound the quantity in the statement of
the lemma.
\begin{align*}
\left\Vert \left|\psi_{L}\right\rangle -\left|\psi_{G}\right\rangle \right\Vert ^{2} & =\left\Vert \left|\phi_{1}\right\rangle -\left|\phi_{1}^{\perp}\right\rangle \right\Vert ^{2}\\
 & \overset{\prettyref{eq:phisPerp}}{=}\left\Vert \left|\phi_{1}\right\rangle \right\Vert ^{2}+\left\Vert \left|\phi_{1}^{\perp}\right\rangle \right\Vert ^{2}\\
 & =2\left\Vert \left|\phi_{1}\right\rangle \right\Vert ^{2} & \because\left\Vert \left|\phi_{1}\right\rangle \right\Vert ^{2}=\left\Vert \left|\phi_{1}^{\perp}\right\rangle \right\Vert ^{2}=1-\left\Vert \left|\phi_{0}\right\rangle \right\Vert ^{2}\\
 & =2P_{{\rm find}}(\mathcal{L},S).
\end{align*}
If $\mathcal{L}$ and $S$ are random variables drawn from a (possibly)
joint distribution $\Pr(\mathcal{L},S)$, the analysis can be generalised
as follows. Let 
\begin{align*}
\rho_{L} & :=\sum_{\mathcal{L},S}\Pr(\mathcal{L},S)\left|\psi_{L}\right\rangle \left\langle \psi_{L}\right|\\
\rho_{G} & :=\sum_{\mathcal{L},S}\Pr(\mathcal{L},S)\left|\psi_{G}\right\rangle \left\langle \psi_{G}\right|
\end{align*}
where $\left|\psi_{G}\right\rangle $ is fixed by $\mathcal{L}$ and
$S$ because $G$ itself is fixed once $\mathcal{L}$ and $S$ is
fixed (by assumption). One can then use  monotonicity of fidelity
to obtain 
\begin{align*}
F(\rho_{L},\rho_{G}) & \ge\sum_{L,S}\Pr(\mathcal{L},S)F(\left|\psi_{L}\right\rangle ,\left|\psi_{G}\right\rangle )\\
 & \ge1-\frac{1}{2}.\sum_{L,S}\Pr(\mathcal{L},S)\left\Vert \left|\psi_{L}\right\rangle -\left|\psi_{G}\right\rangle \right\Vert ^{2} & \because1-\frac{1}{2}F(\left|a\right\rangle ,\left|b\right\rangle )\ge\left\Vert \left|a\right\rangle -\left|b\right\rangle \right\Vert ^{2}\\
 & \ge1-\cancel{\frac{1}{2}}\sum_{L,S}\Pr(\mathcal{L},S)\cancel{2}P_{{\rm find}}(\mathcal{L},S)\\
 & =1-P_{{\rm find}}
\end{align*}
 where $P_{{\rm find}}$ is the expectation of $P_{{\rm find}}(\mathcal{L},S)$
over $\mathcal{L}$ and $S$. It is known that the trace distance
bounds the LHS of the Lemma and the trace distance itself is bounded
by $\sqrt{2-2F}$.

\end{proof}
}

In the following, we resume the use of boldface for the query and
response registers as they do play an active role in the discussion.
\begin{lem}[{\parencite{CCL2020,ambainis_quantum_2018} Bounding $\Pr[{\rm find}:U^{\mathcal{L}\backslash S},\rho]$}]
 \label{lem:boundPfind}Suppose $S$ is a random variable and $\Pr[x\in S]\le p$
for some $p$. Further, assume that $U$ and $\rho$ are uncorrelated\footnote{i.e. the distribution from which $S$ is sampled is uncorrelated to
the distribution from which $U$ and $\rho$ are sampled,} to $S$. Then, (see \Defref{ULS}) 
\[
\Pr[{\rm find}:U^{\mathcal{L}\backslash S},\rho]\le\bar{q}\cdot p
\]
 where $\bar{q}$ is the total number of queries $U$ makes to $\mathcal{L}$.
\end{lem}

\branchcolor{black}{\begin{proof}
Let us begin with the case where the oracle is applied only once,
i.e. $\boldsymbol{Q}$ is a single query register $Q$. Since the
$RW$ registers don't play any significant role, we denote it by $L$.
Let 
\begin{align*}
U\left|\psi\right\rangle  & =\sum_{q,l}\psi(q,l)\left|q,l,0\right\rangle _{QLS}\\
\implies U_{S}U\left|\psi\right\rangle  & =\sum_{q\notin S}\left(\sum_{r,l}\psi(q,l)\left|q,l\right\rangle _{QL}\right)\left|0\right\rangle _{S}+\sum_{q\in S}\left(\sum_{r,l}\psi(q,l)\left|q,l\right\rangle _{QL}\right)\left|1\right\rangle _{S}.
\end{align*}
Since $\mathcal{L}$ leaves registers $QS$ unchanged, 
\begin{align*}
\tr[\mathbb{I}_{QL}\otimes\left|1\right\rangle \left\langle 1\right|_{B}\left(\mathcal{L}\circ U_{S}\circ U\circ\left|\psi\right\rangle \left\langle \psi\right|\right)] & =\tr[\mathbb{I}_{QL}\otimes\left|1\right\rangle \left\langle 1\right|_{B}\left(U_{S}\circ U\circ\left|\psi\right\rangle \left\langle \psi\right|\right)]\\
 & =\sum_{q}\psi^{2}(q)\chi_{S}(q)
\end{align*}
where $\psi(q)=\sum_{l}\psi(q,l)$ and $\chi_{S}$ is the characteristic
function for $S$, i.e. 
\[
\chi_{S}(q)=\begin{cases}
1 & q\in S\\
0 & q\notin S.
\end{cases}
\]
We are yet to average over the random variable $S$. Clearly, $\mathbb{E}(\chi_{S}(q))=\Pr[x\in S]\le p$,
yielding 
\[
\Pr[{\rm find}:U^{\mathcal{L}\backslash S},\rho]\le p.
\]
In the general case, everything goes through unchanged except the
string $q$ is now a set of strings $\boldsymbol{q}$ and 
\[
\chi_{S}(\boldsymbol{q})=\begin{cases}
1 & \boldsymbol{q}\cap S\neq\emptyset\\
0 & \boldsymbol{q}\cap S=\emptyset.
\end{cases}
\]
Consequently, one evaluates $\mathbb{E}(\chi_{S}(\boldsymbol{q}))=\Pr[\boldsymbol{q}\cap S\neq\emptyset]\le\left|\boldsymbol{q}\right|\cdot p=\bar{q}\cdot p$,
by the union bound, yielding 
\[
\Pr[{\rm find}:U^{\mathcal{L}\backslash S},\rho]\le\bar{q}\cdot p.
\]
\end{proof}
}

We now generalise the statement slightly to facilitate the use of
conditional random variables. These become useful for the proof of
hardness for ${\rm QC}_{d}$ circuits.
\begin{cor}
\label{cor:Conditionals}Let $D$ be the query domain of $\mathcal{L}$.
Suppose a set $S\subseteq D$, a quantum state $\rho$ and a unitary
$U$ are drawn from a joint distribution (which may be correlated
with $\mathcal{L}$). Let the set $T\subseteq D$ be another random
variable (again, possibly arbitrarily correlated) and $F$ be the
event that\footnote{One could take $F$ be to a general event as well but for our purposes,
this suffices.} $S\cap T=\emptyset$. Define the random variables $\mathcal{N}:=\mathcal{L}|F$,
$R:=S|F$, $\sigma:=\rho|F$ and $V:=U|F$ and assume that $R$, $\sigma$
and $V$ are uncorrelated.\footnote{i.e. the joint probability distribution of $S|T$, $\rho|T$ and $U|T$
is a product of their individual probability distributions.} Suppose for all $x\in D$, $\Pr[x\in R]\le p$ for some $p.$ Then,
(see \Defref{ULS}) 
\[
\Pr[{\rm find}:V^{\mathcal{\mathcal{N}}\backslash R},\sigma]\le\bar{q}\cdot p
\]
 where $\bar{q}$ is the total number of queries $V$ makes to $\mathcal{L}$. 
\end{cor}

\section{$d$-Serial Simon's Problem « Main Result 1\label{sec:-Serial-Simon's-Problem}}

\branchcolor{black}{In this section, we introduce the $d$-Serial Simon's Problem. To
get some familiarity, we first state the upper bounds---we see that
it can be solved by a ${\rm QC}_{d+1}$ circuit and also by a ${\rm CQ}_{1}$
circuit. As for the lower bound, we begin by showing that no ${\rm QNC}_{d}$
circuit can solve the problem. We then extend this proof to show that
no ${\rm QC}_{d}$ circuit can solve the problem either. }

\subsection{Oracles and Distributions | Simon's and $c$-Serial\label{subsec:Oracles-and-Distributions}}

\branchcolor{black}{We begin with Simon's problem. As will be the case for all problems
we discuss, we consider both search and decision variants. The latter
will usually take the form of distinguishing, say, a Simon's oracle
from a one-to-one oracle. We thus, first define the two oracles, and
then use them to define both variants of the problem.}
\begin{defn}[Distribution for one-to-one functions]
 Let $F:=\left\{ f:\{0,1\}^{n}\to\{0,1\}^{n}|f\text{ is one-to-one}\right\} $
be the set of all one-to-one functions acting on $n$-bit strings,
and let $\mathcal{O}_{F}:=\{\mathcal{O}_{f}:f\in F\}$ be the set
of all oracles associated with these functions. Define the \emph{distribution
over one-to-one functions}, $\mathbb{F}_{R}(n)$, to be the uniform
distribution over $F$ and the \emph{distribution for one-to-one function
oracles}, $\mathbb{O}_{{\rm R}}(n)$, to be the uniform distribution
over $\mathcal{O}_{F}$.\label{def:RandomOneOneFunctions}
\end{defn}

\begin{defn}[Distribution for Simon's function]
\label{def:SimonsFunctionDistr} Let 
\[
F:=\{(f,s)|f:\{0,1\}^{n}\to\{0,1\}^{n}\text{ is two-to-one and }f(x\oplus s)=f(x)\}
\]
be the set of all Simon functions acting on $n$-bit strings and let
$\mathcal{O}_{F}:=\{(\mathcal{O}_{f},s):(f,s)\in F\}$ be the set
of all oracles associated with these functions. Define the \emph{distribution
over Simon's functions}, $\mathbb{F}_{S}(n)$, to be the uniform distribution
over $F$ and the \emph{distribution over Simon's Oracles}, $\mathbb{O}_{S}(n)$
to be the uniform distribution over $\mathcal{O}_{F}$. 
\end{defn}

\begin{defn}[Simon's Problem]
 Fix an integer $n>0$ to denote the problem size. Let $\mathcal{R}\sim\mathbb{O}_{R}(n)$
be an oracle for a randomly chosen one-to-one function and $(\mathcal{S},s)\sim\mathbb{O}_{S}(n)$
be an oracle for a randomly chosen Simon's function which has period
$s$. 
\begin{itemize}
\item Search version: Given $\mathcal{S}$, find $s$.
\item Decision version: Given $\mathcal{O}\in_{R}\{\mathcal{R},\mathcal{S}\}$,
i.e. one of the oracles with equal probability, determine which oracle
was given.
\end{itemize}
\end{defn}

\branchcolor{black}{While we focus on $d$-Serial Simon's Problem, we define the problem
more generally as a $d$-Serial Generic Oracle Problem with respect
a ``generic oracle problem''. To this end, we briefly formalise
the latter.}
\begin{defn}[Generic Oracle Problem]
 Let $(\mathcal{O},r)\sim\mathbb{O}(n)$ where $\mathbb{O}$ is some
fixed distribution over oracles and the corresponding expected answers,
and $n$ is the problem size. Let $\mathcal{R}\sim\cancel{\mathbb{O}}(n)$
where suppose that $\cancel{\mathbb{O}}(n)$ is the distribution over
oracles against which the decision version is defined. The generic
oracle problem is:\label{def:genericOracleProblem}
\begin{itemize}
\item Search version: Given $\mathcal{O},$ find $r$.
\item Decision version: Given $\mathcal{Q}\in_{R}\{\mathcal{O},\mathcal{R}\}$,
determine which oracle was given.
\end{itemize}
\end{defn}

\branchcolor{black}{The $d$-Serial Generic Oracle problem is based on the following idea:
there is a sequence of $d+1$ oracles (indexed $0,1\dots d$), of
which the first $d$ encode a Simon's problems. The zeroth oracle
can be accessed directly. To access the first oracle, however, a secret
key is needed. This secret is the period of the zeroth Simon's function.
The first oracle, once unlocked, behaves as a Simon's oracle whose
period unlocks the second oracle and so on. The Generic oracle is
unlocked by the period of the $(d-1)$th Simon's problem. The problem
is to solve the ``Generic Oracle problem'' using the aforementioned
$d+1$ oracles.

The intuition is quite simple. Consider the $d$-Serial Simon's Oracle
Problem, i.e. where the generic oracle problem is a Simon's problem.
Then, observe that, naïvely, a ${\rm QC}_{d}$ scheme would have to
use all its $d$ depth to solve the $d$ Simon's problems to access
the last oracle and have no quantum depth left for solving the last
Simon's problem. However, with one more depth, i.e. with ${\rm QC}_{d+1}$,
the problem can be solved. Further, note that a ${\rm CQ}_{1}$ scheme
too can solve the problem. We will revisit these statements shortly
but first, we formally define the problem.}
\begin{defn}[$c$-Serial Generic Oracle]
 \label{def:cSerialSimonsOracle}Suppose the Generic Oracle is sampled
from the distributions $\mathbb{O}(n)$ and $\cancel{\mathbb{O}}(n)$
as in \Defref{genericOracleProblem} above where the oracles' domain
is assumed to be $\{0,1\}^{n}$. We define the \emph{$c$-Serial Generic
Function distribution} $\mathbb{F}_{{\rm Serial}}(c,n,\mathbb{O},\cancel{\mathbb{O}})$
and the \emph{$c$-Serial Generic Oracle distribution} $\mathbb{O}_{{\rm Serial}}(c,n,\mathbb{O},\cancel{\mathbb{O}})$
by specifying its sampling procedure. 
\begin{itemize}
\item Sampling step:
\begin{itemize}
\item For $i\in\{0,1,2,\dots c-1\}$, $(f_{i},s_{i})\sim\mathbb{F}_{S}(n)$,
i.e. sample $c$ Simon's functions from $\mathbb{F}_{S}(n)$.
\item Sample $(\mathcal{O},r)\sim\mathbb{O}(n)$ and $\mathcal{R}\sim\cancel{\mathbb{O}}(n)$.
For the search version, $\mathcal{Q}:=\mathcal{O}$ while for the
decision version, $\mathcal{Q}\in_{R}\{\mathcal{O},\mathcal{R}\}$.
\end{itemize}
\item Let $L_{f_{i}}:\{0,1\}^{n}\times\{0,1\}^{n}\to\{0,1\}^{n}\cup\{\perp\}$
for $i\in\{0,1,2,\dots c\}$ be defined as 
\begin{itemize}
\item $L_{f_{0}}(x,z)=f_{0}(x)$ when $i=0$ and 
\item $L_{f_{i}}(x,z):=\begin{cases}
f_{i}(x) & z=s_{i-1}\\
\perp & z\neq s_{i-1}\text{ when }i\in\{1,2,3\dots c-1\}.
\end{cases}$
\item $L_{f_{c}}(x,z):=\begin{cases}
\mathcal{Q}(x) & z=s_{c-1}\\
\perp & z\neq s_{c-1}.
\end{cases}$
\end{itemize}
\item Let $\mathcal{L}$ be the oracle associated with $(L_{f_{i}})_{i=0}^{c}$.
\item Returns
\begin{itemize}
\item $\mathbb{O}_{{\rm Serial}}$
\begin{itemize}
\item Search: When $\mathbb{O}_{{\rm Serial}}(c,n,\mathbb{O})$ is sampled,
consider the search version of the sampling step above and return
$(\mathcal{L},r)$. 
\item Decision: When $\mathbb{O}_{{\rm Serial}}(c,n,\mathbb{O},\cancel{\mathbb{O}})$
is sampled, consider the decision version of the sampling step above
and return $(\mathcal{L},l)$ where $l=0$ if $\mathcal{Q}=\mathcal{O}$
and $l=1$ if $\mathcal{Q}=\mathcal{R}$. 
\end{itemize}
\item $\mathbb{F}_{{\rm Serial}}$
\begin{itemize}
\item When $\mathbb{F}_{{\rm Serial}}(c,n,\mathbb{O})$ is sampled, consider
the search version of the sampling step above and return $(L_{f_{0}},s_{0},L_{f_{1}},s_{1}\dots L_{f_{c-1}},s_{c-1},L_{f_{c}},r)$,
where recall that $(\mathcal{O},r)$ were sampled from $\mathbb{O}(n)$
and $\mathcal{Q}=\mathcal{O}$ (for the search version of the sampling
step).
\end{itemize}
\end{itemize}
\end{itemize}
\end{defn}

\begin{defn}[$c$--Serial Generic Oracle Problem]
 \label{def:cSerialSimonsProblem}As before, we define two variants
of the $c$-Serial Generic Oracle Problem:
\begin{itemize}
\item Search Variant: Let $(\mathcal{L},r)\sim\mathbb{O}_{{\rm Serial}}(c,n,\mathbb{O})$.
Given $\mathcal{L}$, find $r$. 
\item Decision Variant: Let $(\mathcal{L},l)\sim\mathbb{O}_{{\rm Serial}}(c,n,\mathbb{O},\cancel{\mathbb{O}})$.
Given $\mathcal{L}$, determine whether $l=0$ or $1$.
\end{itemize}
\end{defn}

\subsection{Depth upper bounds for $d$-Serial Simon's Problem}

\branchcolor{black}{One can use Simon's algorithm for easily obtaining the following depth
upper bounds for solving $d$-Serial Simon's Problem. Note that while
we only need $d+1$ queries to the oracle, we need to apply Hadamards
before and after, for Simon's algorithm to work. Thus, for ${\rm QC}$
circuits (which by definition allow only one layer of unitary before
an oracle call, not after), we need depth $2d+2$ while for ${\rm CQ}$
circuits (which allow a layer of unitary before and after), we only
need depth $1$.}
\begin{prop}
The decision variant of the $d$-Serial Simon's Problem is in ${\rm BQNC}_{2d+2}^{{\rm BPP}}$
and the search version can be solved using ${\rm QC}_{2d+2}$ circuits.
\end{prop}

\begin{prop}
The decision variant of the $d$-Serial Simon's Problem is in ${\rm BPP}^{{\rm BQNC}_{1}}$
and the search version can be solved using ${\rm CQ}_{1}$ circuits.
\end{prop}

\subsection{Depth lower bounds for $c$-Serial Simon's problem}

\branchcolor{black}{Following \textcite{CCL2020}, it would be useful to define ``shadows''
of oracles for establishing lower bounds. The idea is simple. Given
an oracle and a subset $S$ of the query domain thereof, the shadow
behaves exactly like the oracle when queried outside $S$ and outputs
$\perp$ when queried inside $S$. Since there are multiple functions
involved, for notational ease, we formalise this notion for a sequence
of sets.}
\begin{defn}[$c$--Serial Generic Shadow Oracle]
 \label{def:cSerialShadowOracle}Given 
\begin{itemize}
\item $(L_{f_{i}},s_{i})_{i=0}^{c}$ from the sample space of $\mathbb{F}_{{\rm Serial}}(c,n,\mathbb{O})$,
and 
\item a tuple of subsets $\bar{S}=(S_{1},S_{2},\dots S_{c})$ where $S_{i}\subseteq\{0,1\}^{n}\times\{0,1\}^{n}$,
\end{itemize}
let $\mathcal{L}$ be the oracle associated with $(L_{f_{i}})_{i=1}^{c}$.
Then, the \emph{$c$-Serial Generic Shadow Oracle} (or simply the
shadow) $\mathcal{H}$ for the oracle $\mathcal{L}$ with respect
to $\bar{S}$ is defined to be the oracle associated with $(L_{f_{0}},L_{f_{1}}',L_{f_{2}}'\dots L_{f_{c}}')$
where 
\[
L'_{f_{i}}:=\begin{cases}
L_{f_{i}}(x,z) & (x,z)\in\{0,1\}^{2n}\backslash S_{i}\\
\perp & (x,z)\in S_{i}
\end{cases}
\]
for $i\in\{1,2\dots c\}$.
\end{defn}

\begin{rem}
For the moment, we exclusively consider the case where the ``Generic
Oracle'' (drawn from $\mathbb{O}$) is also a Simon's oracle (which
are drawn from $\mathbb{O}_{S}$) and define the \emph{$c$-Serial
Simon's Oracle/Problem/Shadow Oracle} implicitly.
\end{rem}

\branchcolor{black}{As will become evident shortly, it would be useful to have shadows
$\{\mathcal{M}_{i}\}_{i=1}^{d}$ such that $\mathcal{M}_{j}$ behaves
like the $d$-Serial Simon's oracle $\mathcal{L}$ at all sub-oracles
from $0$ to $j-1$ but outputs $\perp$ for all subsequent sub-oracles
(see \Figref{QNCd}).}
\begin{lyxalgorithm}[$\bar{S}_{j}$ for ${\rm QNC}_{d}$ exclusion]
\label{alg:SforQNCd}~\\
Input: 
\begin{itemize}
\item $1\le j\le d$ and
\item $(L_{f_{i}},s_{i})_{i=0}^{d}$ from the sample space of $\mathbb{F}_{{\rm Serial}}(c,n,\mathbb{O}_{S})$, 
\end{itemize}
Output: $\bar{S}_{j}$, a tuple of $d$ subsets defined as
\[
\bar{S}_{j}:=\begin{cases}
(\emptyset,\emptyset,\dots\emptyset,E\times s_{j-1},E\times s_{j}\dots,E\times s_{d-1}) & \text{for }j>1\\
(E\times s_{0},E\times s_{1}\dots,E\times s_{d-1}) & \text{for }j=1
\end{cases}
\]
where $E=\{0,1\}^{n}$ and $\times$ is the Cartesian product.
\end{lyxalgorithm}

\subsubsection{Warm up | $d$-Serial Simon's Problem is hard for ${\rm QNC}_{d}$
\label{subsec:QNCd}}

\begin{thm}
\label{thm:QNCd-dSerialSimons}Let $(\mathcal{L},s)\sim\mathbb{O}_{{\rm Serial}}(d,n,\mathbb{O}_{S})$,
i.e. let $\mathcal{L}$ be an oracle for a random $d$-Serial Simon's
Problem of size $n$ and period $s$. Let $\mathcal{A}^{\mathcal{L}}$
be any $d$ depth quantum circuit (see \Defref{QNCd} and \Remref{oracleVersionsQNC-CQ-QC})
acting on $\mathcal{O}(n)$ qubits, with query access to $\mathcal{L}$.
Then $\Pr[s\leftarrow\mathcal{A}^{\mathcal{L}}]\le\negl$, i.e. the
probability that the algorithm finds the period is exponentially small. 
\end{thm}

\branchcolor{black}{\begin{proof}
Suppose\footnote{This proof closely follows techniques from \textcite{CCL2020} but
establishes a new result.} $(L_{f_{i}},s_{i})_{i=0}^{d}\sim\mathbb{F}_{{\rm Serial}}(d,n,\mathbb{O}_{S})$
(see \Defref{cSerialSimonsOracle}) and let $\mathcal{L}$ be the
oracle associated with $(L_{f_{i}})_{i=0}^{d}$. For notational consistency,
let $s=s_{d}$. Denote an arbitrary ${\rm QNC}_{d}^{\mathcal{L}}$
circuit, $\mathcal{A}^{\mathcal{L}}$, by 
\[
\mathcal{A}^{\mathcal{L}}:=\Pi\circ U_{d+1}\circ\mathcal{L}\circ U_{d}\dots\mathcal{L}\circ U_{2}\circ\mathcal{L}\circ U_{1}
\]
and suppose $\Pi$ corresponds to the algorithm outputting the string
$s$. For each $i\in\{1,\dots d\}$, construct the tuples $\bar{S}_{i}$
using \Algref{SforQNCd}. Let $\mathcal{M}_{i}$ be the shadow of
$\mathcal{L}$ with respect to $\bar{S}_{i}$ (see \Defref{cSerialShadowOracle}).
Define 
\[
\mathcal{A}^{\mathcal{M}}:=\Pi\circ U_{d+1}\circ\mathcal{M}_{d}\circ U_{d}\dots\mathcal{M}_{2}\circ U_{2}\circ\mathcal{M}_{1}\circ U_{1}.
\]
Note that $\Pr[s\leftarrow\mathcal{A}^{\mathcal{M}}]\le\frac{1}{2^{n}}$
because no $\mathcal{M}_{i}$ contains any information about $f_{d}$,
the last Simon's function, whose period, $s$, is the required solution.
Thus, no algorithm can do better than making a random guess. We now
show that the output distributions of $\mathcal{A}^{\mathcal{L}}$
and $\mathcal{A}^{\mathcal{M}}$ cannot be noticeably different using
the O2H lemma (see \Lemref{O2H}). 

To apply the lemma, one can use the hybrid method as follows (we drop
the $\circ$ symbol for brevity): 
\begin{align*}
 & \left|\Pr[s\leftarrow\mathcal{A}^{\mathcal{L}}]-\Pr[s\leftarrow\mathcal{A}^{\mathcal{M}}]\right|\\
= & \left|\tr[\Pi U_{d+1}\mathcal{L}U_{d}\dots\mathcal{L}U_{2}\mathcal{L}U_{1}\rho_{0}-\Pi U_{d+1}\mathcal{M}_{d}U_{d}\dots\mathcal{M}_{2}U_{2}\mathcal{M}_{1}U_{1}\rho_{0}]\right|\\
\le & \left|\tr[\Pi U_{d+1}\mathcal{L}U_{d}\dots\mathcal{L}U_{2}\underbrace{\mathcal{L}U_{1}\rho_{0}}-\Pi U_{d+1}\mathcal{L}U_{d}\dots\mathcal{L}U_{2}\underbrace{\mathcal{M}_{1}U_{1}\rho_{0}}]\right|+\\
 & \left|\tr[\Pi U_{d+1}\mathcal{L}U_{d}\dots U_{3}\underbrace{\mathcal{L}U_{2}\mathcal{M}_{1}U_{1}\rho_{0}}-\Pi U_{d+1}\mathcal{L}U_{d}\dots\mathcal{L}U_{3}\underbrace{\mathcal{M}_{2}U_{2}\mathcal{M}_{1}U_{1}\rho_{0}}]\right|+\\
 & \vdots\\
 & \left|\tr[\Pi U_{d+1}\underbrace{\mathcal{L}U_{d}\mathcal{M}_{d-1}U_{d-1}\dots U_{3}\mathcal{M}_{2}U_{2}\mathcal{M}_{1}U_{1}\rho_{0}}-\Pi U_{d+1}\underbrace{\mathcal{M}_{d}U_{d}\mathcal{M}_{d-1}\dots U_{3}\mathcal{M}_{2}U_{2}\mathcal{M}_{1}U_{1}\rho_{0}}]\right|\\
\le & {\rm B}(\mathcal{L}\circ U_{1}(\rho_{0}),\mathcal{M}_{1}\circ U_{1}(\rho_{0}))+\\
 & {\rm B}(\mathcal{L}\circ U_{2}(\rho_{1}),\mathcal{M}_{2}\circ U_{2}(\rho_{1}))+\\
 & \vdots\\
 & {\rm B}(\mathcal{L}\circ U_{d}(\rho_{d-1}),\mathcal{M}_{d}\circ U_{d}(\rho_{d-1}))\\
\le & \sum_{i=1}^{d}\sqrt{2\Pr[{\rm find}:U_{i}^{\mathcal{L}\backslash\bar{S}_{i}},\rho_{i-1}]}
\end{align*}
where $\rho_{0}=\left|0\dots0\right\rangle \left\langle 0\dots0\right|$
and $\rho_{i}=\mathcal{M}_{i}\circ U_{i}\circ\dots\mathcal{M}_{1}\circ U_{1}(\rho_{0})$
for $i>0$. To bound the last expression, we apply \Lemref{boundPfind}.
To apply the lemma, however, we must ensure that the subset of queries
at which $\mathcal{L}$ and $\mathcal{M}_{i}$ differ, i.e. $\bar{S}_{i}=(\emptyset,\dots\emptyset,E\times s_{i-1},E\times s_{i}\dots E\times s_{d-1})$,
(recall $E=\{0,1\}^{n}$) is uncorrelated to $U_{i}$ and $\rho_{i-1}$.
Observe that $\rho_{i-1}$ is completely uncorrelated to $s_{i-1},s_{i},s_{i+1}\dots s_{d}$.
At a high level, this is because $\rho_{i-1}$ can at most access
$\mathcal{M}_{1},\dots\mathcal{M}_{i-1}$ and these in turn contain
no information about $\bar{S}_{i}$. To see this, note that even though
to define $\mathcal{M}_{1}$, we used $\bar{S}_{1}=(E\times s_{0},E\times s_{1}\dots,E\times s_{d-1})$,
still $\mathcal{M}_{1}$ contains no information about $s_{1},\dots s_{d}$
because (other than the zeroth sub-oracle which can reveal $s_{0}$)
it always outputs $\perp$. Similarly, $\mathcal{M}_{2}$ contains
information about $s_{0},s_{1}$ but not about $s_{2},s_{3}\dots s_{d}$.
Analogously for $\mathcal{M}_{3}$ and so on. Since the definition
of $\rho_{i-1}$ only involves $\mathcal{M}_{1},\dots\mathcal{M}_{i-1}$,
it contains no information about $s_{i-1},s_{i},s_{i+1}\dots s_{d}$.
As for $U_{i}$, that is uncorrelated to all $s_{i}$s by construction. 

Finally, to apply \Lemref{boundPfind}, we need to bound the probability
that a fixed query, $x$, lands in the set $\bar{S}_{i}$. To this
end, observe that $\Pr[x_{i}\in E\times s_{i}]=\frac{1}{2^{n}}$ when
$s_{i}\in_{R}\{0,1\}^{n}$ and that the union bound readily bounds
the desired probability, i.e. $\Pr[x\in\bar{S}_{i}]\le d\cdot2^{-n}$.
Since $U$ acts on $\poly$ many qubits, $q$ in the lemma can be
set to $\poly$. Thus, we can bound the last inequality by $d\cdot\poly/2^{n}$.
Using the triangle inequality, we get 
\[
\Pr[s\leftarrow\mathcal{A}^{\mathcal{L}}]\le\frac{\poly}{2^{n}}.
\]
\end{proof}
}

\subsubsection{$d$-Serial Simon's Problem is hard for ${\rm QC}_{d}$ }

\branchcolor{black}{To prove our first main result---the same statement for ${\rm QC}_{d}$---we
need to account for the possibility that the classical algorithm can
make ${\rm poly}(n)$ many queries and process these before applying
the next quantum layer. The high-level intuition is quite simple.
We follow essentially the same strategy as in the ${\rm QNC}_{d}$
case, except that we successively condition the distribution over
$\mathcal{L}$ to exclude the cases where the classical algorithm
obtains a non-$\perp$ output. 

To be slightly more precise, fix a particular $\mathcal{L}$. When
the classical algorithm queries locations $T$, it could either get
all $\perp$ responses or not get all $\perp$ responses (that is
some responses may be non-$\perp$). We treat the latter case as though
the classical algorithm solved the problem. For the former case, we
conclude that the classical algorithm ruled out certain values of
$s_{i}$s. We condition on this event and proceed similarly with the
remaining analysis. 

Since $\mathcal{L}$ is actually a random variable, notice that the
probability that $\mathcal{L}$ responds with non-$\perp$ for a given
$T$, is at most $\mathcal{O}(\poly\cdot2^{-n})$. Thus, the probability
that $\mathcal{L}$ responds with $\perp$ for all queries in $T$
is essentially $1$. Since we want an upper bound on the winning probability,
we treat this conditional probability as $1$ and in the subsequent
analysis, use $\mathcal{L}|T$ where $T$ is s.t. $\mathcal{L}(T)=\perp$.
Since the shadows were defined using $\mathcal{L}$, they also get
conditioned and the remaining analysis, essentially goes through unchanged.
The only difference is that when $\Pr[{\rm find}:U^{\mathcal{L}\backslash S},\rho]$
is evaluated, because of the conditioning, the probabilities change
by polynomial factors but these we have anyway been absorbing so the
result remains unchanged.}

\begin{thm}
\label{thm:QCd_hardness_dSerialSimons}Let $(\mathcal{L},s)\sim\mathbb{O}_{{\rm Serial}}(d,n,\mathbb{O}_{S})$,
i.e. let $\mathcal{L}$ be an oracle for a random $d$-Serial Simon's
Problem of size $n$ and period $s$. Let $\mathcal{B}^{\mathcal{L}}$
be any ${\rm QC}_{d}$ circuit (see \Defref{dQC} and \Remref{oracleVersionsQNC-CQ-QC})
with query access to $\mathcal{L}$. Then, $\Pr[s\leftarrow\mathcal{B}^{\mathcal{L}}]\le\negl$,
i.e. the probability that the algorithm finds the period is exponentially
small. 
\end{thm}

\branchcolor{black}{\begin{proof}
The initial part of the proof is almost identical to that of the ${\rm QNC}_{d}$
case. Suppose $(L_{f_{i}},s_{i})_{i=0}^{d}\sim\mathbb{F}_{{\rm Serial}}(d,n,\mathbb{O}_{S})$
(see \Defref{cSerialSimonsOracle}) and let $\mathcal{L}$ be the
oracle associated with $(L_{f_{i}})_{i=0}^{d}$. For notational consistency,
let $s=s_{d}$. Recall that we denoted an arbitrary ${\rm QC}_{d}^{\mathcal{L}}$
circuit (see \Notaref{CompositionNotation}) with oracle access to
$\mathcal{L}$, 
\[
\mathcal{B}^{\mathcal{L}}:=\Pi\circ\mathcal{A}_{c,d+1}^{\mathcal{L}}\circ\mathcal{B}_{d}^{\mathcal{L}}\circ\mathcal{B}_{d-1}^{\mathcal{L}}\dots\circ\mathcal{B}_{1}^{\mathcal{L}}\circ\rho_{0}
\]
where $\mathcal{B}_{i}^{\mathcal{L}}:=\Pi_{i}\circ\mathcal{L}\circ U_{i}\circ\mathcal{A}_{c,i}^{\mathcal{L}}$,
$\rho_{0}=\left|0\dots0\right\rangle \left\langle 0\dots0\right|$
and $\Pi$ corresponds to the algorithm outputting $s$. For each
$i\in\{1,\dots d\}$, construct the tuples $\bar{S}_{i}$ using \Algref{SforQNCd}.
Let $\mathcal{M}_{i}$ be the shadow of $\mathcal{L}$ with respect
to $\bar{S}_{i}$ (see \Defref{cSerialShadowOracle}). Define 
\[
\mathcal{B}^{\mathcal{M}}:=\Pi\circ\mathcal{A}_{c,d+1}^{\mathcal{L}}\circ\mathcal{B}_{d}^{\mathcal{M}}\circ\mathcal{B}_{d-1}^{\mathcal{M}}\dots\circ\mathcal{B}_{1}^{\mathcal{M}}.
\]
where $\mathcal{B}_{i}^{\mathcal{M}}:=\Pi_{i}\circ\mathcal{M}_{i}\circ U_{i}\circ\mathcal{A}_{c,i}^{\mathcal{L}}$.
We have 
\begin{align*}
 & \left|\Pr[s\leftarrow\mathcal{B}^{\mathcal{L}}]-\Pr[s\leftarrow\mathcal{B}^{\mathcal{M}}]\right|\\
= & \left|\tr[\Pi\mathcal{A}_{c,d+1}^{\mathcal{L}}\mathcal{B}_{d}^{\mathcal{L}}\mathcal{B}_{d-1}^{\mathcal{L}}\dots\mathcal{B}_{1}^{\mathcal{L}}\rho_{0}]-\tr[\Pi\mathcal{A}_{c,d+1}^{\mathcal{L}}\mathcal{B}_{d}^{\mathcal{M}}\mathcal{B}_{d-1}^{\mathcal{M}}\dots\mathcal{B}_{1}^{\mathcal{M}}\rho_{0}]\right| & \text{we dropped }\circ\text{ for brevity}\\
\le & {\rm B}(\mathcal{B}_{1}^{\mathcal{L}}(\rho_{0}),\mathcal{B}_{1}^{\mathcal{M}}(\rho_{0}))+\\
 & {\rm B}(\mathcal{B}_{2}^{\mathcal{L}}(\rho_{1}),\mathcal{B}_{2}^{\mathcal{M}}(\rho_{1}))+\\
 & \vdots\\
 & {\rm B}(\mathcal{B}_{d}^{\mathcal{L}}(\rho_{d-1}),\mathcal{B}_{2}^{\mathcal{M}}(\rho_{d-1}))\\
= & \sum_{i=1}^{d}\sqrt{2\Pr[{\rm find}:U_{i}^{\mathcal{L}\backslash\bar{S}_{i}},\mathcal{A}_{c,i}^{\mathcal{L}}\circ\rho_{i-1}]}.
\end{align*}
where for $i\in\{1,2\dots d-1\}$, $\rho_{i}:=\mathcal{B}_{i}^{\mathcal{M}}\circ\dots\circ\mathcal{B}_{1}^{\mathcal{M}}\circ\rho_{0}$.

So far, everything was essentially the same as in the ${\rm QNC}_{d}$
case. The difference arises because of the classical algorithm. We
begin with the first term, $\Pr[{\rm find}:U_{1}^{\mathcal{L}\backslash\bar{S}_{1}},\mathcal{A}_{c,1}^{\mathcal{L}}\circ\rho_{0}]$
and denote by $\bar{T}_{1}=(T_{1,1},T_{1,2}\dots T_{1,d})$ the tuple
of subsets queried by $\mathcal{A}_{c,1}$. There are two possibilities:
either all queries in $\bar{T}_{1}$ result in $\perp$ or at least
one query yields a non-$\perp$ result. We treat the second event
as though the algorithm was able to ``find'' the solution. Denote
the first event by $F_{1}:=\bar{S}_{1}\cap\bar{T}_{1}=\emptyset$
and the second event as $\neg F_{1}:=\bar{S}_{1}\cap\bar{T}_{1}\neq\emptyset$
where the intersection is component-wise. Note that the random variable
of interest here is $\mathcal{L}$ and those derived using it, i.e.
$\mathcal{M}_{i}s$ and $\bar{S}_{i}$s. While the notation is cumbersome,
we can use $\Pr[E]=\Pr[E|F]\Pr[F]+\Pr[E|\neg F]\Pr[\neg F]\le\Pr[E|F]+\Pr[\neg F]$,
to write 
\begin{equation}
\Pr[{\rm find}:U_{1}^{\mathcal{L}\backslash\bar{S}_{1}},\mathcal{A}_{c,1}^{\mathcal{L}}\circ\rho_{0}]\le\Pr[{\rm find}:V_{1}^{\mathcal{N}_{1}\backslash\bar{R}_{1}},\sigma_{0}]+\Pr[\neg F_{1}]\label{eq:QCd_analysis_first}
\end{equation}
where $\sigma_{0}:=\mathcal{A}_{c,1}^{\mathcal{L}}\circ\rho_{0}|F_{1}$,
$\bar{R}_{1}:=\bar{S}_{1}|F_{1}$, $\mathcal{N}_{1}=\mathcal{L}|F_{1}$
and $V_{1}:=U_{1}|F_{1}$. First, it is clear that $\Pr[\neg F_{1}]\le\mathcal{O}(\poly2^{-n})$
by the union bound as $\Pr[x\in E\times s_{i}]\le2^{-n}$ for all
$i\in\{1,2\dots d\}$. Second, $\mathcal{N}_{1}$ is uncorrelated
to $\sigma_{0}$, and $V_{1}$ because we have restricted the sample
space to the cases where the correlation is absent by conditioning
on $F_{1}$ (i.e. $\bar{T}_{1}$ has been effectively removed from
this part of the analysis). In more detail, the algorithm $\mathcal{A}_{c,1}$
(conditioned on event $F_{1}$), ruled out a polynomial number of
locations where the various $E\times s_{i}$s are not. In the remaining
query domain, $E\times s_{i}$s are not restricted, i.e. $\mathcal{N}_{1}$
is uncorrelated with $\sigma_{0}$. This also means that $\bar{R}_{1}$
is uncorrelated with $\sigma_{0}$. Finally, note that $\Pr[x\in\bar{R}_{1}]\le\mathcal{O}\left(\frac{1}{2^{n}-\poly}\right)\le\mathcal{O}(\poly2^{-n})$.
We can therefore apply \Corref{Conditionals} to obtain the bound
the first term in \Eqref{QCd_analysis_first} with $\mathcal{O}(\poly2^{-n})$.
This yields 
\[
\Pr[{\rm find}:U_{1}^{\mathcal{L}\backslash\bar{S}_{1}},\mathcal{A}_{c,1}^{\mathcal{L}}\circ\rho_{0}]\le\mathcal{O}(\poly2^{-n})
\]

We now apply this reasoning to the second term,\footnote{It may help to explicitly state the convention: as we condition, we
introduce more letters, leaving the indices unchanged.} 
\[
\Pr[{\rm find}:U_{2}^{\mathcal{L}\backslash\bar{S}_{2}},\mathcal{A}_{c,2}^{\mathcal{L}}\circ\underbrace{\Pi_{1}\circ\mathcal{M}_{1}\circ U_{1}\circ\mathcal{A}_{c,1}^{\mathcal{L}}\circ\rho_{0}}_{\rho_{1}}].
\]
It is clear,\footnote{It may help to recall that, by construction, $\mathcal{M}_{1}\dots\mathcal{M}_{i-1}$
contain no information about $\bar{S}_{i}$, i.e. the query domain
on which $\mathcal{L}$ and $\mathcal{M}_{i}$ differ.} that $\rho_{1}$ can contain information about $s_{0}$ (as it has
access to $\mathcal{M}_{1}$) and therefore $\bar{S}_{2}$ was chosen
so that $\mathcal{M}_{2}$ and $\mathcal{L}$ behave identically when
their zeroth sub-oracle is queried (whose ``access requires $s_{0}$
and secret is $s_{1}$'', see \Figref{ShadowsQNCd}). Additionally,
$\mathcal{A}_{c,1}^{\mathcal{L}}$ may have partially learnt something
about the remaining $s_{i}$s as well and this would prevent us from
applying \Lemref{boundPfind}. We condition as in the case above to
obtain 
\begin{equation}
\Pr[{\rm find}:U_{2}^{\mathcal{L}\backslash\bar{S}_{2}},\mathcal{A}_{c,2}^{\mathcal{L}}\circ\rho_{1}]\le\Pr[{\rm find}:V_{2}^{\mathcal{N}_{2}\backslash\bar{R}_{2}},\mathcal{A}_{c,2}^{\mathcal{N}_{2}}\circ\sigma_{1}]+\Pr[\neg F_{1}]\label{eq:QCd_2_1}
\end{equation}
where $\sigma_{1}:=\rho_{1}|F_{1}$, $\mathcal{N}_{2}:=\mathcal{L}|F_{1}$,
$V_{2}:=U_{2}|F_{1}$ and $\bar{R}_{2}:=\bar{S}_{2}|F_{1}$. The first
term restricts the sample space (of $\mathcal{L}$) such that $\sigma_{1}$
is uncorrelated with $\bar{R}_{2}$.\footnote{Note that $\sigma_{1}$ is potentially correlated with $\mathcal{N}_{2}$
(because it had access to $\mathcal{M}_{1}$ (see see \Figref{ShadowsQNCd})
which contains information about $s_{0}$ and note that $\mathcal{N}_{2}$
also has dependence on $s_{0}$; conditioning only excludes those
$\mathcal{L}$s where $\bar{T}_{1}$ is not assigned $\perp$)} We are therefore essentially in the same situation as the starting
point of the case above (see \Eqref{QCd_analysis_first}). Let $\bar{T}_{2}$
be the set where $\mathcal{A}_{c,2}^{\mathcal{N}_{2}}$ queries and
define the event $F_{2}:=\bar{R}_{2}\cap\bar{T}_{2}=\emptyset$. Conditioning
on $F_{2}$, we can bound the first term as 
\begin{equation}
\Pr[{\rm find}:V_{2}^{\mathcal{N}_{2}\backslash\bar{R}_{2}},\mathcal{A}_{c,2}^{\mathcal{N}_{2}}\circ\sigma_{1}]\le\Pr[{\rm find}:W_{2}^{\mathcal{O}_{2}\backslash\bar{Q}_{2}},\tau_{2}]+\Pr[\neg F_{2}]\label{eq:QCd_2_2}
\end{equation}
where\footnote{The index of the state, $i$, is incremented to reflect that $\mathcal{M}_{1}\dots\mathcal{M}_{i}$
have been queried by the state (and not $\mathcal{M}_{i+1}\dots\mathcal{M}_{d}$).} $\tau_{1}:=\mathcal{A}_{c,2}^{\mathcal{N}_{2}}\circ\sigma_{1}|F_{2}$,
$\mathcal{O}_{2}:=\mathcal{N}_{2}|F_{1}$, $W_{2}:=V_{2}|F_{2}$ and
$\bar{Q}_{2}:=\bar{R}_{2}|F_{2}$. Now, $\bar{Q}_{2}$ is uncorrelated
to $\tau_{1}$ because we restricted to the part of the sample space
where the effect $\mathcal{A}_{c,2}^{\mathcal{N}_{2}}$ has been accounted
for. To see this, observe that after the conditioning, the reasoning
is analogous to why $\rho_{1}$ is uncorrelated to $\bar{S}_{2}$.
We can now apply \Corref{Conditionals} with $\Pr[x\in\bar{Q}_{2}]\le\mathcal{O}(\poly2^{-n})$,
$\Pr[\neg F_{2}]\le\mathcal{O}(\poly2^{-n})$ in \Eqref{QCd_2_2}
and combine it with \Eqref{QCd_2_1} to obtain 
\[
\Pr[{\rm find}:U_{2}^{\mathcal{L}\backslash\bar{S}_{2}},\mathcal{A}_{c,2}^{\mathcal{L}}\circ\rho_{1}]\le\mathcal{O}(\poly2^{-n}).
\]
Proceeding similarly, one obtains the final bound.
\end{proof}
}

\section{Warm-up | $d$-Shuffled Simon's Problem ($d$-SS)\label{sec:warm-up-dSS}}

\branchcolor{black}{In the previous section, we established the first main result of this
work---for each $d$, we saw that $d$-Serial Simon's problem is
easy for a ${\rm CQ}_{1}$ circuit but hard for ${\rm QC}_{d}$. Our
second main goal is to show that for each $d$, there is a problem
which is easy for ${\rm QC}_{c}$ (for some constant $c>0$, wrt $d$
and $n$) but hard for ${\rm CQ}_{d}$. However, showing depth lower
bounds for ${\rm CQ}_{d}$ circuits can get quite involved. We, therefore,
first re-derive a result due to \textcite{CCL2020}---for each $d$,
there is a problem which is easy for ${\rm CQ}_{2d+1}$ but hard for
${\rm CQ}_{d}$. While we use essentially the same problem (more on
this momentarily), our proof is different and simpler. We build upon
these techniques for proving our second main result in \Secref{ShuffledCollisionsToSimons}.

The problem we consider in this section is called $d$-Shuffl\emph{ed}
Simon's Problem which is essentially the same as the $d$-Shuffl\emph{ing}
Simon's Problem introduced by \citeauthor{CCL2020}. The difference
arises in the way we construct the Shuffler and consequently in its
analysis.}

\subsection{Oracles and Distributions}

\subsubsection{$d$-Shuffler}
\begin{notation}
We identify binary strings with their associated integer values implicitly.
E.g. $f:\{0,1\}^{2n}\to\{0,1\}^{2n}$ then we may use $f(2^{n}-1)$
to denote $f(00\dots0\underbrace{11\dots1}_{n\text{ many}})$.\label{nota:IntegerString}
\end{notation}

\branchcolor{black}{We begin with defining a $d$-Shuffler for a function $f$. We give
two equivalent definitions of a $d$-Shuffler. The first is more intuitive
but the second is easier to analyse. 

A $d$-Shuffler for $f:\{0,1\}^{n}\to\{0,1\}^{n}$ is simply a sequence
of $d$ permutations $f_{0},f_{1}\dots f_{d-1}$ on a larger space,
which uses $2n$ length strings and a function $f_{d}$ such that
$f_{d}\circ f_{d-1}\dots f_{0}(x)=f(x)$ for all $x$ in the domain
of $f$ and for all points untouched by these ``paths'' originating
from the domain of $x$, the permutations are modified to output $\perp$.

This definition makes counting arguments slightly convoluted. We therefore
observe that a sequence of permutations may equivalently be specified
by a sequence of tuples.\footnote{To prevent confusion with the standard notation for permutations,
we denote these tuples by columns.} For instance two permutations over four elements may be expressed
as $[0,1,2,3]^{T}\mapsto[1,3,0,2]^{T}\mapsto[2,0,3,1]^{T}$ but the
advantage here is that restricting the permutation to a subset (as
we did above by defining the permutations to be $\perp$ outside the
``paths''), corresponds to simply dropping some elements from the
tuple: $[0,1]^{T}\mapsto[1,3]^{T}\mapsto[2,0]^{T}$. See \Figref{Sequence-of-tuples}.}

\begin{figure}
\begin{centering}
\includegraphics[width=6cm]{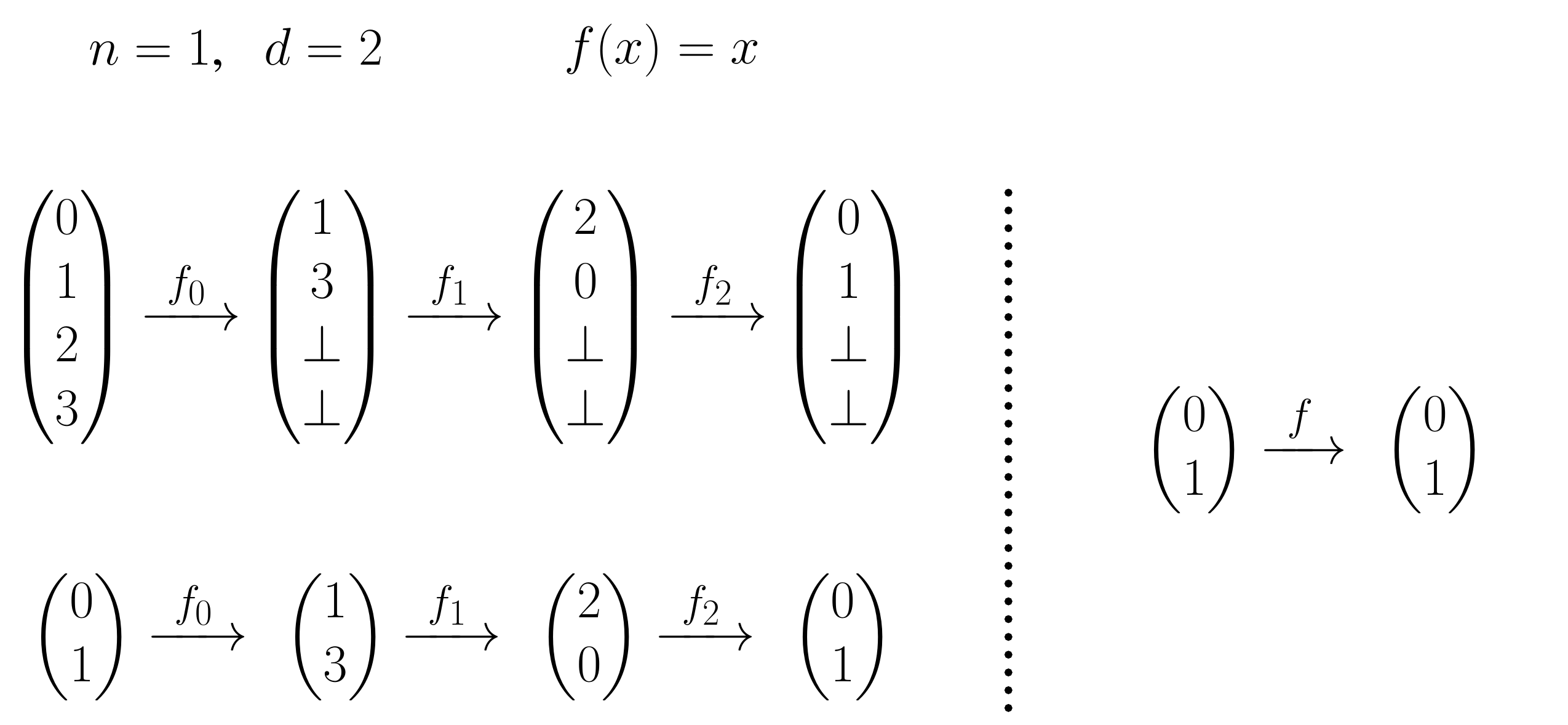}
\par\end{centering}
\caption{Sequence of tuples instead of restricted permutation functions.\label{fig:Sequence-of-tuples} }

\end{figure}

\begin{defn}[Uniform $d$-Shuffler for $f$]
 \label{def:dShuffler}A \emph{uniform $d$-Shuffler} for a function
$f:\{0,1\}^{n}\to\{0,1\}^{n}$ is a sequence of random functions $(f_{0},f_{1},\dots f_{d})$,
where $f_{0},\dots f_{d}:\{0,1\}^{2n}\cup\{\perp\}\to\{0,1\}^{2n}\cup\{\perp\}$
sampled from a distribution $\mathbb{F}_{{\rm shuff}}(d,n,f)$ (to
be defined) such that $f_{d}\circ f_{d-1}\dots\circ f_{0}(x)=f(x)$
for all $x\in Z\times\{0,1\}^{n}$ where $Z:=\{(0,\dots0)\}$ (containing
an $n$-bit zero string). Let $\mathcal{L}$ be the oracle associated
with $(f_{i})_{i=0}^{d}$ and define $\mathbb{O}_{{\rm shuff}}(d,n,f)$
to be the corresponding distribution. 

The sampling process is defined in two equivalent ways. \\
\textbf{First definition: }
\begin{itemize}
\item Sample $f'_{0},\dots f'_{d-1}:\{0,1\}^{2n}\to\{0,1\}^{2n}$ from a
uniform distribution of permutation functions acting on strings of
length $2n$. 
\item Define $f_{d}$ to be such that $f_{d}\circ f'_{d-1}\circ\dots\circ f'_{0}(x)=f(x)$
for all $x\in Z\times\{0,1\}^{n}$ and $f_{d}\circ f'_{d-1}\circ\dots f'_{0}(x)=\perp$
otherwise.
\item For each $i\in\{0,1\dots d-1\}$, 
\begin{itemize}
\item define $X_{i+1}:=f_{i}'\circ\dots f_{0}'(X_{0})$ where $X_{0}=Z\times\{0,1\}^{n}$
and
\item define 
\[
f_{i}(x):=\begin{cases}
f'_{i}(x) & \forall x\in X_{i-1}\\
\perp & \forall x\notin X_{i-1}.
\end{cases}
\]
\end{itemize}
\end{itemize}
\textbf{Second equivalent definition:}
\begin{itemize}
\item Let $t_{0},t_{1}\dots t_{d-1}$ each be a tuple $(x_{1},x_{2}\dots x_{N})$
of size $N:=2^{n}$, sampled uniformly from the collection of all
size $N$ tuples containing distinct elements $x_{i}\in\{0,1\}^{2n}$.
Let $t_{-1}:=(0,1,\dots N)$ and $t_{d}:=(f(0),f(1)\dots f(N-1))$
(see \Notaref{IntegerString}).
\item $f_{0},\dots f_{d}$ | For each $i\in\{0,\dots,d-1\}$, define $f_{i}$
as follows. 
\begin{itemize}
\item If $x\notin t_{i-1}$, then define $f_{i}(x)=\perp$. 
\item Otherwise, suppose $x$ is the $j$th element in $t_{i-1}$ and $y$
is the $j$th element in $t_{i}$. Then, define $f_{i}(x)=y$. 
\end{itemize}
\end{itemize}
Given the second definition, a $d$-Shuffler may equivalently be defined
as a sequence $(t_{0},\dots t_{d-1},t_{d})$ of tuples sampled as
described above. We overload the notation and let $\mathbb{F}_{{\rm shuff}}(d,n,f)$
return $(t_{i})_{i=0}^{d}$ and similarly let $\mathbb{O}_{{\rm shuff}}(d,n,f)$
return the oracle associated with $(t_{i})_{i=0}^{d}$.

Finally, if $f$ is omitted when $\mathbb{F}_{{\rm shuff}}$ or $\mathbb{O}_{{\rm shuff}}$
are invoked, then it is assumed that $f\sim\mathbb{F}_{R}(n)$.
\end{defn}

\begin{notation}
In this section, we drop the word ``uniform'' for conciseness. 
\end{notation}

\subsubsection{The $d$-SS Problem}
\begin{defn}[$d$-Shuffled Simon's Distribution and Oracle]
\label{def:dShuffledSimonsDistr}We define the $d$-Shuffled Simon's
Function distribution, $\mathbb{F}_{{\rm SS}}(d,n)$ and the corresponding
$d$-Shuffled Simon's Oracle distribution $\mathbb{O}_{{\rm SS}}(d,n)$
by specifying its sampling procedure.
\begin{itemize}
\item Sample a random Simon's function $(f,s)\sim\mathbb{F}_{{\rm S}}(n)$
(see \Defref{SimonsFunctionDistr}).
\item Sample a $d$-Shuffler $(f_{i})_{i=0}^{d}\sim\mathbb{F}_{{\rm Shuff}}(d,n,f)$
where $f_{i}$ are functions (see \Defref{dShuffler}). 
\end{itemize}
Return $((f_{i})_{i=0}^{d},s)$ when $\mathbb{F}_{{\rm SS}}(d,n)$
is sampled and $(\mathcal{F},s)$ when $\mathbb{O}_{{\rm SS}}(d,n)$
is sampled where $\mathcal{F}$ is the oracle associated with $(f_{i})_{i=0}^{d}$. 
\end{defn}

\begin{defn}[$d$-Shuffled Simon's Problem]
 Let $(\mathcal{F},s)\sim\mathbb{O}_{{\rm SS}}(d,n)$ (see \Defref{dShuffledSimonsDistr})
be sampled from the $d$-Shuffled Simon's Oracle distribution and
$\mathcal{Q}\sim\mathbb{O}_{{\rm Shuff}}(d,n)$ be a random $d$-Shuffler
(see \Defref{dShuffler}). The $d$-Shuffled Simon's Problem is,
\begin{itemize}
\item Search version: Given $\mathcal{F}$, find $s$. 
\item Decision version: Given either $\mathcal{F}$ or $\mathcal{Q}$ with
equal probability, output $1$ if $\mathcal{F}$ was given and $0$
otherwise.
\end{itemize}
\end{defn}

\subsection{Shadow Boilerplate}

\branchcolor{black}{For the analysis, we need to consider shadow oracles associated with
the $d$-Shuffled Simon's Oracle. The definition is somewhat redundant---it
is the direct analogue of \Defref{cSerialShadowOracle}. See \Figref{Shadows-for-d-ShuffledSimons}.}
\begin{defn}[$d$-Shuffled Simon's Shadow Oracle]
\label{def:d-ShuffledSimonShadow} Given 
\begin{itemize}
\item $((f_{i})_{i=0}^{d},s)$ from the sample space of $\mathbb{F}_{{\rm SS}}(d,n)$
and 
\item a tuple of subsets $\bar{S}=(S_{1},S_{2}\dots S_{d})$ where each
$S_{i}\subseteq\{0,1\}^{2n}$,
\end{itemize}
let $\mathcal{F}$ be the oracle associated with $(f_{i})_{i=0}^{d}$.
Then, the $d$-Shuffled Simon's Shadow Oracle (or simply the shadow)
$\mathcal{G}$for the oracle $\mathcal{F}$ with respect to $\bar{S}$
is defined to be the oracle associated with $(f_{0},f_{1}'\dots f_{d}')$
where 
\[
f'_{i}:=\begin{cases}
f_{i}(x) & x\in\{0,1\}^{2n}\backslash S_{i}\\
\perp & x\in S_{i}.
\end{cases}
\]
for all $i\in\{1,2\dots d\}$. 
\end{defn}

\begin{figure}
\begin{centering}
\includegraphics[width=10cm]{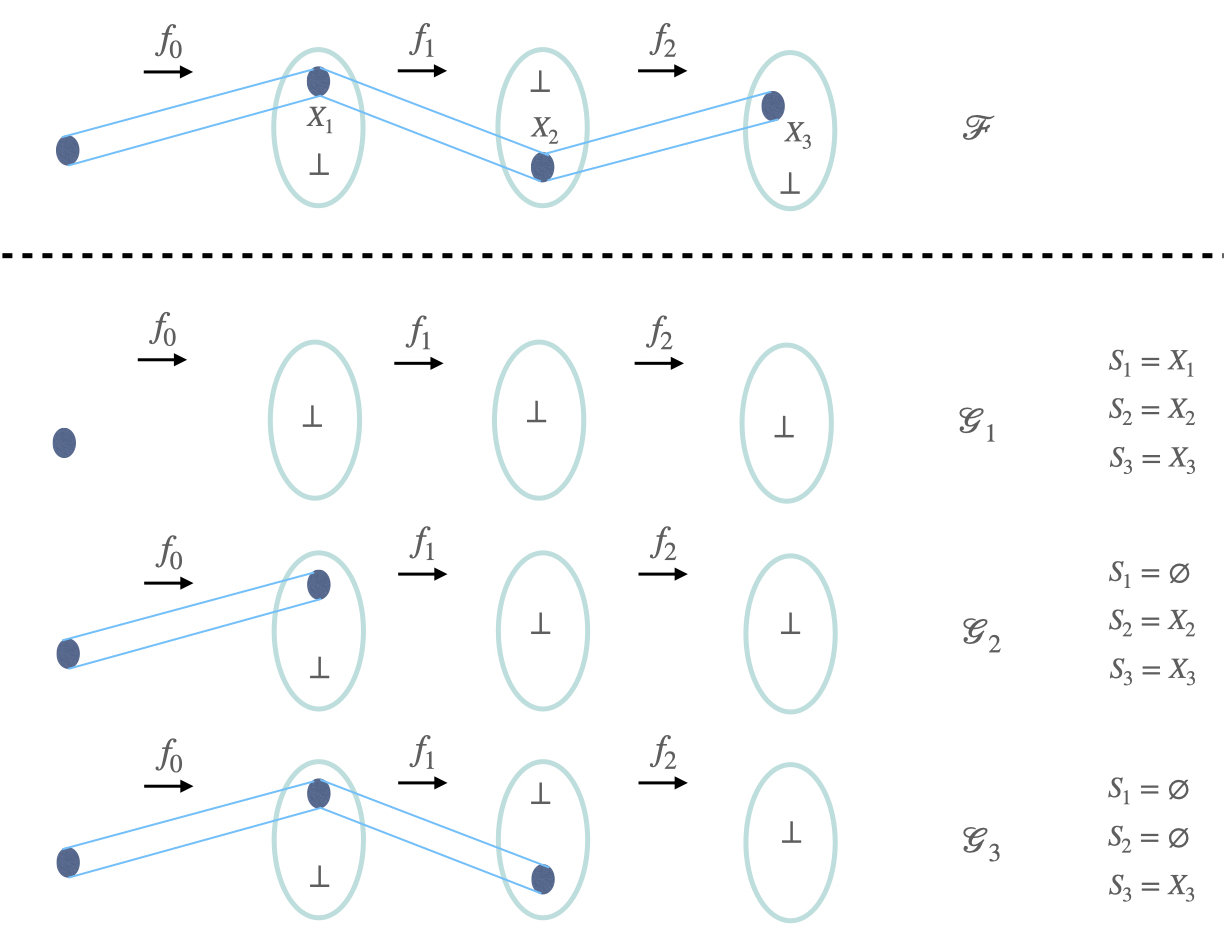}
\par\end{centering}
\caption{\label{fig:Shadows-for-d-ShuffledSimons}Shadows for $d$-Shuffled
Simon's Oracle. Ignore the arrows below the function labels $f_{0},\dots f_{2}$
at first. Suppose $d=3$. Consider the row for $\mathcal{F}$. As
before, the ovals represent the query domains of functions $f_{0},f_{1},f_{2}$
and $f_{3}$. Similarly for the shadows $\mathcal{G}_{1},\dots\mathcal{G}_{3}$.
The shaded regions denote a non-$\perp$ response. Let $\bar{S}_{i}$
be the query domain at which $\mathcal{F}$ and $\mathcal{G}_{i}$
differ. The construction, as before, ensures that $\mathcal{G}_{1}\dots\mathcal{G}_{i-1}$
do not contain any information about $\bar{S}_{i}$.}

\end{figure}

\branchcolor{black}{The analysis here is a straightforward adaptation of the ${\rm QNC}_{d}$
analysis for $d$-Serial Simon's. Thus, we use the analogue of \Algref{SforQNCd_using_dSS}.
The only difference is in the numbering convention. The domain of
the $i$th oracle was determined by $s_{i-1}$ in \Algref{SforQNCd_using_dSS}
while here, the domain of the $i$th oracle is determined by $X_{i}$.}
\begin{lyxalgorithm}[$\bar{S}_{j}$ for ${\rm QNC}_{d}$ exclusion using $d$-SS]
\label{alg:SforQNCd_using_dSS}~\\
Input:
\begin{itemize}
\item $1\le j\le d$ and 
\item $((f_{i})_{i=0}^{d},s)$ from the sample space of $\mathbb{F}_{{\rm SS}}(d,n)$.
\end{itemize}
Output: $\bar{S}_{j}$, a tuple of $d$ subsets, defined as 
\[
\bar{S}_{j}:=\begin{cases}
(\emptyset,\dots\emptyset,X_{j},\dots X_{d}) & j>1\\
(X_{1},\dots X_{d}) & j=1
\end{cases}
\]
where for each $i\in\{1,\dots d\}$, $X_{i}=f_{i-1}(X_{i-1})$ with
$X_{0}=\{0,1\}^{n}$. 

\end{lyxalgorithm}

\begin{prop}
\label{prop:Shadow_dSS}For all $2\le i\le d$, let $\mathcal{G}_{1}\dots\mathcal{G}_{i-1}$
denote the shadows of $(\mathcal{F},s)\in\mathbb{O}_{{\rm SS}}(d,n)$
(see \Defref{d-ShuffledSimonShadow}) with respect to $\bar{S}_{1}\dots\bar{S}_{i-1}$
(constructed using \Algref{SforQNCd_using_dSS} with the index and
$\mathcal{F}$ as inputs). Then, $\mathcal{G}_{1}\dots\mathcal{G}_{i-1}$
contain no information about $\bar{S}_{i}$.
\end{prop}

\branchcolor{black}{\begin{proof}
To see this, note that even though to define $\mathcal{G}_{1}$, we
used $\bar{S}_{1}=(X_{1},X_{2}\dots X_{d}),$ still $\mathcal{G}_{1}$
contains no information about $X_{2},\dots X_{d}$ (and so no information
about $\bar{S}_{2}$ either) because it always outputs $\perp$ (except
for the zeroth oracle which can reveal $X_{1}$). Similarly, $\mathcal{G}_{2}$
contains information about $X_{1},X_{2}$ but not about $X_{3}\dots X_{d}$
(thus $\bar{S}_{3}$). Analogously for $\mathcal{G}_{3}$ and so on. 
\end{proof}
}
\begin{prop}
\label{prop:x_in_S_dSS_shadow}Let $\bar{S}_{i}$ be the output of
\Algref{SforQNCd_using_dSS} with $i$ and $(\mathcal{F},s)$ as inputs
where $(\mathcal{F},s)\sim\mathbb{O}_{{\rm SS}}(d,n)$ (see \Defref{d-ShuffledSimonShadow}).
Let $x$ be some fixed query (in the query domain of $\mathcal{F}$).
Then, $\Pr[x\in\bar{S}_{i}]\le\mathcal{O}(d\cdot2^{-n})$. 
\end{prop}

\branchcolor{black}{\begin{proof}
For any fixed $x_{i}\in\{0,1\}^{2n}$, $\Pr[x_{i}\in X_{i}]=\frac{1}{2^{n}}$
(see \Remref{Pr_x_in_X_or_t} using $N=2^{n}$ and $M=2^{2n}$). The
union bound, then, readily bounds the desired probability.
\end{proof}
}

\subsection{Depth Lower bounds for $d$-Shuffled Simon's Problem ($d$-SS)}

\branchcolor{black}{Again, as in the analysis for $d$-Serial Simon's problem, we begin
with briefly demonstrating the depth lower bound for ${\rm QNC}_{d}$
(with classical post-processing) and then generalise it to ${\rm CQ}_{d}$. }

\subsubsection{$d$-Shuffled Simon's Problem is hard for ${\rm QNC}_{d}$}

\begin{thm}[$d$-SS is hard for ${\rm QNC}_{d}$]
\label{thm:d-SS-is-hard-for-QNC-d}Let $(\mathcal{F},s)\sim\mathbb{O}_{{\rm SS}}(d,n)$,
i.e. let $\mathcal{F}$be an oracle for a random $d$-Shuffled Simon's
problem of size n and period $s$. Let $\mathcal{A}^{\mathcal{F}}$
be any $d$ depth quantum circuit (see \Defref{QNCd} and \Remref{oracleVersionsQNC-CQ-QC})
acting on $\mathcal{O}(n)$ qubits, with query access to $\mathcal{F}$.
Then $\Pr[s\leftarrow\mathcal{A}^{\mathcal{F}}]\le\negl$, i.e. the
probability that the algorithm finds the period is exponentially small. 
\end{thm}

This follows from the proof of \Thmref{QNCd-dSerialSimons} with no
essential change. The main difference is in the details of how the
shadow is constructed---earlier it was comprised of $E\times s_{i}$s,
now $X_{i}$s comprise it. However, that doesn't change the relevant
properties (see \Propref{Shadow_dSS} and \Propref{x_in_S_dSS_shadow})
for the main argument (which is given explicitly in \Subsecref{Hardness-of-dSS-for-QNCd}
for completeness).

We expect the ${\rm QNC}_{d}$ hardness of $d$-SS to readily generalise
to hardness for ${\rm QC}_{d}$ (following the $d$-Serial Simon's
hardness proof) but we do not prove it here.

\subsubsection{$d$-Shuffled Simon's Problem is hard for ${\rm CQ}_{d}$ | idea\label{subsec:idea_dSShardforCQd}}

\branchcolor{black}{The generalisation to ${\rm CQ}_{d}$ takes some work. We first describe
the basic idea behind our approach and then formalise these. These
draw from the insights of \textcite{CCL2020} but differ substantially
in their implementation. Let $((f_{i})_{i=0}^{d},s)\sim\mathbb{F}_{{\rm SS}}(d,n)$.
For our discussion, we require three key concepts. 
\begin{enumerate}
\item The first, already introduced in the proof of \Thmref{QCd_hardness_dSerialSimons},
was the notion of conditioning the oracle distribution (and the quantities
that depend on it) based on a ``transcript'' as the algorithm proceeds
and analysing the conditioned cases. In the proof of \Thmref{QCd_hardness_dSerialSimons},
the locations queried by the classical algorithm constituted this
transcript but it could, and as we shall see it will, more generally
contain other correlated variables.
\item The second, is the notion of an almost uniform $d$-Shuffler. Neglect,
for the moment, the last function $f_{d}$ which is used to define
the $d$-Shuffler, $(f_{i})_{i=0}^{d}$. Then, informally, suppose
that an algorithm (which can even be computationally unbounded) is
given access to a uniform $d$-Shuffler and it produces a $\poly$
sized advice (a string that is correlated with the Shuffler). One
can show that for advice strings which appear with non-vanishing probability,
the $d$-Shuffler conditioned on the advice string, continues to stay
almost uniform. 
\item The third, is basically the bootstrapping of the proof of \Thmref{d-SS-is-hard-for-QNC-d},
by letting it now play the role of the algorithm mentioned above.
This ties the loose end left above, the question of correlation with
$f_{d}$. A generic ${\rm QNC}_{d}$ circuit may reveal some information
about $(f_{0},\dots f_{d-1})$ and $f_{d}$ in the poly-length string
it outputs, however, the proof of \Thmref{d-SS-is-hard-for-QNC-d}
guarantees that the behaviour of this algorithm cannot be very different
from one which has no access to $f_{d}$. 
\end{enumerate}
In the next section, we formalise step two and in the one that follows,
we stitch everything together to establish $d$-SS is hard for ${\rm CQ}_{d}$.
We later use parts of this proof for establishing our second main
result.}

\section{Technical Results II\label{sec:Technical-Results-II}}

\branchcolor{black}{The description in \Subsecref{idea_dSShardforCQd} above, used the
notion of an almost uniform $d$-Shuffler. In this section, we make
this notion precise. }

\subsection{Sampling argument for Uniformly Distributed Permutations\label{subsec:Sampling-argument-for}}

\branchcolor{black}{The basic idea used in this section is called a ``pre-sampling argument''
which we have adapted from the work of \textcite{coretti_random_2017,CCL2020}.
It was considered earlier in cryptographic contexts by \textcite{unruh_random_2007}
and for communication complexity by \textcite{goos_rectangles_2015,kothari_approximating_2017}.
For our purposes, we need to generalise their result.

We first describe the idea in its simplest form (that considered by
\textcite{coretti_random_2017}). Let $N=2^{n}$ and fix some small
$\delta>0$. Suppose there is an $N$-bit random string $X$ which
is uniformly distributed. Suppose an arbitrary function of $X$, $f(X)$
is known. Then, given we focus on values of $f(X)$ which occur with
probability above a threshold, say $\gamma$, i.e. for $r$ such that
$\Pr[f(X)=r]\ge\gamma=2^{-m}$, the main result, informally, is that
$X|(f(X)=r)$ may be viewed as a ``convex combination of random variables''
which are ``$\delta$ far from'' $X$'s with a small number of bits
(scaling as $-\log\gamma/\delta=m/\delta$) fixed. We justify and
reify the phrases in quotes shortly. In our setting, the random variable
$X$ is replaced by the $d$-Shuffler and the function $f$ would
encode the advice generated by a ${\rm QNC}_{d}$ circuit gives after
acting upon the $d$-Shuffler.\footnote{The term ``pre-sampling argument'' arose from the cryptographic
application in the context of random oracles. There, the adversary
is allowed to arbitrarily interact with the random oracle (or pre-sample
it) before initiating the protocol but only allowed to keep a poly
sized advice from that interaction.}

While \textcite{CCL2020} already generalised this argument to the
case of a sequence of permutations over $N$ elements for their analysis,
we show here that the idea itself can be applied quite generally---first,
one needn't restrict to uniform distributions, second, a rather limited
structure on the random variable suffices, i.e. one needn't restrict
to strings or permutations. Using these, one can also show a composition
result, where repeated advice are given. This is pivotal to the analysis
of ${\rm CQ}_{d}$ circuits where ${\rm QNC}_{d}$ circuits repeatedly
give advice.

In the following, for clarity, we present our results for a single
uniformly distributed permutation but do not use that fact in our
derivation. The results thus directly lift to the $d$-Shuffler with
minor tweaks to the notation.}

\subsubsection{Convex Combination of Random Variables}

We first make the notion of ``convex combination of random variables''
precise. Consider a function $f$ which acts on a random permutation,
say $t$, to produce an output, i.e. $f(t)=r$ where $r$ is an element
in the range of $f$.\footnote{The function will later be interpreted as an algorithm and the random
permutation accessed via an oracle.} This range can be arbitrary. We say a \emph{convex combination} $\sum_{i}p_{i}t_{i}$
of random variables $t_{i}$ is \emph{equivalent} to $t$ if for all
functions $f$, and all outputs $s$ in its range, $\sum_{i}p_{i}\Pr[f(t_{i})=s]=\Pr[f(t)=s]$.
This relation is denoted by $\sum_{i}p_{i}t_{i}\equiv t$. 

\subsubsection{The ``parts'' notation}

While permutations are readily defined as an ordered set of distinct
elements, it would nonetheless be useful to introduce what we call
the ``parts'' notation which allows one to specify parts of the
permutation.
\begin{notation}
\label{nota:permPaths}Consider a permutation $t$ over $N$ elements,
labelled $\{0,1\dots N-1\}$. 
\begin{itemize}
\item \emph{Parts:} Let $S=\{(x_{i},y_{i})\}_{i=1}^{M}$ denote the mapping
of $M\le N$ elements under some permutation, i.e. there is some permutation
$t$, such that $t(x_{i})=y_{i}$. Call any such set $S$ a ``part''
and its constituents ``paths''. 
\begin{itemize}
\item Denote by $\Omega_{{\rm parts}}(N)$ the set of all such ``parts''. 
\item Call two parts $S=\{(x_{i},y_{i})\}_{i}$ and $S'=\{(x'_{i'},y'_{i'})\}_{i'}$
\emph{distinct} if for all $i,i'$ (a) $x_{i}\neq x_{i'}$, and (b)
there is a permutation $t$ such that $t(x_{i})=y_{i}$ and $t(x_{i'})=y_{i'}$.
\item Denote by $\Omega_{{\rm parts}}(N,S)$ the set of all parts $S'\in\Omega_{{\rm parts}}(N)$
such that $S'$ is distinct from $S$.
\end{itemize}
\item \emph{Parts in $t$:} The probability that $t$ maps the elements
as described by $S$ may be expressed as $\Pr[\land_{i=1}^{M}(t(x_{i})=y_{i})]=\Pr[S\subseteq{\rm parts}(t)]$
where ${\rm parts}(t):=\{(x,t(x))\}_{x=0}^{N-1}$. 
\item \emph{Conditioning $t$ based on parts:} Finally, use the notation
$t_{S}$ to denote the random variable $t$ conditioned on $S\subseteq{\rm parts}(t)$.
\end{itemize}
\branchcolor{black}{To clarify the notation, consider the following simple example.}
\end{notation}

\begin{example}
Let $N=2$. Then $\Omega_{{\rm parts}}(N)=\{\{(0,0)\},\{(1,1)\},\{(0,0),(1,1)\},\{(0,1)\},\{(1,0)\},\{(0,1),(1,0)\}\}$
and there are only two permutations, $t(x)=x$ and $t'(x)=x\oplus1$
for all $x\in\{0,1\}$. An example of a part $S$ is $S=\{(0,0)\}$.
A part (in fact the only part) distinct from $S$ is $(1,1)$, i.e.
$\Omega_{{\rm parts}}(N,S)=\{(1,1)\}$. 
\end{example}

\subsubsection{$\delta$ non-uniform distributions}

\branchcolor{black}{Using the ``parts'' notation, we define uniform distributions over
permutations and a notion of being $\delta$ non-uniform---distributions
which are at most $\delta$ ``far from'' being being uniform.\footnote{Clarification to a possible conflict in terms: We use the word uniform
in the sense of probabilities---a uniformly distributed random variable---and
not quite in the complexity theoretic sense---produced by some Turing
Machine without advice. }}
\begin{defn}[uniform and $\delta$ non-uniform distributions]
 \label{def:uniformDistr}Consider the set, $\Omega(N)$, of all
possible permutations of $N$ objects labelled $\{0,1,2\dots N-1\}$.
Let $\mathbb{F}$ be a distribution over $\Omega$. Call $\mathbb{F}$
a \emph{uniform distribution} if for $u\sim\mathbb{F}$, $\Pr[S\subseteq\parts(u)]=\frac{\left(N-M\right)!}{N!}$
for all parts $S\in\Omega_{\parts}(N)$.

An arbitrary distribution $\mathbb{F}^{\delta}$ over $\Omega$ is
$\delta$ non-uniform if it satisfies for $t\sim\mathbb{F}^{\delta}$
\[
\Pr[S\subseteq\parts(t)]\le2^{\delta\left|S\right|}\cdot\Pr[S\subseteq\parts(u)]
\]
for all parts $S\in\Omega_{\parts}(N)$.

Finally, $\mathbb{F}^{p,\delta}$ over $\Omega$ is $(p,\delta)$
non-uniform if there is a subset of parts $S$ of size $|S|\le p$
such that the distribution conditioned on $S$ specifying a part of
the permutation, becomes $\delta$ non-uniform over parts distinct
from $S$. Formally, let $t'\sim\mathbb{F}^{p,\delta}$. Then $t'$
is $(p,\delta)$ non-uniformly distributed if $t'_{S}$ is $\delta$
non-uniformly distributed over all $S'\in\Omega_{\parts}(N,S)$ (see
\Notaref{permPaths}), i.e. 
\begin{equation}
\Pr[S'\subseteq\parts(t')|S\subseteq\parts(t')]\le2^{\delta|S'|}\cdot\Pr[S'\subseteq\parts(u)|S\subseteq\parts(u)].\label{eq:p_delta_uniform}
\end{equation}
\end{defn}

\branchcolor{black}{In \Eqref{p_delta_uniform}, we are conditioning a uniform distribution
using the ``parts'' notation which may be confusing. The following
should serve as a clarification.}
\begin{note}
Let $u\sim\mathbb{F}$ as above. Then, we have $\Pr[S'\subseteq\parts(u)|S\subseteq\parts(u)]=\frac{\left(N-\left|S\right|-\left|S'\right|\right)!}{\left(N-\left|S\right|\right)!}$
where $S'\in\Omega_{\parts}(N,S)$ and $S\in\Omega_{\parts}(N)$.
Let $S=\{(x_{i},y_{i})\}_{i=1}^{|S|}$. Then, the conditioning essentially
specifies that the $|S|$ elements in $X=(x_{i})_{i=1}^{|S|}$ must
be mapped to $Y=(y_{i})_{i=1}^{|S|}$ by $u$, i.e. $u(x_{i})=y_{i}$,
but the remaining elements $\{0,1\dots N-1\}\backslash X$ are mapped
uniformly at random to $\{0,1\dots N-1\}\backslash Y$.
\end{note}

\branchcolor{black}{Clearly, for $\delta=0$, the $\delta$ non-uniform distribution becomes
a uniform distribution. However, this can be achieved by relaxing
the uniformity condition in many ways. The $\delta$ non-uniform distribution
is defined the way it is to have the following property. Notice that
$|S|$ appears in a form such that the product of two probabilities,
$\Pr[S_{1}\subseteq\parts(t)]$ and $\Pr[S_{2}\subseteq\parts(t)]$
yields $|S_{1}|+|S_{2}|$, e.g. $(1+\delta)^{|S|}$ instead of $2^{\delta|S|}$
would also have worked.\footnote{The former was chosen by \textcite{CCL2020} while the latter by \textcite{coretti_random_2017}
and possibly others.} This property plays a key role in establishing that in the main decomposition
(as described informally in \Subsecref{Sampling-argument-for}), the
number of ``paths'' (in the informal discussion it was bits) fixed
is small. We chose the prefactor $2^{|S|}$ for convenience---unlikely
events in our analysis are those which are exponentially suppressed,
and we therefore take the threshold parameter to be $\gamma=2^{-m}$.
These choices result in a simple relation between $|S|$ and $m$. }
\begin{notation}
\label{nota:DeltaUniformlyDistributed}To avoid double negation, we
use the phrase ``$t$ is more than $\delta$ non-uniform'' to mean
that $t$ is not $\delta$ non-uniform. Similarly, we use the phrase
``$t$ is at most $\delta$ non-uniform'' to mean that $t$ is $\delta$
non-uniform. 
\end{notation}

\branchcolor{black}{As shall become evident, the only property of a uniform distribution
we use in proving the main proposition of this section, is the following.
It not only holds for all distributions over permutations, but also
for $d$-Shuffler. We revisit this later.}
\begin{note}
Let $t$ be a permutation sampled from an arbitrary distribution $\mathbb{F}'$
over $\Omega(N)$. Let $S,S'\subseteq\Omega_{\parts}(N)$ be \emph{distinct}
parts (see \Notaref{permPaths}). Then, 
\[
\Pr[S\subseteq\parts(t)\land S'\subseteq\parts(t)]=\Pr[S\cup S'\subseteq\parts(t)].
\]
If the parts are not distinct, then both expressions vanish. 
\end{note}

\subsubsection{Advice on uniform yields $\delta$ non-uniform}

\branchcolor{black}{We are now ready to state and prove the simplest variant of the main
proposition of this section.}
\begin{prop}[$\mathbb{F}|r'\equiv{\rm conv}(\mathbb{F}^{p,\delta})$]
\label{prop:sumOfDeltaNonUni_perm} Premise: 
\begin{itemize}
\item Let $u\sim\mathbb{F}$ where $\mathbb{F}$ is a uniform distribution
over all permutations, $\Omega$, on $\{0,1\dots N-1\}$, as in \Defref{uniformDistr}
with $N=2^{n}$.
\item Let $r$ be a random variable which is arbitrarily correlated to $u$,
i.e. let $r=g(u)$ where $g$ is an arbitrary function.
\item Fix any $\delta>0$, $\gamma=2^{-m}>0$ ($m$ may be a function of
$n$) and some string $r'$.
\item Suppose 
\begin{equation}
\Pr[r=r']\ge\gamma.\label{eq:conditionOn_gamma_and_l}
\end{equation}
\item Let $t$ denote the variable $u$ conditioned on $r=r'$, i.e. let
$t=u|(g(u)=r')$.
\end{itemize}
Then, $t$ is ``$\gamma$-close'' to a convex combination of finitely
many $(p,\delta)$ non-uniform distributions, i.e. 
\[
t\equiv\sum_{i}\alpha_{i}t_{i}+\gamma't'
\]
 where $t_{i}\sim\mathbb{F}_{i}^{p,\delta}$ and $\mathbb{F}_{i}^{p,\delta}$
is $(p,\delta)$ non-uniform with $p=\frac{2m}{\delta}$. The permutation
$t'$ is sampled from an arbitrary (but normalised) distribution over
$\Omega$ and $\gamma'\le\gamma$.
\end{prop}

\branchcolor{black}{\begin{proof}
Suppose that $t$ is more than $\delta$ non-uniformly distributed
(see \Defref{uniformDistr} and \Notaref{DeltaUniformlyDistributed}),
otherwise then there is nothing to prove (set $\alpha_{1}$ to $1$,
and $t_{i}$ to $t$, remaining $\alpha_{i}$s and $\gamma'$ to zero).
Recall $\Omega_{\parts}(N)$ is the set of all parts (see \Notaref{permPaths}).
Let the subset $S\in\Omega_{\parts}(N)$ be the maximal subset of
parts (i.e. subset with the largest size) such that 
\begin{equation}
\Pr[S\subseteq\parts(t)]>2^{\delta\cdot|S|}\cdot\Pr[S\subseteq\parts(u)].\label{eq:S_not-delta-non-uniform}
\end{equation}
\begin{claim}
\label{claim:easyConditionDeltaUniform}Let $S$ and $t$ be as described
above. The random variable $t$ conditioned on being consistent with
the paths in $S\in\Omega_{\parts}(N)$, i.e. $t_{S}$, is $\delta$
non-uniformly distributed over $S'\subseteq\Omega_{\parts}(N,S)$,
is $\delta$ non-uniformly distributed.
\end{claim}

We prove \Claimref{easyConditionDeltaUniform} by contradiction. Suppose
that $t_{S}$ is ``more than'' $\delta$ non-uniform. Then, there
exists some $S'\in\Omega_{\parts}(N,S)$ such that 
\begin{align}
\Pr[S'\subseteq\parts(t_{S})] & =\Pr[S'\subseteq\parts(t)|S\subseteq\parts(t)]>2^{\delta\cdot|S'|}\cdot\Pr[S'\subseteq\parts(u)|S\subseteq\parts(u)].\label{eq:_S'_not_delta}
\end{align}
Since $S'$ violates the $\delta$ non-uniformity condition for $t_{S}$,
the idea is to see if the union $S\cup S'$ violates the $\delta$
non-uniformity condition for $t$. If it does, we have a contradiction
because $S$ was by assumption the maximal subset satisfying this
property. Indeed,
\begin{align*}
\Pr[S\cup S'\subseteq\parts(t)] & =\Pr[S\subseteq\parts(t)\land S'\subseteq\parts(t)] & \text{\ensuremath{\because} \ensuremath{S} and \ensuremath{S'} are distinct}\\
 & =\Pr[S\subseteq\parts(t)]\Pr[S'\subseteq\parts(t)|S\subseteq\parts(t)] & \text{conditional probability}\\
 & >2^{\delta\cdot(|S|+|S'|)}\cdot\Pr[S\subseteq\parts(u)]\Pr[S'\subseteq\parts(u)|S\subseteq\parts(u)] & \text{\text{using \prettyref{eq:S_not-delta-non-uniform} and \prettyref{eq:_S'_not_delta}}}\\
 & =2^{\delta\cdot|S\cup S'|}\cdot\Pr[S\cup S'\subseteq\parts(u)] & \text{\ensuremath{\because} \ensuremath{S} and \ensuremath{S'} are disjoint}
\end{align*}
which completes the proof. 

\Claimref{easyConditionDeltaUniform} shows how to construct a $\delta$
non-uniform distribution after conditioning but we must also bound
$|S|$. This is related to how likely is the $r'$ we are conditioning
upon, i.e. the probability of $g(u)$ being $r'$. 
\begin{claim}
\label{claim:sizeS}One has 
\[
\left|S\right|<\frac{m}{\delta}.
\]
\end{claim}

While \Eqref{S_not-delta-non-uniform} lower bounds $\Pr[S\subseteq\parts(t)]$,
the upper bound is given by
\begin{align}
\Pr[S\subseteq\parts(t)] & =\Pr[S\subseteq\parts(u)|(g(u)=r')]\nonumber \\
 & =\Pr[S\subseteq\parts(u)\land g(u)=r']/\Pr[g(u)=r']\nonumber \\
 & \le\Pr[S\subseteq\parts(u)\land g(u)=r']\cdot\gamma^{-1}\nonumber \\
 & \le\Pr[S\subseteq\parts(u)]\cdot\gamma^{-1}.\label{eq:S_upperBound}
\end{align}
Combining these, we have $2^{\delta\cdot|S|}<2^{m}$, i.e., $|S|<\frac{m}{\delta}$.

Using Bayes rule on the event that $S\subseteq\parts(t)$ we conclude
that 
\[
t\equiv\alpha_{1}t_{1}+\alpha'_{1}t'_{1}
\]
 where $\alpha_{1}=\Pr[S\subseteq\parts(t)]$, $t_{1}=t_{S}$, i.e.
$t$ conditioned on $S\subseteq\parts(t)$, $\alpha'_{1}=1-\alpha_{1}$
and $t'_{1}$ is $t$ conditioned on $S\nsubseteq\parts(t)$. Further,
while $t_{1}$ is $(p,\delta)$ non-uniform (from \Claimref{easyConditionDeltaUniform}
and \Claimref{sizeS}), $t_{1}'$ may not be. Proceeding as we did
for $t$, if $t_{1}'$ is itself $\delta$ non-uniform, there is nothing
left to prove (we set $\alpha_{2}=\alpha_{1}'$ and $t_{2}=t_{1}'$
and the remaining $\alpha_{i}$s and $\gamma'$ to zero). Also assume
that $\alpha_{1}'>\gamma$ because otherwise, again, there is nothing
to prove. 

Therefore, suppose that $t'_{1}$ is not $\delta$ non-uniform. Note
that the proof of \Claimref{easyConditionDeltaUniform} goes through
for any permutation which is not $\delta$ non-uniform. Thus, the
claim also applies to $t_{1}'$ where we denote the maximal set of
parts by $S_{1}$. Let $t_{2}$ be $t_{1}'$ conditioned on $S_{1}\subseteq\parts(t_{1}')$
and $t_{2}'$ be $t_{1}'$ conditioned on $S_{1}\nsubseteq\parts(t_{1}')$.
Using Bayes rule as before, we have 
\[
t\equiv\alpha_{1}t_{1}+\alpha_{2}t_{2}+\alpha_{2}'t_{2}'.
\]
Adapting the statement of \Claimref{easyConditionDeltaUniform} (with
$t_{1}'$ playing the role of $t$ and $S_{1}$ playing the role of
$S$) to this case, we conclude that $t_{2}$ is $\delta$ non-uniform
but we still need to show that $|S_{1}|\le p$. We need the analogue
of \Claimref{sizeS} which we assert is essentially unchanged.
\begin{claim}
\label{claim:S_k_bound_general}One has
\begin{equation}
\left|S_{i}\right|<\frac{2m}{\delta}.\label{eq:BoundS1}
\end{equation}
\end{claim}

The proof is deferred to \Subsecref{tech_res_non-uniform}. The factor
of two appears because for the general case, we use both $\alpha_{i}'>\gamma$
and $\Pr[g(u)=r']>\gamma$. One can iterate the argument above. Suppose
\begin{equation}
t\equiv\alpha_{1}t_{1}+\dots\alpha_{j}t_{j}+\alpha_{j}'t'_{j}\label{eq:generalSumT}
\end{equation}
 where $t_{1},\dots t_{j}$ are $(p,\delta)$ non uniformly distributed
while $t'_{j}$ is not and $\alpha_{j}':=\Pr[S\nsubseteq\parts(t)\land\dots\land S_{j-1}\nsubseteq\parts(t)]>\gamma$
(else one need not iterate). Let $S_{j}$ be the maximal set such
that $t_{j+1}:=t'_{j}|S_{j}\subseteq\parts(t_{j}')$ is $\delta$
non-uniform (which must exist from \Claimref{easyConditionDeltaUniform})
and let $t_{j+1}':=t_{j}'|S_{j}\nsubseteq\parts(t_{j}')$. Let $\alpha'_{j+1}:=\Pr[S_{j}\nsubseteq\parts(t'_{j})]$
which equals $\Pr[S\nsubseteq\parts(t)\land\dots\land S_{j}\nsubseteq\parts(t)]$.
From \Claimref{S_k_bound_general}, $\left|S_{j}\right|<2m/\delta\le p$
therefore $t_{j+1}$ is $(p,\delta)$ non-uniform. 

We now argue that the sum in \Claimref{S_k_bound_general} contains
finitely many terms. At every iteration, $\alpha'_{i}$ strictly decreases
because at each step, more constraints are added; $S_{i}\neq S_{j}$
for all $i\neq j$ (otherwise conditioning on $S_{j}$ (if $j\ge i$)
as in \Claimref{easyConditionDeltaUniform} could not have any effect).
Since $\Omega_{\parts}(N)$ is finite, the decreasing sequence $\alpha'_{1}\dots\alpha'_{i}$
must, for some integer $i$, satisfy $\alpha_{i}\le\gamma$ after
finitely many iterations.
\end{proof}
}

\subsubsection{Iterating advice and conditioning on uniform distributions | $\delta$
non-$\beta$-uniform distributions}

\branchcolor{black}{Once generalised to the $d$-Shuffler (which, as we shall, see is
surprisingly simple), recall that the way we intend to use the above
result is to repeatedly get advice from a quantum circuit, a role
played by $g$ in the previous discussion. However, the way it is
currently stated, one starts with a uniformly distributed permutation
$u$ for which some advice $g(u)$ is given but one ends up with $(p,\delta)$
non-uniform distributions. We want the result to apply even when we
start with a $(p,\delta)$ non-uniform distribution. 

As should become evident shortly, the right generalisation of \Propref{sumOfDeltaNonUni_perm}
for our purposes is as follows. Assume that the advice being conditioned
occurs with probability at least $\gamma=2^{-m}$ and think of $m$
as being polynomial in $n$; $\delta>0$ is some constant and $p=2m/\delta$. 
\begin{itemize}
\item Step 1: Let $t\sim\mathbb{F}^{\delta'}$ be $\delta'$ non-uniform\footnote{Notation: When I say $t$ is $\delta$ non-uniform, it is implied
that $t$ is sampled from a $\delta$ non-uniform distribution.} and $s\sim\mathbb{F}^{\delta'}|r$ be $t|(g(t)=r)$. Then it is straightforward
to show that $s\equiv\sum_{i}\alpha_{i}s_{i}$ where $s_{i}$ are
$(p,\delta+\delta')$ non-uniform, which we succinctly write as 
\[
\mathbb{F}^{\delta'}|r\equiv{\rm conv}(\mathbb{F}^{p,\delta+\delta'}).
\]
\end{itemize}
Observation: If $t\sim\mathbb{F}^{p,\delta}$ is $(p,\delta)$ non-uniform,
then there is some $S$ of size at most $p$ such that $t\sim\mathbb{F}^{\delta|\beta}$
is $\delta$ non-$\beta$-uniform where\footnote{The conditioning is in superscript because it is non-standard; standard
would be $S\subseteq\parts(t)$ which is too long.} $\beta:=(S)$. A $\beta$-uniform distribution is simply a uniform
distribution conditioned on having $S$ as parts. This amounts to
basically making the conditioning explicit. Having this control will
be of benefit later.
\begin{itemize}
\item Step 2: It is not hard to show that Step 1 goes through unchanged
if non-uniform is replaced with non-$\beta$-uniform for an arbitrary
$\beta$.
\end{itemize}
These combine to yield the following. Let $t\sim\mathbb{F}^{\delta'|\beta}$
be a $\delta'$ non-$\beta$-uniform distribution and $s\sim\mathbb{F}^{\delta'|\beta}|r$
be $t|(g(t)=r)$. Then $t\equiv\sum_{i}\alpha_{i}s_{i}$ where $s_{i}\sim\mathbb{F}^{p,\delta+\delta'|\beta}$
are $(p,\delta+\delta')$ non-$\beta$-uniform,\footnote{The last term with $\alpha_{k}<\gamma$ is suppressed for clarity
in this informal discussion.} which we briefly express as 
\[
\mathbb{F}^{\delta'|\beta}|r\equiv{\rm conv}(\mathbb{F}^{p,\delta+\delta'|\beta}).
\]
 Observe that this composes well, 
\begin{equation}
\mathbb{F}^{p,\delta+\delta'|\beta}|r\equiv{\rm conv}(\mathbb{F}^{2p,2\delta+\delta'|\beta}).\label{eq:comp}
\end{equation}
To see this, consider the following:
\begin{itemize}
\item For some $S_{i}$, $s_{i}$ (as defined in the statement above) is
$\delta'':=\delta+\delta'$ non-$\beta'$-uniform where $\beta':=(S\cup S_{i})$
if $\beta=(S)$. 
\item With $t$ set to $s_{i}$, $\beta$ set to $\beta'$, one can apply
the above to get $s_{i}|(h(s_{i})=r')\equiv\sum_{i}\alpha'_{i}q_{i}$
where $q_{i}$ are $(p,\delta+\delta'')$ non-$\beta'$-uniform. 
\item Note that $q_{i}$ are also $(2p,2\delta+\delta')$ non-$\beta$-uniform;
which we succinctly denoted as $\mathbb{F}^{2p,2\delta+\delta'|\beta}$.
\end{itemize}
Clearly, if this procedure is repeated $\tilde{n}\le\poly$ times,
starting from $\delta=0$ and $\beta=(\emptyset)$, then the final
convex combination would be over $\mathbb{F}^{\tilde{n}p,\tilde{n}\delta}$.
As we shall see, for our use, it suffices to ensure that $\tilde{n}\delta$
is a small constant and that $\tilde{n}p=\frac{\tilde{n}m}{\delta}\le\poly$.
Choosing $\delta=\Delta/\tilde{n}$ for some small fixed $\Delta>0$
yields $\tilde{n}\delta=\Delta$ and $\tilde{n}p=\frac{\tilde{n}^{2}m}{\Delta}$
which is indeed bounded by $\poly$ (recall $m$ and $\tilde{n}$
are bounded by $\poly$).}

\branchcolor{black}{One can define a notion of closeness to any arbitrary distribution,
as we did for closeness to uniform. To this end, first consider the
following.}
\begin{defn}[$\delta$ non-$\mathbb{G}$ distributions---$\mathbb{G}^{\delta}$]
 \label{def:nonG}Let $s$ be sampled from an arbitrary distribution,
$\mathbb{G}$, over the set of all permutations $\Omega(N)$ of $N$
objects and fix any $\delta>0$. 

Then, a distribution $\mathbb{G}^{\delta}$ is\emph{ $\delta$ non-$\mathbb{G}$}
if $s'\sim\mathbb{G}^{\delta}$ satisfies 
\[
\Pr[S\subseteq\parts(s')]\le2^{\delta|S|}\cdot\Pr[S\subseteq\parts(s)]
\]
 for all $S\in\Omega_{\parts}(N)$. 

Similarly, a distribution $\mathbb{G}^{p,\delta}$ is $(p,\delta)$
non-$\mathbb{G}$ if there is a subset $S'\in\Omega_{\parts}(N)$
of size at most $\left|S'\right|\le p$ such that conditioned on $S'\subseteq\parts(s)$,
$s''\sim\mathbb{G}^{p,\delta}$ satisfies 
\[
\Pr[S\subseteq\parts(s'')|S'\subseteq\parts(s'')]\le2^{\delta|S'|}\cdot[S\subseteq\parts(s)|S'\subseteq\parts(s)]
\]
 for all $S\in\Omega_{\parts}(N,S')$, i.e. conditioned on $S'$ is
a part of both $s$ and $s''$, $s''$ is $\delta$ non-$\mathbb{G}$.
\end{defn}

\branchcolor{black}{We now define $\beta$-uniform as motivated above and using the previous
definition, define $\delta$ non-$\beta$-uniform.}
\begin{defn}[$\beta$-uniform and $\delta$ non-$\beta$-uniform distributions---$\mathbb{F}^{|\beta}$
and $\mathbb{F}^{\delta|\beta}$]
\label{def:beta-uniform}Let $u\sim\mathbb{F}(N)$ be sampled from
a uniform distribution over all permutations, $\Omega(N)$, of $\{0,1\dots N-1\}$
as in \Notaref{DeltaUniformlyDistributed}. A permutation $s\sim\mathbb{F}^{|\beta}(N)$
sampled from a $\beta$-uniform distribution is $s=u|(S\subseteq\parts(u))$
where\footnote{As alluded to earlier, we define $\beta$ to be a redundant-looking
``one-tuple'' $(S)$ here but this is because later when we generalise
to $d$-Shufflers, we set $\beta=(S,T)$ where $T$ encodes paths
not in $u$.} $\beta=:(S)$ and $S\in\Omega_{\parts}(N)$. 

A distribution $\mathbb{F}^{\delta|\beta}$ is $\delta$ non-$\beta$-uniform
if it is $\delta$ non-$\mathbb{G}$ with $\mathbb{G}$ set to a $\beta$-uniform
distribution (see \Defref{nonG}, above). Similarly, a distribution
$\mathbb{F}^{p,\delta|\beta}$ is $(p,\delta)$ non-$\beta$-uniform
if it is $(p,\delta)$ non-$\mathbb{G}$ with $\mathbb{G}$, again,
set to a $\beta$-uniform distribution.
\end{defn}

\branchcolor{black}{We now state the general version of \Propref{sumOfDeltaNonUni_perm}.}
\begin{prop}[$\mathbb{F}^{\delta'|\beta}|r'={\rm conv}(\mathbb{F}^{(p,\delta+\delta')|\beta})$]
\label{prop:composableP_Delta_non_beta_uniform}Let $t\sim\mathbb{F}^{\delta'|\beta}(N)$
be sampled from a $\delta'$ non-$\beta$-uniform distribution with
$N=2^{n}$. Fix any $\delta>0$ and let $\gamma=2^{-m}$ be some function
of $n$. Let $s\sim\mathbb{F}^{\delta'|\beta}|r$, i.e. $s=t|(h(t)=r)$
and suppose $\Pr[h(t)=r]\ge\gamma$ where $h$ is an arbitrary function
and $r$ some string in its range. Then $s$ is ``$\gamma$-close''
to a convex combination of finitely many $(p,\delta+\delta')$ non-$\beta$-uniform
distributions, i.e. 
\[
s\equiv\sum_{i}\alpha_{i}s_{i}+\gamma's'
\]
where $s_{i}\sim\mathbb{F}_{i}^{p,\delta+\delta'|\beta}$ with $p=2m/\delta$.
The permutation $s'$ may have an arbitrary distribution (over $\Omega(2^{n})$)
but $\gamma'\le\gamma$.
\end{prop}

The proof follows from minor modifications to that of \Propref{sumOfDeltaNonUni_perm}
and is thus deferred to \Subsecref{tech_res_non-uniform}.

\subsection{Sampling argument for the $d$-Shuffler\label{subsec:Implications-for--Shuffler}}

\branchcolor{black}{Our objective in this section is to state the analogue of \Propref{composableP_Delta_non_beta_uniform}
for the $d$-Shuffler. To this end, we first define a more abstract
notation for the $d$-Shuffler which builds upon \Defref{dShuffler}.}
\begin{notation}[Abstract notation for $d$-Shufflers]
\label{nota:Abstract-d-Shuffler} Represent a uniform $d$-Shuffler
sampled from $\mathbb{F}_{{\rm shuff}}(d,n,f)$ (see \Defref{dShuffler})
abstractly as $\Xi$. Denote by 
\begin{itemize}
\item ${\rm func}_{i}(\Xi)$ the function $f_{i}$, for $i\in\{0,1\dots d\}$,
which correspond to the first definition, 
\item ${\rm tup}(\Xi)$ the tuples $(t_{0},\dots t_{d})$, which correspond
to the second definition,
\begin{itemize}
\item ${\rm dom}_{i}(\Xi)$ the unordered set corresponding to the tuple
$t_{i-1}$; ${\rm dom}(\Xi):=({\rm dom}_{i}(\Xi))_{i=0}^{d}$ and
${\rm dom}(\Xi^{*}):=({\rm dom}_{i}(\Xi))_{i=0}^{d-1}$,
\end{itemize}
\item ${\rm mat}(\Xi)$ the tuples $(t_{i})_{i=0}^{d}$ stacked as columns
in a matrix, 
\[
\left[\begin{array}{ccccc}
0 & t_{0}[0] & t_{1}[0] & \dots & t_{d}[0]\\
1 & t_{0}[1] & t_{1}[1] & \dots & t_{d}[1]\\
\vdots & \vdots & \vdots &  & \vdots\\
N-1 & t_{0}[N-1] & t_{1}[N-1] & \dots & t_{d}[N-1]
\end{array}\right]=\left[\begin{array}{ccccc}
0 & f_{0}(0) & f_{1}f_{0}(0) & \dots & f_{d}f_{d-1}\dots f_{0}[0]\\
1 & f_{0}(1) & f_{1}f_{0}(1) & \dots & f_{d}f_{d-1}\dots f_{0}[1]\\
\vdots & \vdots & \vdots &  & \vdots\\
N-1 & f_{0}(N-1) & f_{1}f_{0}(N-1) & \dots & f_{d}f_{d-1}\dots f_{0}[N-1]
\end{array}\right],
\]
of size $N\times(d+2)$ where $N=2^{n}$, ${\rm mat}(\Xi^{*})$ the
same matrix with the last column removed, 
\item $\paths(\Xi)$ the set of rows of $\mat(\Xi)$, i.e. $\paths(\Xi):=\{\mat(\Xi)[i]\}_{i=0}^{N-1}$,
and similarly $\paths(\Xi^{*})$ the set of rows of $\mat(\Xi^{*})$,
i.e. $\paths(\Xi^{*}):=\{\mat(\Xi^{*})[i]\}_{i=0}^{N-1}$,
\item $\parts(\Xi)$ the power set of rows of ${\rm mat}(\Xi)$, i.e. $\parts(\Xi):=\mathscr{P}[\paths(\Xi)]=\{{\rm mat}(\Xi)[i]\}_{i=0}^{N-1}$,
and similarly $\parts(\Xi^{*}):=\mathscr{P}[\paths(\Xi^{*})]=\mathscr{P}\{{\rm mat}(\Xi^{*})[i]\}_{i=0}^{N-1}$,
\item $\beta(\Xi)$ the ordered set $(\parts(\Xi^{*}),{\rm dom}(\Xi^{*}),\parts(\Xi),{\rm dom}(\Xi))$.
\end{itemize}
Call $\Xi^{*}$ an empty $d$-Shuffler because it doesn't contain
any information about $f$ (which is contained in $f_{d}$).
\end{notation}

\branchcolor{black}{We can now define the notion of paths and parts for an empty $d$-Shuffler,
the analogue of \Notaref{permPaths}. To keep the notation simple,
we overload the symbol $\Omega_{\parts}$. We use $\Omega_{\parts}(d,n)$.}
\begin{notation}
\label{nota:dShufflerPartsPaths}Further, denote by $\boldsymbol{x}=(x_{-1},x_{0},\dots x_{d-1})$
a tuple of $d+1$ elements. Use $\boldsymbol{x}[i]$ to denote the
$i$th element of the tuple, i.e. $x_{i}$.
\begin{itemize}
\item \emph{Parts}: Consider a set $S=\{\boldsymbol{x}_{j}\}_{j=0}^{M}$
with $M\le N$ containing mappings of $\{\boldsymbol{x}_{j}[-1]\}_{j}$
specified by some empty $d$-Shuffler, i.e. there is some $d$-Shuffler
$\Xi$ such that $S\subseteq\parts(\Xi^{*})$. Call any such set $S$
a ``part'' and its constituents ``paths''.
\begin{itemize}
\item Denote by $\Omega_{\parts}(d,n)$ the set of all such sets $S$, i.e.
the set of all ``parts''.
\item Call two parts $S$ and $S'$ \emph{distinct} if $S\cap S'=\emptyset$
and there is some $d$-Shuffler $\Xi$ such that $S\cup S'\subseteq\parts(\Xi^{*})$. 
\item Denote by $\Omega_{\parts}(d,n,S)$ the set of all parts $S'\in\Omega_{\parts}(d,n)$
such that $S'$ is distinct from $S$.
\end{itemize}
\item \emph{Parts in $\Xi^{*}$}: The probability that $\Xi^{*}$ maps the
elements as described by $S$ may be expressed as $\Pr[\land_{j}\boldsymbol{x}_{j}\in\parts(\Xi^{*})]=\Pr[S\subseteq\parts(\Xi^{*})]$. 
\end{itemize}
\end{notation}

\begin{example}
Let $N=2^{n}=2$ and $d=2$. Then an empty $d$-Shuffler $\Xi^{*}$
can take the following values
\[
\mat(\Xi^{*})\in\left\{ \left[\begin{array}{ccc}
0 & 0 & 0\\
1 & 1 & 1
\end{array}\right],\left[\begin{array}{ccc}
0 & 0 & 1\\
1 & 1 & 0
\end{array}\right],\left[\begin{array}{ccc}
0 & 1 & 0\\
1 & 0 & 1
\end{array}\right],\left[\begin{array}{ccc}
0 & 1 & 1\\
1 & 0 & 0
\end{array}\right],\left[\begin{array}{ccc}
0 & 0 & 0\\
1 & 2 & 2
\end{array}\right],\left[\begin{array}{ccc}
0 & 0 & 2\\
1 & 2 & 0
\end{array}\right]\dots\right\} .
\]
Let $(t_{0},t_{1}):={\rm tup}(\Xi^{*})$. Then the first of these
corresponds to $t_{0}=t_{1}=\mathbb{I}$ (identity permutation) and
$f_{0}=f_{1}=\mathbb{I}$, the second to $t_{0}=\mathbb{I}$, $t_{1}=\mathbb{X}$
(``swap'') and $f_{0}=\mathbb{I},f_{1}=\mathbb{X}$, the third to
$t_{0}=\mathbb{X}$, $t_{1}=\mathbb{I}$ and $f_{0}=\mathbb{X}$,
$f_{1}=\mathbb{X}$, the fourth to $t_{0}=\mathbb{X}$, $t_{1}=\mathbb{X}$
and $f_{0}=\mathbb{X}$, $f_{1}=\mathbb{I}$, and so on. The set of
parts is $\Omega_{\parts}(d,n)=\{\mathscr{P}\{(0,0,0),(1,1,1)\},\mathscr{P}\{(0,0,1),(1,1,0)\},\mathscr{P}\{(0,1,0),(1,0,1)\},\mathscr{P}\{(0,1,1),(1,0,0)\},\mathscr{P}\{(0,0,0),(1,2,2)\},\dots\}.$
An example of a part is $S=\{(0,0,0)\}$. Parts distinct from $S$,
i.e. elements of $\Omega_{\parts}(n,d,S)$, are $\{(1,1,1)\},\{(1,2,1)\},\{(1,1,2)\}\dots$. 
\end{example}

\branchcolor{black}{We can now define $\beta$-uniform distributions as uniform distributions
over $d$-Shufflers, given some information about the $d$-Shuffler.
Unlike the simpler case of permutations where $\beta$ was simply
specifying parts, we would now like to condition on both, certain
parts being present in ${\parts}(\Xi)$ and certain locations yielding
$\perp$, i.e. queries outside ${\rm dom}(\Xi)$. While this may look
involved, it is simply conditioning on certain aspects of the uniform
distribution over $d$-Shufflers.

The notion of $\delta$ non $\beta$-uniform distribution and $(p,\delta)$
non $\beta$-uniform distribution is defined analogous to the simpler
case of permutations. The only subtlety is that the conditioning is
only over the parts of the empty $d$-Shuffler, i.e. the information
about the function $f$ is not included. Why do we do this?

Suppose $d=0$ and $f$ is a Simon's function. Then having a $(p,\delta)$
non $\beta$-uniform distribution is pointless. The advice could just
be the period of $f$ and conditioning on even one path with collision,
would mean we restrict to distributions over $f$ with that same period
(for more details, see \Subsecref{Delta_with_Simons_not_smart}).

Thus, the notion of $(p,\delta)$ non $\beta$-uniform distributions
only makes sense when the underlying distribution is not severely
constrained upon some paths being revealed. This allows us to show
that the output distribution given the advice and the $(p,\delta)$
non $\beta$-uniform distributions is not very different from that
produced by being given access to only a $\beta$-uniform distribution.

More concretely, we show that the ``finding'' probability even with
$\Xi^{*}$, where $\Xi\sim\mathbb{F}_{{\rm shuff}}^{p,\delta|\beta}$,
is exponentially suppressed so long as $p$ and the conditions in
$\beta$ are poly sized. We return to this after stating the main
result.}
\begin{defn}[$\beta$-uniform and $(p,\delta)$ non $\beta$-uniform distributions
for $d$-Shufflers---$\mathbb{F}_{{\rm shuff}}^{|\beta}$ and $\mathbb{F}_{{\rm shuff}}^{p,\delta|\beta}$]
\label{def:d-Shuffler_p_delta_uniform} Let $\Xi\sim\mathbb{F}_{{\rm shuff}}(d,n,f)$
be a $d$-Shuffler sampled from the uniform distribution over $d$-Shufflers.
Let $\beta=(\beta_{{\rm incl}},\beta_{{\rm excl}})$ where $\beta_{{\rm incl}}$
specifies elements in $\beta(\Xi)$ and ${\rm \beta_{excl}}$ specifies
elements not in $\beta(\Xi)$ (see \Notaref{Abstract-d-Shuffler}).
A $d$-Shuffler $\Xi'\sim\mathbb{F}_{{\rm shuff}}^{|\beta}(d,n,f)$
sampled from a $\beta$-uniform distribution over $d$-Shufflers is
$\Xi'=\Xi|(\beta_{{\rm incl}}\in\beta(\Xi)\land\beta_{{\rm excl}}\notin\beta(\Xi))$. 

A distribution $\mathbb{F}_{{\rm shuff}}^{\delta|\beta}$ is $\delta$
non-$\beta$-uniform over $d$-Shufflers if it is $\delta$ non-$\mathbb{G}$
with $\mathbb{G}$ set to a $\beta$-uniform distribution over $d$-Shufflers
(see \Defref{nonG}) and the conditioning is over the empty $d$-Shuffler's
parts, i.e. $\parts(\Xi^{*})$ for $\Xi\sim\mathbb{F}_{{\rm shuff}}^{\delta|\beta}$.
Similarly, a distribution $\mathbb{F}_{{\rm shuff}}^{p,\delta|\beta}$
is $(p,\delta)$ non-$\beta$-uniform over $d$-Shufflers if it is
$(p,\delta)$ non-$\mathbb{G}$ with $\mathbb{G}$, again, set to
a $\beta$-uniform distribution over $d$-Shufflers (and the conditioning
is over the empty $d$-Shuffler's parts).
\end{defn}

\branchcolor{black}{As we already noted, the proof of \Defref{beta-uniform} does not
depend on the underlying distribution $\mathbb{F}$ and the conditioning
$\beta$. We can therefore lift that result directly for the case
of $d$-Shufflers. }
\begin{prop}[$\mathbb{F}_{{\rm shuff}}^{\delta'|\beta}|r'={\rm conv}(\mathbb{F}_{{\rm shuff}}^{(p,\delta+\delta')|\beta})$]
\label{prop:main_delta_non_uniform_d_shuffler} Let $\Xi^{t}\sim\mathbb{F}_{{\rm shuff}}^{\delta'|\beta}(d,n,f)$
be sampled from a $\delta'$ non $\beta$-uniform distribution over
all $d$-Shufflers. Let $N=2^{n}$. Fix any $\delta>0$ and let $\gamma=2^{-m}$
be some function of $n$. Let $\Xi^{s}\sim\mathbb{F}^{\delta'|\beta}|r$,
i.e. $\Xi^{s}:=\Xi^{t}|(h(\Xi^{t})=r)$ and suppose that $\Pr[h(\Xi^{t})=r]\ge\gamma$
where $h$ is an arbitrary function and $r$ some string in its range.
Then $\Xi^{s}$ is ``$\gamma$-close'' to a convex combination of
finitely many $(p,\delta+\delta')$ non-$\beta$-uniform distributions
over $d$-Shufflers, i.e. 
\begin{equation}
\Xi^{s}\equiv\sum_{i}\alpha_{i}\Xi_{i}^{s}+\gamma'\Xi^{s\prime}\label{eq:main_delta_non_uniform_d_shuffler}
\end{equation}
where $\Xi_{i}^{s}\sim\mathbb{F}_{{\rm shuff},i}^{(p,\delta+\delta')|\beta}$
with $p=2m/\delta$ (the $i$ in $\mathbb{F}_{{\rm shuff},i}^{(p,\delta+\delta')|\beta}$,
indicates that each $\Xi_{i}^{s}$ can come from a different distribution
which is still $(p,\delta+\delta')$ non-$\beta$-uniform; e.g. they
may be fixing different paths but their count is bounded by $p$).
The d-Shufflers $\Xi^{s\prime}$ may have arisen from an arbitrary
distribution over $d$-Shufflers but $\gamma'\le\gamma$.
\end{prop}

\branchcolor{black}{To show that the ``finding'' probability remains small even with
$\Xi^{*}$ when $\Xi\sim\mathbb{F}_{{\rm Shuff}}^{(p,\delta)}$, the
key ingredient is the following lemma. To make the $\beta$ notation
easier to use, we introduce the following.}
\begin{notation}
\label{nota:ProperBeta}We say $\beta=(\beta_{{\rm incl}},\beta_{{\rm excl}})$
as introduced in \Defref{d-Shuffler_p_delta_uniform} is \emph{proper}
if the following holds. Let $\beta_{{\rm incl}}=:(\{\boldsymbol{x}_{i}\}_{i},(H_{j})_{j=1}^{d-1},\{\boldsymbol{y}_{i}\}_{i},(I_{j})_{j=1}^{d})$
(recall recall $\beta(\Xi)=(\parts(\Xi^{*}),{\rm dom}(\Xi^{*}),\parts(\Xi),{\rm dom}(\Xi))$).
We require $\boldsymbol{x}_{i}[j]\in H_{j}$ and $\boldsymbol{y}_{i}[j]\in I_{j}$,
i.e. for each path $\boldsymbol{x}_{i}$ (required to be in $\paths(\Xi^{*})$
when $\beta\in\beta(\Xi)$), the constituent points are required to
be in $H_{j}$ (viz. to be in the domain ${\rm dom}_{i}(\Xi^{*})$
when $\beta\in\beta(\Xi)$) and similarly for $\boldsymbol{y}_{i}$
and $I_{j}$. 
\end{notation}

\branchcolor{black}{This requirement is obviously redundant, i.e. if $\beta$ is not proper
and $\beta'$ is made proper by including the necessary elements in
$H_{j}$ and $I_{j}$, then $\Xi|\beta\in\beta(\Xi)$ is identical
to $\Xi|\beta'\in\beta(\Xi)$. However, it serves to simplifying the
notation. We use it below.}
\begin{lem}
\label{lem:delta_beta_output_non_perp}Let\footnote{If one wishes to use $\mathbb{F}_{{\rm Shuff}}^{(p,\delta)|\beta}$,
then one must first explicitly absorb the size $p$ paths fixed into
$\beta'$. Otherwise, one could always pick $x$ among those paths
and obtain $\Pr[x\in{\rm dom}_{i}(\Xi)]=1$.} $\Xi\sim\mathbb{F}_{{\rm Shuff}}^{\delta|\beta}(n,d,f)$ (see \Defref{d-Shuffler_p_delta_uniform})
and suppose that $\beta=(\beta_{{\rm incl}},\beta_{{\rm excl}})$
is proper (as in \Notaref{ProperBeta} above) and that it only specifies
$\poly$ many paths and answers to queries, i.e. let $\beta_{{\rm incl}}=:(\{\boldsymbol{x}_{i}\}_{i=1}^{K},(H_{j})_{j=1}^{d},\{\boldsymbol{y}_{i}\}_{i=1}^{L},(I_{j})_{j=1}^{d})$
then $K,L,|H_{j}|,|I_{j}|\le\poly$ (see \Notaref{Abstract-d-Shuffler})
and similarly for $\beta_{{\rm excl}}$. Fix any $x$ which is neither
specified by $\beta_{{\rm incl}}$, i.e. $x\notin\cup_{j}(H_{j}\cup I_{j})$,
nor by $\beta_{{\rm excl}}$. Then, for each $i\in\{1,2\dots d\}$,
\[
\Pr[x\in{\rm dom}_{i}(\Xi)]=\Pr[{\rm func}_{i}(\Xi)(x)\neq\perp]\le2^{\delta}\cdot\poly\cdot2^{-n}.
\]
\end{lem}

\branchcolor{black}{\begin{proof}
Let $\Xi^{u}\sim\mathbb{F}_{{\rm Shuff}}(n,d,f)$, $M=2^{2n}$, $N=2^{n}$.
Then, it is clear that $\Pr[x\in{\rm dom}_{i}(\Xi^{u})]=N/M=2^{-n}$
by looking at the permutation $t_{i-1}=:{\rm tup}_{i}(\Xi^{u})$.
Consider points with the following domains $y_{0}\in\{0,\dots N-1\},y_{1},\dots y_{d-1}\in\{0,\dots M-1\}$.
We express the aforementioned probability in terms of individual paths
as 
\begin{align*}
\Pr[x\in{\rm dom}_{i}(\Xi^{u})] & =\Pr[\lor_{\{y_{j}:j\neq i\}}\{(y_{0},y_{1}\dots y_{i-2},x,y_{i},\dots y_{d-1})\}\in\parts(\Xi^{u*})]\\
 & =\sum_{\{y_{j}:j\neq i\}}\Pr[\{(y_{0},y_{1}\dots y_{i-2},x,y_{i},\dots y_{d-1})\}\in\parts(\Xi^{u*})]
\end{align*}
because $\Pr[\{(y_{0}\dots,x,\dots y_{d-1})\}\in\parts(\Xi^{u*})\land\{(y_{0}',\dots x,\dots y_{d-1}')\}\in\parts(\Xi^{u*})]=0$
if for any $k$, s.t. $y_{k}\neq y'_{k}$. The logical AND and the
sum are over the domains of $y_{i}$. Suppose $\Xi'\sim\mathbb{F}_{{\rm Shuff}}^{\delta}(n,d,f)$.
Then, using the same reasoning, one has 
\begin{align*}
\Pr[x\in{\rm dom}_{i}(\Xi')] & =\sum_{\{y_{j}:j\neq i\}}\Pr[\{(y_{0},y_{1}\dots y_{i-2},x,y_{i},\dots y_{d-1})\}\in\parts(\Xi^{\prime*})]\\
 & \le2^{\delta}\sum_{\{y_{j}:j\neq i\}}\Pr[\{(y_{0},y_{1}\dots y_{i-2},x,y_{i},\dots y_{d-1})\}\in\parts(\Xi^{u*})]\\
 & =2^{\delta}\Pr[x\in{\rm dom}_{i}(\Xi^{u})]=2^{\delta}\cdot2^{-n}.
\end{align*}
To obtain the main result, it suffices to observe\footnote{One has $\Pr[x\in{\rm dom}_{i}(\Xi^{\beta'})]=N'/M'\le N/(M-\poly)=2^{-n}(1-\poly\cdot2^{-2n})^{-1}\le2^{-n}(1+\poly\cdot2^{-2n})\le\poly\cdot2^{-n}$
where $N'$ is $N$ minus some polynomial and $M'$ is $M$ minus
some polynomial.} that $\Pr[x\in{\rm dom}_{i}(\Xi^{\beta'})]\le\poly\cdot2^{-n}$ where
$\Xi^{\beta'}\sim\mathbb{F}_{{\rm Shuff}}^{|\beta'}(n,d,f)$ and $\beta'$
only specifies polynomially many constraints. 
\end{proof}
}

Final remark for this section, the distribution over oracles $\mathbb{O}_{{\rm Shuff}}^{(p,\delta)|\beta}$
is implicitly defined from $\mathbb{F}_{{\rm Shuff}}^{(p,\delta)|\beta}$
and similarly for others. 

\section{$d$-Shuffled Simon's Problem (cont.)\label{sec:dSS_hard_2}}

\branchcolor{black}{With the sampling argument for the $d$-Shuffler in place, we are
almost ready to return to establishing that $d$-SS is hard for ${\rm CQ}_{d}$
circuits; it remains to introduce (and mildly modify) some $d$-Shuffler
notation to facilitate its use in the context of $d$-SS. }
\begin{defn}[Extension of \Defref{dShuffledSimonsDistr}]
 Let $\mathbb{F}_{{\rm SS}}^{\delta|\beta}(d,n)$ be exactly the
same as $\mathbb{F}_{{\rm SS}}(d,n)$ except that $\mathbb{F}_{{\rm Shuff}}^{\delta|\beta}$
is used instead of $\mathbb{F}_{{\rm Shuff}}$. 
\end{defn}

\begin{notation}
\label{nota:dShufflerPathsParts2} In the previous section (see \Notaref{dShufflerPartsPaths})
we used $\Omega_{\parts}$ to refer to the set of $\parts(\Xi^{*})$
for all $\Xi$ $d$-Shufflers in the sample space of $\mathbb{F}_{{\rm Shuff}}(d,n,f)$. 
\begin{itemize}
\item $\Omega_{\parts}$, $\Omega_{\parts^{*}}$. \\
Now, instead we use $\Omega_{\parts^{*}}$ for the aforementioned
and $\Omega_{\parts}$ for the set of $\parts(\Xi)$ for all $\Xi$.
\item Convention for $\paths$ and $\paths^{*}$; $\boldsymbol{y}$ and
$\boldsymbol{x}$. \\
A part of $\Xi$, $Y\in\Omega_{\parts}(d,n)$, is denoted as $Y=\{\boldsymbol{y}_{k}\}_{k=1}^{|Y|}$
where each path $\boldsymbol{y}_{k}$ is a tuple of $d+2$ elements
indexed from $-1$ to $d$. Similarly a part of $\Xi^{*}$, $X\in\Omega_{\parts^{*}}(d,n)$
is denoted as $X=\{\boldsymbol{x}_{k}\}_{k=1}^{|X|}$ where each path
$\boldsymbol{x}_{k}$ is a tuple of $d+1$ elements indexed from $-1$
to $d-1$.
\end{itemize}
\end{notation}

\subsection{Shadow Boilerplate (cont.)}

\branchcolor{black}{We will need the following generalisation of \Algref{SforQNCd_using_dSS}
to prove our result. }

\begin{lyxalgorithm}[$\bar{S}_{j}$ for ${\rm CQ}_{d}$ exclusion using $d$-SS]
\label{alg:SforCQd_using_dSS}Fix $d$ and $n$.~\\
Input:
\begin{itemize}
\item $1\le j\le d$,
\item a part of a $d$-Shuffler (i.e. a set of paths) $Y=\{\boldsymbol{y}_{k}\}_{k=1}^{|Y|}\in\Omega_{\parts}(d,n)$
to ``expose'' (could be implicitly specified using $\beta$ as in
\Notaref{ProperBeta} and \Notaref{Abstract-d-Shuffler})
\item $((f_{i})_{i=0}^{d},s)$ from the sample space $\mathbb{F}_{{\rm SS}}(d,n)$.
\end{itemize}
Output: $\bar{S}_{j}$, a tuple of $d$ subsets, defined as 
\[
\bar{S}_{j}:=\begin{cases}
(\emptyset,\dots\emptyset,X_{j}\backslash Y_{j},\dots X_{d}\backslash Y_{d}) & j>1\\
(X_{1}\backslash Y_{1},\dots X_{d}\backslash Y_{d}) & j=1
\end{cases}
\]
where for each $i\in\{1,\dots d\}$, $X_{i}=f_{i-1}(X_{i-1})$ with
$X_{0}=\{0,1\}^{n}$ and $Y_{i}=\{\boldsymbol{y}_{k}[i]\}_{k=1}^{|Y|}$.
\end{lyxalgorithm}

\begin{prop}
\label{prop:Shadow_dSS_withFixedSets}Let 
\begin{itemize}
\item $((f_{i})_{i},s)\in\mathbb{F}_{SS}(d,n)$,
\item $Y=\{y_{k}\}_{k=1}^{|Y|}\in\Omega_{\parts}(d,n)$ be a set of paths
(see \Notaref{dShufflerPathsParts2}), 
\item For all $2\le i\le d$, let $\mathcal{G}_{1}\dots\mathcal{G}_{i-1}$
denote the shadows of $\mathcal{F}$ (oracle corresponding to $(f_{i})_{i}$)
with respect to $\bar{S}_{1}\dots\bar{S}_{i-1}$ (constructed using
\Algref{SforCQd_using_dSS} with the index, $Y$, and $((f_{i})_{i},s)$
as input).
\end{itemize}
Then, $\mathcal{G}_{1}\dots\mathcal{G}_{i-1}$ contain no information
about $\bar{S}_{i}$.
\end{prop}

\branchcolor{black}{\begin{proof}
The argument is analogous to that given for \Propref{Shadow_dSS}.
$\mathcal{G}_{1}$ was defined using $\bar{S}_{1}=(X_{1}\backslash Y_{1},\dots X_{d}\backslash Y_{d})$
and yet $\mathcal{G}_{1}$ contains no information $X_{2}\backslash Y_{2},\dots X_{d}\backslash Y_{d}$.
To see this, observe that $\mathcal{G}_{1}$ contains information
about $f_{0}$ which in turn contains information about $X_{1}$,
i.e. $f_{0}(X_{0})=X_{1}$. However, the remaining sub-oracles, i.e.
from $1$ to $d$, output $\perp$ everywhere except for the paths
$Y=\{\boldsymbol{y}_{k}\}_{k=1}^{|Y|}$. Thus, they do not reveal
$X_{2}\backslash Y_{2},\dots X_{d}\backslash Y_{d}$. Similarly, $\mathcal{G}_{2}$
is defined using $\bar{S}_{2}=(\emptyset,X_{2}\backslash Y_{2},\dots X_{d}\backslash Y_{d})$
and while it contains information about $X_{1},X_{2}$ and the paths
$Y$, it contains no information about $X_{3}\backslash Y_{3},\dots X_{d}\backslash Y_{d}$
(thus $\bar{S}_{3}$). Analogously for $\mathcal{G}_{3}$ and so on.
\end{proof}
}

The following is a direct consequence of \Lemref{delta_beta_output_non_perp}
but expressed in a form more convenient for the following discussion.
\begin{prop}
\label{prop:x_in_S_dSS_shadow_withFixedSets}Let the premise be as
in \Propref{Shadow_dSS_withFixedSets}. Suppose the $d$-Shuffler
is sampled from a $\delta$ non-uniform distribution, i.e. $(\Xi,s)\sim\mathbb{F}_{{\rm SS}}^{\delta}$
and let $\mathcal{F}$ be the oracle associated with $\Xi$. Further,
let $\vec{I}=(I_{j})_{j=1}^{d}$ be the ``excluded domain'', i.e.
$I_{j}\cap{\rm dom}_{j}(\Xi)=\emptyset$, where $I_{j}\subseteq\{0,1\}^{2n}$,
such that $\left|I_{j}\right|\le\poly$. Also suppose that $Y$ satisfies
$Y\subseteq\parts(\Xi)$ and $|Y|\le\poly$.

Let $\boldsymbol{x}\in\Omega_{\paths^{*}}(d,n)$ be a query (in the
query domain of $\mathcal{F}$) such that the query is not in the
excluded set, i.e. $\boldsymbol{x}[j]\notin I_{j}$ for any $j\in\{1\dots d\}$
and the query does not intersect with any known paths, i.e. $\boldsymbol{x}\cap Y^{*}=\emptyset$
(let $Y^{*}$ be the set of paths in $Y$ with the last element removed). 

Then, the probability that any part of $\boldsymbol{x}$ lands on
$\bar{S}_{i}$---given that $\bar{S}_{i}$ is not at locations specified
by $\vec{I}$ and that the paths $Y$ are included in the $d$-Shuffler---is
vanishingly small, i.e. $\Pr[\boldsymbol{x}\cap\bar{S}_{i}\neq\emptyset|\vec{I}\land Y]\le\mathcal{O}(2^{\delta}\cdot\poly\cdot2^{-n})$. 
\end{prop}

\branchcolor{black}{\begin{proof}[Proof sketch.]
 Apply \Lemref{delta_beta_output_non_perp} with $\beta=(\beta_{{\rm incl}},\beta_{{\rm excl}})$
where $\beta_{{\rm incl}}':=(\emptyset,\emptyset,Y,\emptyset)$, $\beta_{{\rm incl}}$
is the proper version of $\beta_{{\rm incl}}$ and $\beta_{{\rm excl}}:=(\emptyset,\emptyset,\emptyset,\vec{I})$
(see \Notaref{ProperBeta}). 
\end{proof}
}

\subsection{Depth lower bounds for $d$-Shuffled Simon's Problem ($d$-SS) (cont.)
\label{subsec:CQ-d-SS-hardness}}

\subsubsection{$d$-Shuffled Simon's is hard for ${\rm CQ}_{d}$}
\begin{thm}[$d$-SS is hard for ${\rm CQ}_{d}$]
\label{thm:d-SS-is-hard-for-CQ-d}Let $(\mathcal{F},s)\sim\mathbb{O}_{{\rm SS}}(d,n)$,
i.e. let $\mathcal{F}$ be an oracle for a (uniformly) random $d$-Shuffled
Simon's problem of size $n$ and period $s$. Let $\mathcal{C}^{\mathcal{F}}$
be any ${\rm CQ}_{d}$ circuit (see \Defref{dCQ} and \Remref{oracleVersionsQNC-CQ-QC})
acting on $\mathcal{O}(n)$ qubits, with query access to $\mathcal{F}$.
Then $\Pr[s\leftarrow\mathcal{C}^{\mathcal{F}}]\le\negl$, i.e. the
probability that the algorithm finds the period is negligible. 
\end{thm}

\branchcolor{black}{To prove the theorem, we first describe the main argument while asserting
some intermediate results, which are proven separately next.} 

We begin with setting up the basic notation we need for the proof. 
\begin{itemize}
\item Denote the initial state by $\sigma_{0}$ which is initialised to
all zeros (without loss of generality). 
\item Recall from \Notaref{CompositionNotation} that a ${\rm CQ}_{d}$
circuit can be represented as $\mathcal{C}=\mathcal{C}_{\tilde{n}}\circ\cdots\mathcal{C}_{2}\circ\mathcal{C}_{1}$
where\footnote{We dropped $\mathcal{A}_{c,\tilde{n}+1}$ without loss of generality
as it can be absorbed into another $\mathcal{C}_{\tilde{n}+1}$.} $\tilde{n}\le\poly$. Here, we write $\mathcal{C}_{i}:=\vec{U}_{i}\circ\mathcal{A}_{c,i}$
where $\vec{U}_{i}$ contains $d$ layers of unitaries, followed by
a measurement. We drop the subscript $c$ (which stood for ``classical'')
from $\mathcal{A}_{c,i}$ for brevity. 
\item We abuse the notation slightly. We drop the distinction between oracles
and functions and work solely with the latter. Let $(f,s)\sim\mathbb{F}_{{\rm Simon}}(n)$,
$((f_{i})_{i=0}^{d},s)\sim\mathbb{F}_{{\rm Shuff}}(n,d,f)$ and $\mathcal{F}$
correspond to $(f_{i})_{i=0}^{d}$ (see \Defref{SimonsFunctionDistr}
and \Defref{dShuffler}, \Notaref{Abstract-d-Shuffler}). In the following,
we use $\tilde{n}$ sets of shadow oracles each set is denoted by
$\vec{\mathcal{G}}_{i}=(\mathcal{G}_{d,i},\dots\mathcal{G}_{1,i})$,
one set for each $\mathcal{C}_{i}$. 
\begin{itemize}
\item We write $\vec{U}_{i}^{\vec{\mathcal{G}}_{i}}$ to denote $\Pi_{i}\circ U_{d+1,i}\circ\mathcal{G}_{d,i}\circ U_{d,i}\circ\dots\mathcal{G}_{1,i}\circ U_{1,i}$.
We drop the ``$\circ$'' for brevity when no confusion arises.
\item These shadows are constructed using \Algref{SforCQd_using_dSS} where
the paths are specified as we outline the main argument.
\end{itemize}
\item Note that, for each $i$, after $\mathcal{C}_{i}$ the state is classical
allowing use to consider ``transcripts''. Parameters: Set $\delta=\Delta/\tilde{n}$,
$\gamma=2^{-m}$ where $\Delta>0$ is some small constant and $m$
is such that $m-\tilde{m}\ge\Omega(n)$, where $\tilde{m}$ is the
length of the ``advice'', i.e. number of bits $\mathcal{A}_{i}$
sends to $\vec{U}_{i}$.
\end{itemize}
To simplify the notation, we assume that for each query made by a
classical circuit $\mathcal{A}_{c,i}$ that results in a non-$\perp$
query, the circuit learns the entire path associated with that query.
This can only improve the success probability of the ${\rm CQ}_{d}$
circuit $\mathcal{C}$ and thus it suffices to upper bound this quantity
instead. 

\branchcolor{black}{\begin{proof}
The main argument consists of two steps as usual. The first is that
the output of any ${\rm CQ}_{d}$ circuit with access to $\mathcal{F}$
differs from that of the same circuit with access to shadow oracles
$\{\vec{\mathcal{G}}_{i}\}_{i}$, with negligible probability, i.e.
\begin{equation}
{\rm B}\left[\mathcal{A}_{\tilde{n}+1}^{\mathcal{F}}\vec{U}_{\tilde{n}}^{\mathcal{F}}\mathcal{A}_{\tilde{n}}^{\mathcal{F}}\dots\vec{U}_{1}^{\mathcal{F}}\mathcal{A}_{1}^{\mathcal{F}}(\sigma_{0}),\quad\mathcal{A}_{\tilde{n}+1}^{\mathcal{F}}\vec{U}_{\tilde{n}}^{\vec{\mathcal{G}}_{\tilde{n}}}\mathcal{A}_{\tilde{n}}^{\mathcal{F}}\dots\vec{U}_{1}^{\vec{\mathcal{G}}_{1}}\mathcal{A}_{1}^{\mathcal{F}}(\sigma_{0})\right]\le\poly\cdot2^{-n}\label{eq:TD_CQ_d_step1_main}
\end{equation}
given $d\le\poly$. The second is that no ${\rm CQ}_{d}$ circuit
with access to shadow oracles can find the period with non-negligible
probability, i.e.
\begin{equation}
\Pr[s\leftarrow\mathcal{A}_{\tilde{n}+1}^{\mathcal{F}}\vec{U}_{\tilde{n}}^{\vec{\mathcal{G}}_{\tilde{n}}}\mathcal{A}_{\tilde{n}}^{\mathcal{F}}\dots\vec{U}_{1}^{\vec{\mathcal{G}}_{1}}\mathcal{A}_{1}^{\mathcal{F}}(\sigma_{0})]\le\mathcal{O}(2^{-n}).\label{eq:shadowOnlyUseless_CQ_d_dSS}
\end{equation}

\textbf{Step One | Using shadow oracles causes negligible change in
output}\\
We outline the proof of step one. Using the hybrid argument, one can
bound the LHS of \Eqref{TD_CQ_d_step1_main} as 
\begin{align}
\le & \sum_{i=1}^{\tilde{n}}{\rm B}\Big[\underbrace{\mathcal{A}_{\tilde{n}+1}^{\mathcal{F}}\vec{U}_{\tilde{m}}^{\mathcal{F}}\mathcal{A}_{\tilde{m}}^{\mathcal{F}}\dots\vec{U}_{i+1}^{\mathcal{F}}\mathcal{A}_{i+1}^{\mathcal{F}}}_{\text{can be dropped \ensuremath{\because} monotonicity B}}\ \ \ \vec{U}_{i}^{\mathcal{F}}\mathcal{A}_{i}^{\mathcal{F}}\vec{U}_{i-1}^{\vec{\mathcal{G}}_{i}}\mathcal{A}_{i-1}^{\mathcal{F}}\dots\vec{U}_{1}^{\vec{\mathcal{G}}_{1}}\mathcal{A}_{1}^{\mathcal{F}},\nonumber \\
 & \quad\quad\quad\overbrace{\mathcal{A}_{\tilde{n}+1}^{\mathcal{F}}\vec{U}_{\tilde{m}}^{\mathcal{F}}\mathcal{A}_{\tilde{m}}^{\mathcal{F}}\dots\vec{U}_{i+1}^{\mathcal{F}}\mathcal{A}_{i+1}^{\mathcal{F}}}\ \ \ \vec{U}_{i}^{\vec{\mathcal{G}}_{i}}\mathcal{A}_{i}^{\mathcal{F}}\vec{U}_{i-1}^{\vec{\mathcal{G}}_{i}}\mathcal{A}_{i-1}^{\mathcal{F}}\dots\vec{U}_{1}^{\vec{\mathcal{G}}_{1}}\mathcal{A}_{1}^{\mathcal{F}}\Big]\nonumber \\
\le & \sum_{i=1}^{\tilde{n}}{\rm B}\left[\vec{U}_{i}^{\mathcal{F}}\mathcal{A}_{i}^{\mathcal{F}}\vec{U}_{i-1}^{\vec{\mathcal{G}}_{i}}\mathcal{A}_{i-1}^{\mathcal{F}}\dots\vec{U}_{1}^{\vec{\mathcal{G}}_{1}}\mathcal{A}_{1}^{\mathcal{F}},\quad\vec{U}_{i}^{\vec{\mathcal{G}}_{i}}\mathcal{A}_{i}^{\mathcal{F}}\vec{U}_{i-1}^{\vec{\mathcal{G}}_{i}}\mathcal{A}_{i-1}^{\mathcal{F}}\dots\vec{U}_{1}^{\vec{\mathcal{G}}_{1}}\mathcal{A}_{1}^{\mathcal{F}}\right]\label{eq:sum_of_Bs_CQ_d}
\end{align}
where we suppressed $\sigma_{0}$ for brevity. 

\textbf{The $i=1$ case}\\
Begin with $i=1$. Let $\mathcal{A}_{1}^{\f}(\sigma_{0})=:\sigma_{1}$.
Then, one has 
\begin{align*}
{\rm B}(\vec{U}_{1}^{\mathcal{F}}(\sigma_{1}),\vec{U}_{1}^{\vec{\g}_{1}}(\sigma_{1})) & \le{\rm B}(\mathcal{F}U_{d,1}\dots\mathcal{F}U_{1,1}(\sigma_{1}),\quad\g_{d,1}U_{d,1}\dots\g_{1,1}U_{1,1}(\sigma_{1}))\\
 & \le\sum_{j=1}^{d}{\rm B}(\mathcal{F}U_{j,1}\ \ \ \underbrace{\g_{j-1,1}U_{j-1,1}\dots\g_{1,1}U_{1,1}(\sigma_{1})}_{:=\rho_{j-1,1}}, & \text{hybrid argument}\\
 & \quad\quad\quad\g_{j,1}U_{j,1}\ \ \ \overbrace{\g_{j-1,1}U_{j-1,1}\dots\g_{1,1}U_{1,1}(\sigma_{1})}).
\end{align*}
Using the O2H lemma (see \Lemref{O2H}), one can bound each term in
the sum as 
\begin{equation}
{\rm B}(\mathcal{F}U_{j,1}\rho_{j-1,1},\g_{j,1}U_{j,1}\rho_{j-1,1})\le\sqrt{\Pr[{\rm find}:U_{j,1}^{\mathcal{F}\backslash\bar{S}_{j,1}},\rho_{j-1,1}]}\label{eq:CQ_d_FU1}
\end{equation}
where $\bar{S}_{j,1}$ is used to define the shadow oracle $\mathcal{G}_{j,1}$
and $\bar{S}_{j,1}$ is defined momentarily. 

Let $Y_{1}=\{\boldsymbol{y}_{k,1}\}_{k=1}^{|Y_{1}|}$ denote the paths
uncovered by $\mathcal{A}_{1}^{\mathcal{F}}$ and $\vec{I}_{1}=(I_{l,1})_{l=1}^{d}$
denote the queries which yielded $\perp$, given the paths uncovered
were $Y_{1}$. Note that these are random variables (even if we take
$\mathcal{A}_{1}$ to be deterministic) correlated with $\mathcal{F}$.
Define the transcript $T(\sigma_{1}):=(\vec{I}_{1},Y_{1})$. We generalise
these as we proceed. 

With this in place, define for each $j\in\{1,\dots d\}$, $\bar{S}_{j,1}$
to be the output of \Algref{SforCQd_using_dSS} with $j$, $Y_{1}$
and $((f_{i})_{i=0}^{d},s)$ as inputs. To bound \Eqref{CQ_d_FU1},
condition its RHS squared, as 
\begin{align*}
 & \le\sum_{\vec{I}_{1},Y_{1}}\Pr[\vec{I}_{1}\land Y_{1}]\Pr[{\rm find}:U_{j,1}^{\mathcal{F}\backslash\bar{S}_{j,1}},\rho_{j-1,1}|\vec{I}_{1}\land Y_{1}]\\
 & \le\poly\cdot2^{-n}\cdot\sum_{\vec{I}_{1},Y_{1}}\Pr[\vec{I}_{1}\land Y_{1}]=\poly\cdot2^{-n}
\end{align*}
where the second inequality follows from \Lemref{boundPfind} which
bounds the probability of finding in terms of landing queries on $\bar{S}_{j,1}$
which in turn is bounded by \Propref{x_in_S_dSS_shadow_withFixedSets}
with parameters\footnote{The $\delta$ here refers to that in the proposition and not the one
we defined; excuse the minor the notational conflict.} $\delta\leftarrow0$, $\vec{I}\leftarrow\vec{I_{1}}$ and $Y\leftarrow Y_{1}$.
Note that the conditions to apply \Lemref{boundPfind} (via \Corref{Conditionals})
are satisfied once we condition all the variables (as in \Corref{Conditionals})
to $\vec{I}_{1}\land Y_{1}$ and use \Propref{Shadow_dSS_withFixedSets}
to show they are uncorrelated.\footnote{Once $Y_{1}$ are fixed, this may be viewed as a special case of the
proof of \Thmref{QCd_hardness_dSerialSimons} where the $\perp$s
are known right at the start.}

\textbf{The $i=2$ case}\\
For $i=2$ may be proved by following the $i=1$ case, with one key
difference---the use of $\delta$ non-uniform distributions over
the $d$-Shuffler. Let $\sigma_{2}:=\mathcal{A}_{2}^{\mathcal{F}}\vec{U}_{1}^{\vec{\mathcal{G}_{1}}}\mathcal{A}_{1}^{\mathcal{F}}(\sigma_{0})$
and $\rho_{j-1,2}:=\mathcal{G}_{j-1,2}U_{j,2}\dots\mathcal{G}_{1,2}U_{1,2}(\sigma_{2})$.
Proceeding as before, one can write the $i=2$ term of \Eqref{sum_of_Bs_CQ_d}
as
\begin{equation}
{\rm B}(\vec{U}_{2}^{\mathcal{F}}(\sigma_{2}),\vec{U}_{2}^{\vec{\mathcal{G}}_{2}}(\sigma_{2}))\le\sum_{j=1}^{d}\sqrt{\Pr[{\rm find}:U_{j,2}^{\mathcal{F}\backslash\bar{S}_{j,2}},\rho_{j-1,2}]}\label{eq:CQ_d_FU2}
\end{equation}
where, again, $\bar{S}_{j,2}$ is used to define the shadow oracle
$\mathcal{G}_{j,2}$ and $\bar{S}_{j,2}$ is defined after the transcript
for $\sigma_{2}$, $T(\sigma_{2})$, is defined. 

Let $T(\sigma_{2}):=(\vec{I}_{2},Y_{2},S_{1},s_{1},\vec{I}_{1},Y_{1})$
where $Y_{1}$ and $\vec{I}_{1}$ were defined for the $i=1$ case.
Here $s_{1}$ is the output produced by the first ${\rm QNC}_{d}$
circuit, conditioned on $\vec{I}_{1}\land Y_{1}$ (i.e. the information
learnt by the classical algorithm $\mathcal{A}_{1}$ being $\vec{I}_{1}$
and $Y_{1}$). This next step is the key difference between the two
cases and to this end, it may be helpful to recall that the main random
variable being conditioned is the $d$-Shuffler, $\mathcal{F}$ (which
is abstractly referred to as $\Xi$). It may also be helpful to look
at \Exaref{conditioningPlusGivingAccess} before proceeding. We use
\Propref{main_delta_non_uniform_d_shuffler} with the parameters $\gamma=2^{-m}$,
$\delta=\Delta/\tilde{n}$ (recall), $\delta'=0$ and $\beta$ encoding
$\vec{I}_{1}$ and $Y_{1}$. To prove our result, one may treat the
parts fixed by this conditioning process (i.e. the parts fixed in
each term of \Eqref{main_delta_non_uniform_d_shuffler}) as a random
variable which the circuit can access (as illustrated in \Exaref{conditioningPlusGivingAccess}).
Let $S_{1}^{*}$ be\footnote{The star, as usual, indicates that the last coordinate is not specified.}
the aforementioned random variable (given $s_{1},\vec{I}_{1},Y_{1}$)
and note that its size $|S_{1}^{*}|$ is at most $2m/\delta\le\poly$.
Now $\mathcal{A}_{2}^{\mathcal{F}}$ takes as input $(S_{1}^{*},s_{1},\vec{I}_{1},Y_{1})$
so it can learn the last coordinate of the paths in $S_{1}^{*}$;
let $S_{1}$ denote these complete paths corresponding to $S_{1}^{*}$.
Suppose it further learns paths $Y_{2}$ (given the transcript so
far) and $\vec{I}_{2}$ (given the transcript and $Y_{2}$). These
completely specify the transcript $T(\sigma_{2})$. 

We may now define, for each $j\in\{1,\dots d\}$, $\bar{S}_{j,2}$
to be the output of \Algref{SforCQd_using_dSS} with $j$, $Y_{1}\cup S_{1}\cup Y_{2}$
and $((f_{i})_{i},s)$ as inputs. To bound \Eqref{CQ_d_FU2}, one
may condition each the square of each term of the RHS, $\Pr[{\rm find}:U_{j,2}^{\mathcal{F}\backslash\bar{S}_{j,2}},\rho_{j-1,2}]$,
as
\begin{align*}
 & =\sum_{s_{1},\vec{I}_{1},Y_{1}}\Pr[s_{1}]\Pr[\vec{I}_{1}]\Pr[Y_{1}]\Pr[{\rm find}:U_{j,2}^{\mathcal{F}\backslash\bar{S}_{j,2}},\rho_{j-1,2}|Y_{1}\land\vec{I}_{1}\land s_{1}]\\
 & \le\sum_{\stackrel{\vec{I}_{2},Y_{2},S_{1},\vec{I}_{1}Y_{1}}{s_{1}:\Pr[s_{1}\ge2^{-m}]}}{\rm Pr}_{\delta}[\vec{I}_{2}]{\rm Pr}_{\delta}[Y_{2}]\Pr[S_{1}]\Pr[s_{1}]\Pr[\vec{I}_{1}]\Pr[Y_{1}]{\rm Pr}_{\delta}[{\rm find}:U_{j,2}^{\mathcal{F}\backslash\bar{S}_{j,2}},\rho_{j-1,2}|T(\sigma_{2})]+2^{-(m-\tilde{m})}\\
 & \le2^{\delta}\cdot\poly\cdot2^{-n}\cdot\sum_{\stackrel{\vec{I}_{2},Y_{2},S_{1},\vec{I}_{1}Y_{1}}{s_{1}:\Pr[s_{1}\ge2^{-m}]}}{\rm Pr}_{\delta}[\vec{I}_{2}]{\rm Pr}_{\delta}[Y_{2}]\Pr[S_{1}]\Pr[s_{1}]\Pr[\vec{I}_{1}]\Pr[Y_{1}]\ \ +\ \ 2^{-(m-\tilde{m})}\le\negl
\end{align*}
where the first step follows from the rules of conditional probabilities.
The second step follows from an application of \Propref{main_delta_non_uniform_d_shuffler}
where by $\Pr_{\delta}$ we mean that the random variable $\mathcal{F}$
is, instead of being sampled from $\mathbb{F}_{{\rm Shuff}}$, is
sampled from $\mathbb{F}_{{\rm Shuff}}^{\delta}$; we needn't use
$\beta$ here as the conditioning is explicitly stated. The third
step, as in the $i=1$ case, follows from the application of \Lemref{boundPfind}
(via \Corref{Conditionals}) by observing that after conditioning,
the variables are uncorrelated (using \Propref{Shadow_dSS_withFixedSets})
and by using the bound asserted by \Propref{x_in_S_dSS_shadow_withFixedSets}
with parameters $\delta\leftarrow\delta$, $\vec{I}\leftarrow\vec{I}_{1}\cup\vec{I}_{2}$
and $Y\leftarrow Y_{1}\cup Y_{2}\cup S_{1}$. 

\label{exa:conditioningPlusGivingAccess}Let $\mathcal{C}$ be a ${\rm CQ}_{d}$
circuit (for concreteness) with query access to a $d$-Shuffler $\mathcal{F}$,
hiding the period of a random Simon's function. Let the period by
$s$ (a random variable). Then, 
\begin{align*}
\max_{\mathcal{C}}\Pr[s\leftarrow\mathcal{C}^{\mathcal{F}}] & =\max_{\mathcal{C}}\sum_{S}\Pr[S]\cdot\Pr[s\leftarrow\mathcal{C}^{\mathcal{F}|S}]\\
 & \le\max_{\mathcal{C}}\sum_{S}\Pr[S]\cdot\Pr[a\leftarrow\mathcal{C}^{\mathcal{F}|S}(S)].
\end{align*}

\textbf{The general $i\in\{1,\dots\tilde{n}\}$ case}\\
This is a straightforward generalisation of the $i=2$ case and hence
only key steps are outlined. Let $\sigma_{i}:=\mathcal{A}_{i}^{\mathcal{F}}\vec{U}_{i-1}^{\vec{\mathcal{G}}_{i-1}}\dots\mathcal{A}_{2}^{\mathcal{F}}\vec{U}_{1}^{\vec{\mathcal{G}}_{1}}\mathcal{A}_{1}^{\mathcal{F}}(\sigma_{0})$
and the transcript $T(\sigma_{i})=:(\vec{I}_{i},Y_{i},S_{i-1},s_{i-1},\vec{I}_{i-1},Y_{i-1},\dots S_{1},s_{1},\vec{I}_{1},Y_{1})$
where $\vec{I}_{j}$ encodes the locations which yielded $\perp$
when queried by $\mathcal{A}_{j}$ (given the transcript before that),
$Y_{j}$ encodes the paths uncovered by $\mathcal{A}_{j}$ (given
the transcript before), $s_{j}$ denotes the output of the $j$th
quantum circuit (given the transcript before that) and $S_{j}$ denotes
the paths uncovered using \Propref{main_delta_non_uniform_d_shuffler}
and allowing the circuit access to it (as in \Exaref{conditioningPlusGivingAccess}).
Let $\rho_{j-1,i}:=\mathcal{G}_{j-1,i}U_{j-1,i}\dots\mathcal{G}_{1,i}U_{1,i}(\sigma_{i})$.
The $i$th term in \Eqref{sum_of_Bs_CQ_d} can be expressed as 
\begin{equation}
{\rm B}(\vec{U}_{i}^{\mathcal{F}}(\sigma_{i}),\vec{U}_{i}^{\vec{\mathcal{G}}_{i}}(\sigma_{i}))\le\sum_{j=1}^{d}\sqrt{\Pr[{\rm find}:U_{j,i}^{\mathcal{F}\backslash\bar{S}_{j,i}},\rho_{j-1,i}]}\label{eq:CQ_d_FU_i}
\end{equation}
using \Lemref{O2H}, where $\vec{\mathcal{G}}_{j,i}$ is the shadow
oracle defined using $\bar{S}_{j,i}$ which in turn is defined as
the output of \Algref{SforCQd_using_dSS} with $j$, $Y_{1}\cup S_{1}\cup\dots Y_{i-1}\cup S_{i-1}\cup Y_{i}$
and $((f_{i})_{i},s)$ as inputs. The square of the $j$th term of
the RHS of \Eqref{CQ_d_FU_i}, $\Pr[{\rm find}:U_{j,i}^{\mathcal{F}\backslash\bar{S}_{j,i}},\rho_{j-1,i}]$,
can be bounded as 
\begin{align}
 & \le\sum_{\stackrel{\vec{I}_{i},Y_{i},S_{i-1},\dots S_{1},\vec{I}_{1},Y_{1}}{s_{i-1}:\Pr[s_{i-1}\ge2^{-m}],\dots s_{1}:\Pr[s_{1}\ge2^{-m}]}}{\rm Pr}_{i\cdot\delta}[\vec{I}_{i}]{\rm Pr}_{i\cdot\delta}[Y_{i}]{\rm Pr}_{(i-1)\cdot\delta}[S_{i-1}]\dots{\rm Pr}[\vec{I}_{1}]\Pr[Y_{1}]{\rm Pr}_{i\cdot\delta}[{\rm find}:U_{j,i}^{\mathcal{F}\backslash\bar{S}_{j,i}},\rho_{j-1,i}|T(\sigma_{i})]+i\cdot2^{-(m-\tilde{m})}\label{eq:BoundPr_find_U_ji}\\
 & \le2^{\Delta}\cdot\poly\cdot2^{-n}+i\cdot2^{-(m-\tilde{m})}\le\negl\label{eq:BoundPr_find_U_ji_negl}
\end{align}
where ${\rm Pr}_{i\cdot\delta}$ is evaluated with $\mathcal{F}$
sampled from $\mathbb{F}_{{\rm Shuff}}^{i\cdot\delta}$ instead of
$\mathbb{F}_{{\rm Shuff}}$. This follows from repeated applications
of \Propref{main_delta_non_uniform_d_shuffler} (with, for the $k$th
application, $\gamma=2^{-m}$, $\delta=\Delta/\tilde{n}$, $\delta'=(k-1)\cdot\delta$
and $\beta$ encoding $\vec{I}:=\vec{I}_{1}\cup\vec{I}_{2}\dots\cup\vec{I}_{k}$
and $Y:=Y_{1}\cup Y_{2}\cup\dots Y_{k}\cup S_{1}\cup S_{2}\dots S_{k-1}$).
After conditioning on the transcript, the variables can be shown to
be uncorrelated using \Propref{Shadow_dSS_withFixedSets} allowing
the application of \Lemref{boundPfind} (via \Corref{Conditionals}).
The ``$p$'' in \Lemref{boundPfind} can be bounded using \Propref{x_in_S_dSS_shadow_withFixedSets}
with parameters $\delta\leftarrow i\cdot\delta$, $\vec{I}\leftarrow\vec{I}_{1}\cup\dots\vec{I}_{i}$
and $Y\leftarrow Y_{1}\cup\dots Y_{i}\cup S_{1}\cup\dots S_{i-1}$. 

Since each term in \Eqref{sum_of_Bs_CQ_d} is bounded by $\poly\cdot2^{-n}$,
and there are at most polynomially many terms, one obtains the bound
asserted by \Eqref{TD_CQ_d_step1_main}.

\textbf{Step Two | The problem is hard using only shadow oracles}

It remains to show that \Eqref{shadowOnlyUseless_CQ_d_dSS} holds.
This is easily seen by observing that for the poly time classical
algorithm, $\mathcal{A}_{j}^{\mathcal{F}}$, having access to the
Simon's function directly or having access via the $d$-Shuffler,
is equivalent. Further, the $d$-depth quantum circuits $\vec{U}_{j}^{\vec{\mathcal{G}}_{j}}$
cannot reveal any more information about the Simon's function than
what the classical algorithm already possessed. This is because the
shadows do not contain any more information than that already possessed
by the classical algorithm. Thus, the success probability of the classical
algorithm is bounded by the probability of a ${\rm BPP}$ machine
solving the Simon's problem.
\end{proof}
}

\section{$d$-Shuffled Collisions-to-Simon's ($d$-SCS) Problem « Main Result
2\label{sec:ShuffledCollisionsToSimons}}

\branchcolor{black}{We describe a problem, the $d$-Shuffled Collisions-to-Simon's problem,
which a ${\rm QC}_{4}$ circuit can solve but no ${\rm CQ}_{d'}$
circuit for $d'\le d$. However, we show this using a non-standard
oracle model. We first describe this model.}

\subsection{Intrinsically Stochastic Oracle\label{subsec:Intrinsically-Stochastic-Oracle}}

\branchcolor{black}{We define an\emph{ Intrinsically Stochastic Oracle }to be a standard
oracle except that it is allowed to sample (and use) new instances
of a random variable, each time it is queried. We begin with an example.
Suppose $\mathcal{O}$ takes no input and produces a uniformly random
bit $b$. Two identical copies $\mathcal{O}_{1}$ and $\mathcal{O}_{2}$
of $\mathcal{O}$, when queried, have a probability $1/2$ of agreeing.
To contrast with the standard oracle model, let $c$ be a random bit
sampled uniformly at random. Let $\mathcal{Q}$ be a (standard) oracle
which outputs $c$. Multiple copies of $\mathcal{Q}$ would output
the same $c$.

Classically, the latter (i.e. the standard oracle) makes sense, because
one can pull out the randomness from all reasonable models of computations
into a random string specified at the beginning. Thus one can make
any implementation of an oracle repeatable, making the aforesaid notion
of intrinsically stochastic oracles questionable at best. Quantumly,
however, it is conceivable that a process produces superpositions
and a part of this is concealed from the circuit making the query;
for instance, perhaps it is somehow scrambled so that it is effectively
unavailable to the circuit. In particular, this can be used to simulate
the measurement of a quantum superposition which results in a random
outcome, each time the measurement is performed; even if the process
is repeated identically.

We should note that non-standard oracles have been used before to
prove separations between complexity classes. An example is the so-called
quantum oracle used to prove a separation between the classes $\textsf{{QMA}}$
and $\textsf{{QCMA}}$ \cite{aaronson2007quantum}.}

\begin{defn}[Intrinsically Stochastic Oracle]
 Let $X,Y$ be finite sets and let $\mathbb{F}_{Y}$ be some distribution
over $Y$. Let $g(x,y):X\times Y\to Z$ be a function. An \emph{intrinsically
stochastic oracle (ISO) $\mathcal{O}$ with respect to $\mathbb{F}_{Y}$},
corresponding to $g$ is defined by its action at each query: $\mathcal{O}$
samples $y\sim\mathbb{F}_{Y}$ and on input $\left|x\right\rangle \left|z\right\rangle $,
produces $\left|x\right\rangle \left|z\oplus g(x,y)\right\rangle $.
Its action on a superposition query is defined by linearity as in
the standard oracle model.
\end{defn}

\begin{rem}
Note that for superposition queries, the same $y$ is used for each
part constituting the superposition. A possible quantum realisation
of $\mathcal{O}$ could be the following: $\mathcal{O}\left|x\right\rangle _{Q}\left|z\right\rangle _{R}\left|0\right\rangle _{R'}$
produces $\sum_{y}\frac{1}{\sqrt{\Pr[y]}}\left|x\right\rangle _{Q}\left|z\oplus g(x,y)\right\rangle _{R}\left|y\right\rangle _{R'}$
and scrambling the last register, $R'$, effectively tracing it out
and producing a mixed state. Or, one could imagine that $\mathcal{O}\left|x\right\rangle _{Q}\left|z\right\rangle _{R}$
produces the same state except that the oracle holds $R'$. This can
arise naturally in an interactive setting. 
\end{rem}

\subsection{Oracles and Distributions | The $d$-SCS Problem}

\branchcolor{black}{When access to the Simon's function is restricted via a $d$-Shuffler,
we saw that the problem is hard for both ${\rm CQ}_{d}$ and ${\rm QC}_{d}$.
Here, we consider a variant of this restriction. The basic idea is
that one is given access to a $2$-to-$1$ function, $f$ (not necessarily
a Simon's function). For each colliding pair, one can associate exactly
one colliding pair of some Simon's function $g$. The objective is
to find the period of $g$ but access to $g$ is not provided directly.
Instead, one is given restricted access to a mapping which takes colliding
pairs of $f$ to colliding pairs of $g$. The details of this restriction,
which use the $d$-Shuffler, determine the hardness of the problem.
In one case, we essentially reduce to $d$-SS while in another, we
obtain $d$-Shuffled Collisions-to-Simon's ($d$-SCS).

We proceed more precisely now, beginning with a few definitions. 

}
\begin{defn}[Distribution for $2\to1$ functions]
 Let $F$ be a set of all functions $f:\{0,1\}^{n}\to\{0,1\}^{n}$
such that for each $x_{0}\in\{0,1\}^{n}$, there is exactly one $x_{1}\in\{0,1\}^{n}$
(distinct from $x_{0}$) such that $f(x_{0})=f(x_{1})$. Define the
\emph{uniform distribution over $2\to1$ functions}, $\mathbb{F}_{2\to1}(n)$,
to be the uniformly random distribution over $F$ and denote the corresponding
oracle distribution by $\mathbb{O}_{2\to1}(n)$. \label{def:TwoToOne}
\end{defn}

\begin{defn}[Collisions-to-Simons (CS) map]
 Consider any $2$-to-$1$ function $f:\{0,1\}^{n}\to\{0,1\}^{n}$,
and any Simon's function $g:\{0,1\}^{n}\to\{0,1\}^{n}$. Suppose by
$f^{-1}(y)$ we denote pre-images under $f$ listed in the ascending
order (and similarly for $g^{-1}(z)$). We call a bijective map $p:\{0,1\}^{n}\to\{0,1\}^{n}$
a\emph{ Collisions to Simons map} from $f$ to $g$, if for all $k\in\{1,\dots,2^{n}/2\}$,
it maps the pre-images $f^{-1}(y)$ to the pre-images $g^{-1}(z)$
where $y$ is the $k$th largest value of $f$ and $z$ the $k$th
largest value of $g$. Denote the inverse map by $p_{{\rm inv}}$.
\end{defn}

\begin{rem}
Note that $f(x_{0})$ may not equal $g(p(x_{0}))$.
\end{rem}

\branchcolor{black}{If one were to give access to $p$ and $p^{-1}$ only via a $d$-Shuffler,
this problem would essentially reduce to $d$-SS because without access
to $p$, the Simon's function $g$ is effectively hidden and finding
its period with non-vanishing probability is impossible using both
${\rm QC}_{d}$ and ${\rm CQ}_{d}$ circuits.

We therefore consider the following restriction. One is given access
to a $2$-to-$1$ function $f:\{0,1\}^{n}\to\{0,1\}^{n}$ via a stochastic
oracle which allows a quantum circuit to hold two colliding pre-images
in superposition. Further, instead of the CS map $p$ (and $p_{{\rm inv}}$),
one is given access to an ``encrypted'' CS map $p'$ (and $p_{{\rm inv}}^{\prime}$).
Further, one can access a random permutation $h:\{0,1\}^{n}\to\{0,1\}^{n}$
via a $d$-Shuffler. The role of $p'$ is to allow evaluation of $p(x)$
only when $h(f(x))$ is also accompanied with the input. The idea
is that while holding a superposition of images, one needs a $(d+1)$-depth
computation to evaluate $h$ which a ${\rm QC}_{d}$ circuit can but
a ${\rm CQ}_{d}$ circuit cannot. We now describe the problem formally.}
\begin{defn}[$d$-SCS distribution]
\label{def:dSCS_distr}Define the \emph{$d$-Shuffled Collisions-to-Simons
Function} \emph{distribution} $\mathbb{F}_{{\rm SCS}}(n)$ by its
sampling procedure.
\begin{itemize}
\item \emph{The two-to-one function, Simons function and the CS map:} Sample
a random $2\to1$ function, $f\sim\mathbb{F}_{2\to1}(n)$ (see \Defref{TwoToOne})
and a random Simon's function, $(g,s)\sim\mathbb{F}_{{\rm Simons}}(n)$
(see \Defref{SimonsFunctionDistr}). Let $p$ be the Collisions-to-Simons
map for $f$ and $g$.
\item \emph{Stochastic Oracle for $f$:} Let $Y=\{y:f(x)=y,x\in\{0,1\}^{n}\}$
and denote by $\mathbb{F}_{Y}$ the uniform distribution over $Y$.
Let $b\in\{0,1\}$. Define $\mathcal{S}$ to be a stochastic oracle
wrt $\mathbb{F}_{Y}$ corresponding to the function $g(b)=(f^{-1}(y)[b],y)$,
i.e. $\mathcal{S}\left(\left|0\right\rangle _{Q}+\left|1\right\rangle _{Q}\right)\left|0\right\rangle _{RR'}=\left(\left|0\right\rangle _{Q}\left|x_{0}\right\rangle _{R}+\left|1\right\rangle _{Q}\left|x_{1}\right\rangle _{R}\right)\left|y\right\rangle _{R'}$
where $f^{-1}(y)=(x_{0},x_{1})$ and $y$ is stochastic. 
\item \emph{Random permutation hidden in a $d$-Shuffler:} Sample $h\sim\mathbb{F}_{{\rm R}}(n)$
from the uniformly random distribution over one-to-one functions (see
\Defref{RandomOneOneFunctions}) and let $\Xi\sim\mathbb{F}_{{\rm Shuff}}(d,n,h)$
be a $d$-Shuffler (see \Defref{dShuffler}) encoding $h$.
\item \emph{Encrypted CS map:} Define $p':\{0,1\}^{2n}\to\{0,1\}^{n}\cup\{\perp\}$
to be $p'(h(f(x)),x)=p(x)$ when the input is of the form $(h(f(x)),x)$
and $\perp$ otherwise. Similarly, define $p_{{\rm inv}}^{\prime}:\{0,1\}^{2n}\to\{0,1\}^{n}\cup\{\perp\}$
to be $p'_{{\rm inv}}(h(f(x)),p(x))=p_{{\rm inv}}(p(x))=x$ when the
input is of the form $(h(f(x)),p(x))$ and $\perp$ otherwise.
\end{itemize}
Return $(\mathcal{S},\Xi,p',p'_{{\rm inv}},s)$ when $\mathbb{F}_{{\rm SCS}}(n)$
is sampled.
\end{defn}

\begin{defn}[The $d$-SCS problem]
 \label{def:dSCS_problem}The $d$-Shuffled Collisions to Simons
problem is defined as follows. Let $(\mathcal{S},\Xi,p',p'_{{\rm inv}},s)\sim\mathbb{F}_{{\rm SCS}}(n)$.
Given oracle access to $\mathcal{S},\Xi,p'$ and $p'_{{\rm inv}}$,
find $s$.
\end{defn}

\subsection{Depth Upper Bounds for $d$-SCS}

\subsubsection{${\rm QC}_{4}$ can solve the $d$-SCS problem}

\branchcolor{black}{As we asserted, a ${\rm QC}_{d}$ circuit can solve $d$-SCS. We prove
a stronger statement below.}
\begin{prop}
The $d$-Shuffled Collisions-to-Simon's problem can be solved using
${\rm QC}_{4}$. \label{prop:dSCS_usingQC-4}
\end{prop}

\begin{proof}
Run polynomially many copies of the following parallelly.
\begin{enumerate}
\item Quantumly, query $\mathcal{S}$ on a superposition $\left(\left|0\right\rangle +\left|1\right\rangle \right)_{Q}$
(neglecting normalisation) to obtain $(\left|0\right\rangle _{Q}\left|x_{0}\right\rangle _{R}+\left|1\right\rangle _{Q}\left|x_{1}\right\rangle _{R})\left|y=f(x_{0})=f(x_{1})\right\rangle _{R'}$
for some randomly chosen $y$ in the range of $f$. Measure $R'$
to learn $y$.
\item Classically, compute, using the $d$-Shuffler $\Xi$, the function
$h(y)$.
\item Quantumly, use the encrypted CS map $p'$ with $h(y)$ to get
\[
\left|0\right\rangle \left|x_{0}\right\rangle \left|p(x_{0})\right\rangle +\left|1\right\rangle \left|x_{1}\right\rangle \left|p(x_{1})\right\rangle 
\]
where note that $\{p(x_{0}),p(x_{1})\}=g^{-1}(y')$ for some $y'$
(or equivalently, it holds that $g(p(x_{0}))=g(p(x_{1}))$).
\item Use $p'_{{\rm inv}}$ with $h(y)$ to erase $x_{0}$ and $x_{1}$.
We are left with $\left|0\right\rangle \left|p(x_{0})\right\rangle +\left|1\right\rangle \left|p(x_{1})\right\rangle $.
\end{enumerate}
Proceed as in Simon's algorithm (Hadamard and measure; solve the equations
to obtain $p(x_{0})\oplus p(x_{1})=s$ (since these are preimages
of a collision in $g$ which is a Simon's function)).
\end{proof}

\subsubsection{${\rm CQ}_{d+6}$ can also solve the $d$-SCS problem}

\branchcolor{black}{It is easy to see that a ${\rm CQ}_{d+6}$ circuit can also solve
$d$-SCS by running the $d$-Shuffler in the algorithm above using
$d+1$ quantum depth. Note that the upper bound for $d$-SCS, i.e.
$d+6$, is tighter than that for $d$-SS, i.e. $2d+1$. It remains
to establish a lower bound, i.e. no ${\rm CQ}_{d}$ circuit can solve
$d$-SCS with non-vanishing probability.}

\subsection{Depth Lower Bounds for $d$-SCS}

\subsubsection{$d$-SCS is hard for ${\rm QNC}_{d}$}
\begin{thm}
\label{thm:d-SCS-is-hard-QNC-d}Let $(\mathcal{S},\Xi,p',p'_{{\rm inv}},s)\sim\mathbb{F}_{{\rm SCS}}(n)$
and let $\mathcal{O}$ denote the oracles $\mathcal{S},\Xi,p',p'_{{\rm inv}}$
collectively. Denote an arbitrary ${\rm QNC}_{d}$ circuit with oracle
access to $\mathcal{O}$, (followed by classical post-processing)
by $\mathcal{A}^{\mathcal{O}}$. Then, $\Pr[s\leftarrow\mathcal{A}^{\mathcal{O}}(\rho_{0})]\le\negl$
where $\rho_{0}$ is some fixed (wrt $\mathcal{O}$) initial state. 
\end{thm}

\branchcolor{black}{\begin{proof}
As we saw, for instance in the proof of \Thmref{d-SS-is-hard-for-QNC-d},
that the function $h$ hidden by a $d$-Shuffler is essentially inaccessible
to ${\rm QNC}_{d}$ circuits (i.e. in any ${\rm QNC}_{d}$ circuit,
one can replace the oracle with a shadow oracle (which contains no
information about $h$) and the output distributions change at most
negligibly). We assume this part of the analysis has already been
performed. More precisely, we only show\footnote{To be exact, we should use shadows for $\Xi$, which means the circuit
may learn something about the ``paths'' inside the $d$-Shuffler
(but not the value of $h$). As will be evident, these do not influence
the analysis for ${\rm QNC}_{d}$. For ${\rm QC}_{d}$, we take this
into account.} 
\[
\left|\Pr[s\leftarrow\mathcal{A}^{\mathcal{M}}(\rho_{0})]\right|\le\negl
\]
where $\mathcal{M}$ denotes the oracles $\mathcal{S},p',p'_{{\rm inv}}$.
We write the circuits more explicitly as $\mathcal{A}^{\mathcal{M}}:=\Pi\circ U_{d+1}\circ\mathcal{M}\circ U_{d}\dots\mathcal{M}\circ U_{1}$
and $\mathcal{A}^{\mathcal{N}}:=\Pi\circ U_{d+1}\circ\mathcal{N}\circ U_{d}\dots\mathcal{N}\circ U_{1}$
where $\mathcal{N}$ denotes the oracles $\mathcal{S},p'',p''_{{\rm inv}}$
where $p'',p''_{{\rm inv}}$ are defined as $p',p'_{{\rm inv}}$ except
that they always output $\perp$. One can apply the hybrid argument
as before to obtain 
\begin{align*}
\left|\Pr[q\leftarrow\mathcal{A}^{\mathcal{M}}(\rho_{0})]-\Pr[q\leftarrow\mathcal{A}^{\mathcal{N}}(\rho_{0})]\right| & \le B(\mathcal{M}U_{d}\dots\mathcal{M}U_{1}(\rho_{0}),\mathcal{N}U_{d}\dots\mathcal{N}U_{1}(\rho_{0}))\\
 & \le\sum_{i=1}^{d}B(\mathcal{M}U_{i}(\rho_{i-1}),\mathcal{N}U_{i}(\rho_{i-1}))\\
 & \le\sum_{i=1}^{d}\sqrt{2\Pr[{\rm find}:U_{i}^{\mathcal{M}\backslash\bar{X}},\rho_{i-1}]}
\end{align*}
where $\rho_{i-1}=\mathcal{N}U_{i-1}\dots\mathcal{N}U_{1}(\rho_{0})$
and $\bar{X}=(\emptyset,X,X_{{\rm inv}})$ where $X,X_{{\rm inv}}\subseteq\{0,1\}^{2n}$
is the non-trivial domain of $p',p'_{{\rm inv}}$ respectively, i.e.
for all $x\in X$, $p'(x)\neq\perp$ and similarly for all $x\in X_{{\rm inv}}$,
$p'_{{\rm inv}}(x)\neq\perp$. Since $\rho_{i-1}$ contains no information
about $h$, one can bound $\Pr[{\rm find}:U_{i}^{\mathcal{M}\backslash\bar{X}},\rho_{i-1}]$
using \Lemref{boundPfind} by $\poly\cdot2^{-n}$. Further, $\Pr[s\leftarrow\mathcal{A}^{\mathcal{N}}(\rho_{0})]\le\negl$
since $\mathcal{N}$ contains no information about the period $s$,
so one cannot do better than guessing. Together, these yield the asserted
result. 
\end{proof}
}

\subsubsection{$d$-SCS is hard for ${\rm CQ}_{d}$}

\branchcolor{black}{We now show that ${\rm CQ}_{d}$ circuits solve $d$-SCS with at most
negligible probability. }

\begin{thm}
\label{thm:dSCS_is_CQ_d_hard}Let $(\mathcal{S},\Xi,p',p'_{{\rm inv}},s)\sim\mathbb{F}_{{\rm SCS}}(n)$
and let $\mathcal{O}$ denote the oracles $\mathcal{S},\Xi,p',p'_{{\rm inv}}$.
Denote an arbitrary ${\rm CQ}_{d}$ circuit with oracle access to
$\mathcal{O}$ by $\mathcal{C}^{\mathcal{O}}$. Then, $\Pr[s\leftarrow\mathcal{C}^{\mathcal{O}}(\rho_{0})]\le\negl$
where $\rho_{0}$ is some fixed (wrt $\mathcal{O}$) initial state. 
\end{thm}

We use the proof of \Thmref{d-SS-is-hard-for-CQ-d} as a template.
\begin{itemize}
\item We again write the ${\rm CQ}_{d}$ circuit as $\mathcal{C}=\mathcal{C}_{\tilde{n}}\circ\dots\circ\mathcal{C}_{1}$
where $\tilde{n}\le\poly$ and $\mathcal{C}_{i}:=\vec{U}_{i}\circ\mathcal{A}_{i}$
where $\vec{U}_{i}$ contains $d$ layers of unitaries, followed by
a measurement and $\mathcal{A}_{i}$ denotes a poly time classical
computation.
\item Denote\footnote{Just notation; to preserve similarity with \Thmref{d-SS-is-hard-for-CQ-d}.}
by $\mathcal{F}$ the $d$-Shuffler $\Xi$. Denote by $\mathcal{M}$
the oracles corresponding to $p',p'_{{\rm inv}}$. We write $\mathcal{C}^{\mathcal{F},\mathcal{M}}$
to denote a ${\rm CQ}_{d}$ circuit with oracle access to both and
to $\mathcal{S}$ which we do not write explicitly in this section. 
\item For each $\mathcal{C}_{i}$ (which together constitute $\mathcal{C}$),
the shadow oracles are defined differently.
\begin{itemize}
\item Denote by $\mathcal{N}_{i}$ the oracle corresponding to the shadows
of $p'$ and $p'_{{\rm inv}}$, with respect to $\bar{X}_{i}$ which
is defined later.
\item Denote by $\vec{\mathcal{G}}_{i}=(\mathcal{G}_{d,i},\dots\mathcal{G}_{1,i})$
a tuple of shadow oracles of $\mathcal{F}$ with respect to $\bar{S}_{i}$
which is also defined later.
\item Write $\vec{U}_{i}^{\vec{\mathcal{G}}_{i},\mathcal{N}_{i}}$ to denote
$\Pi_{i}\circ U_{d+1,i}\circ(\mathcal{G}_{d,i},\mathcal{N}_{i})\circ U_{d,i}\circ\dots(\mathcal{G}_{1,i},\mathcal{N}_{i})\circ U_{1,i}$
where $(\mathcal{G}_{k,i},\mathcal{N}_{i})$ may be viewed as a single
oracle constituted by $\mathcal{G}_{k,i}$ and $\mathcal{N}_{i}$
(much like the sub-oracles in the $d$-Shuffler or $d$-Serial Simons\footnote{More explicitly, we mean that one may parallelly apply polynomially
many copies of $\mathcal{G}_{k,i}$ and parallelly, polynomially many
copies of $\mathcal{N}_{i}$.}). 
\item The shadows $\vec{\mathcal{G}}_{i}$ are constructed, as before, using
\Algref{SforCQd_using_dSS}, but the way the paths are specified is
slightly different. The details appear in the proof.
\end{itemize}
\item The parameters, $\delta,\Delta,m,\tilde{m}$ are the same as before.
\item When we write $\Pr[{\rm find}:U^{\mathcal{F};\mathcal{M}\backslash X},\rho]$
or $\Pr[{\rm find}:U^{\mathcal{M};\mathcal{F}\backslash\bar{S}},\rho]$,
we allow $U$ unrestricted access to the oracles before the semi-colon
``;'' and understand the rest in the usual $\Pr[{\rm find}:\dots]$
notation (see \Defref{prFind}).
\end{itemize}
\begin{proof}
We first show that replacing $\mathcal{F}$ and $\mathcal{M}$ with
their shadows makes almost no difference, i.e.
\begin{equation}
{\rm B}\left[\mathcal{A}_{\tilde{n}+1}^{\mathcal{F},\mathcal{M}}\vec{U}_{\tilde{n}}^{\mathcal{F},\mathcal{M}}\dots\vec{U}_{1}^{\mathcal{F},\mathcal{M}}\mathcal{A}_{1}^{\mathcal{F},\mathcal{M}}(\rho_{0}),\quad\mathcal{A}_{\tilde{n}+1}^{\mathcal{F},\mathcal{M}}\vec{U}_{\tilde{n}}^{\mathcal{\vec{G}}_{\tilde{n}},\mathcal{N}_{\tilde{n}}}\dots\vec{U}_{1}^{\vec{\mathcal{G}}_{1},\mathcal{N}_{1}}\mathcal{A}_{1}^{\mathcal{F},\mathcal{M}}(\rho_{0})\right]\le\negl\label{eq:StepOne_again-1}
\end{equation}
given $d\le\poly$ and then we show that with the shadows, no ${\rm CQ}_{d}$
circuit can solve the problem with non-negligible probability, i.e.
\begin{equation}
\Pr\left[s\leftarrow\mathcal{A}_{\tilde{n}+1}^{\mathcal{F},\mathcal{M}}\vec{U}_{\tilde{n}}^{\mathcal{\vec{G}}_{\tilde{n}},\mathcal{N}_{\tilde{n}}}\dots\vec{U}_{1}^{\vec{\mathcal{G}}_{1},\mathcal{N}_{1}}\mathcal{A}_{1}^{\mathcal{F},\mathcal{M}}(\rho_{0})\right]\le\negl.\label{eq:StepTwo_again-1}
\end{equation}

\textbf{Step one}

Let $\sigma_{i}:=\mathcal{A}_{i}^{\mathcal{F},\mathcal{M}}\vec{U}_{i-1}^{\vec{\mathcal{G}}_{i-1},\mathcal{N}_{i-1}}\dots\mathcal{A}_{2}^{\mathcal{F},\mathcal{M}}\vec{U}_{1}^{\vec{\mathcal{G}}_{1},\mathcal{N}_{1}}\mathcal{A}_{1}^{\mathcal{F},\mathcal{M}}(\sigma_{0})$
and the transcript 
\[
T(\sigma_{i}):=(\vec{I}_{i},H_{i,}I_{i},R_{i},\ \ S_{i-1},s_{i-1},\vec{I}_{i-1},H_{i-1},I_{i-1},R_{i-1}\quad\dots\quad S_{1},s_{1},\vec{I}_{1},H_{1},I_{1},R_{1})
\]
 where $R_{j}$ encodes the tuple $(x,p(x),p_{{\rm inv}}(x))$ obtained
by $\mathcal{A}_{j}$, $I_{j}$ encodes the locations of $p$ and
$p_{{\rm inv}}$ which output $\perp$ when queried by $\mathcal{A}_{j}$
(given the transcript before that), $H_{j}$ encodes the paths of
$\mathcal{F}$ uncovered by $\mathcal{A}_{j}$ (given the transcript
before), $\vec{I}_{j}$ encodes the locations in sub-oracles of $\mathcal{F}$
which yielded $\perp$ when queried by $\mathcal{A}_{j}$ (given the
transcript before), $s_{j}$ denotes the output of the $j$th quantum
circuit (given the transcript before) and $S_{j}$ denotes the paths
uncovered by the ``sampling argument'' and allowing the circuit
to access it (see \Propref{main_delta_non_uniform_d_shuffler} and
\Exaref{conditioningPlusGivingAccess}).\footnote{We take conditionals only to make the later argument easier; it doesn't
reflect or require that certain variables are known before the other.} 

We make the following assumptions about the transcript, without loss
of generality because these only make it easier for the circuit to
solve the problem. We assume $s_{j}$ contains all the stochastic
values $y$ obtained by the $j$th quantum part of the circuit\footnote{We also assume that the quantum circuit learns all tuples $(x,f(x)=y)$
learnt by the classical circuit preceding it.} by querying $\mathcal{S}$. We assume that $\mathcal{A}_{j}$ learns
the paths, $H_{j}$, in $\mathcal{F}$ corresponding to all the stochastic
values $y$ returned to it by $\mathcal{S}$ and those contained in
$s_{j-1}$ (it only possibly makes $\mathcal{A}$ better at solving
the problem, so this is without loss of generality). Finally, we assume
that if for some $j$, $R_{j}$ contains $x_{0},x_{1}$ which are
collisions in $f$ (i.e. $y=f(x_{0})=f(x_{1})$) and the path corresponding
to $y$ is in $H_{j}$, then for all $j'\ge j$, $R_{j'}$ is ``collision-complete''---i.e.
for each $x$ in $R_{j'}$ such that the path corresponding to $y=f(x)$
is $H_{j}$, it also contains $x'$ such that $f(x')=y$. Again, this
only makes it easier for the circuit to solve the problem.\footnote{We do this so that if the classical algorithm reveals information
about the period $s$, then for values of $p'$ it can access (because
it learns $h(y)$), images of both $x_{0}$ and $x_{1}$ are exposed
in the shadows of $p'$ where $x_{0}$ and $x_{1}$ are such that
$f(x_{0})=f(x_{1})=y$). Otherwise the subsequent quantum part would
be able to distinguish between a shadow of $p'$ and $p'$ by trying
$x_{0}\oplus s$ for instance with the same $y$. Note that the classical
algorithm can also not evaluate $p$ (and similarly $p_{{\rm inv}}$)
at arbitrary $x$s directly because it needs $h(f(x))$ and for this
it needs $f(x)$ which can only be accessed via the stochastic oracle
$\mathcal{S}$ but that in turn does not take $x$ as an input; $\mathcal{S}$
simply outputs a stochastic $y$ and $x\in f^{-1}(y)$. Thus the probability
that a classical algorithm learns a colliding pair $x_{0}$, $x_{1}$
is anyway negligible. We account for this possibility here, regardless. } 

Let $\rho'_{j-1,i}:=(\mathcal{G}_{j-1},\mathcal{M})\circ U_{j-1,i}\dots(\mathcal{G}_{1,i},\mathcal{M})\circ U_{1,i}(\sigma_{i})$
and $\rho_{j-1,i}:=(\mathcal{G}_{j-1,i},\mathcal{N}_{i})\circ U_{j-1,i}\dots(\mathcal{G}_{1,i},\mathcal{N}_{i})\circ U_{1,i}(\sigma_{i})$.
Following the proof of \Thmref{d-SS-is-hard-for-CQ-d}, we note that
to bound the LHS of \Eqref{StepOne_again-1}, it suffices to bound
the analogue of \Eqref{CQ_d_FU_i}, i.e. 
\begin{align}
{\rm B}(\vec{U}_{i}^{\mathcal{F},\mathcal{M}}(\sigma_{i}),\vec{U}_{i}^{\vec{\mathcal{G}}_{i},\mathcal{N}_{i}}(\sigma_{i})) & \le{\rm B}(\vec{U}_{i}^{\mathcal{F},\mathcal{M}}(\sigma_{i}),\vec{U}_{i}^{\vec{\mathcal{G}}_{i},\mathcal{M}}(\sigma_{i}))+{\rm B}(\vec{U}_{i}^{\mathcal{\vec{G}}_{i},\mathcal{M}}(\sigma_{i}),\vec{U}_{i}^{\vec{\mathcal{G}}_{i},\mathcal{N}_{i}}(\sigma_{i}))\label{eq:CQ_d_FU_i_dSCS}\\
 & \le\sum_{j=1}^{d}\Pr\left[{\rm find}:U^{\mathcal{M};\mathcal{F}\backslash\bar{S}_{j,i}},\rho_{j-1,i}'\right]+\sum_{j=1}^{d}\Pr\left[{\rm find}:U^{\vec{\mathcal{G}}_{i};\mathcal{M}\backslash\bar{X}_{i}},\rho_{j-1,i}\right]\label{eq:CQ_d_FU_i_dSCS-1}
\end{align}
where $\mathcal{G}_{j,i}$ is the shadow oracle with respect to $\bar{S}_{j,i}$
which in turn is defined the output of \Algref{SforCQd_using_dSS}
with $j$, $H_{1}\cup S_{1}\cup\dots H_{i-1}\cup S_{i-1}\cup H_{i}$
and $((f_{i})_{i},s)$ as inputs (recall that $\Xi$ was just an abstract
symbol for a $d$-Shuffler which can be used to specify $(f_{i})_{i}$
explicitly). We define $\bar{X}_{i}$ so that the shadows of $\mathcal{M}$
output $\perp$ everywhere except for the locations in $H$, where
they behave as $\mathcal{M}$. Let $X=\{(x,p(x),h(f(x)):x\in\{0,1\}^{n}\}$
and $X_{{\rm inv}}=\{(p(x),x,h(f(x)):x\in\{0,1\}^{n}\}$. Define $\bar{X}_{i}=(X\backslash\{(x,p(x),h(f(x)):(x,p(x),p_{{\rm inv}}(x))\in R_{j}\},X_{{\rm inv}}\backslash\{(p(x),x,h(f(x)):(x,p(x),p_{{\rm inv}}(x))\in R_{j}\})$.
Recall $\mathcal{N}$ is the shadow of $\mathcal{M}$ with respect
to $\bar{X}_{i}$. One can now proceed as we did in the proof of \Thmref{d-SS-is-hard-for-CQ-d}
and use \Eqref{BoundPr_find_U_ji_negl} to conclude that the first
term in \Eqref{CQ_d_FU_i_dSCS-1} is negligible. The second term is
also negligible using \Eqref{BoundPr_find_U_ji_negl} and noting that
one needs $h(f(x))$ in the third argument of $p'$ and $p'_{{\rm inv}}$
to obtain a non-$\perp$ response. This, for any $y$ which is not
already known, can be at best guessed with negligible probability
(for $y$s which are already known, the value is already exposed in
the shadow). That is because neither $\vec{\mathcal{G}}_{i}$ nor
$\rho_{j-1,i}$ contain any information about $h$ at the new $y$
values (other than the set of locations where $h$ outputs $\perp$s
which rule out polynomially many locations; one can proceed as in
the proof \Thmref{QCd_hardness_dSerialSimons}). Therefore the probability
of ``find'' is negligible. 

\textbf{Step Two}

To see that \Eqref{StepTwo_again-1} holds, we first define the event
$y$-distinct as follows: 
\begin{itemize}
\item all $y$s returned by $\mathcal{S}$ are distinct and
\item if $\mathcal{S}$ is queried after $h$ has been evaluated\footnote{This will include, and therefore account for, the values in $S_{i}$
(as defined in the transcript) exposed by the sampling argument.} for values in $Y$, then the $y$ returned by $\mathcal{S}$ is not
in $Y$.
\end{itemize}
This event happens with overwhelming (i.e. $1-\negl$) probability
because $y$ is stochastically chosen (so doesn't depend on any other
random variable) and all excluded $y$s constitute a set of size at
most polynomial. Conditioned on this event, we make the following
observations about the classical and quantum parts of an execution
of the ${\rm CQ}_{d}$ circuit $\mathcal{A}_{\tilde{n}+1}^{\mathcal{F},\mathcal{M}}\vec{U}_{\tilde{n}}^{\mathcal{\vec{G}}_{\tilde{n}},\mathcal{N}_{\tilde{n}}}\dots\vec{U}_{1}^{\vec{\mathcal{G}}_{1},\mathcal{N}_{1}}\mathcal{A}_{1}^{\mathcal{F},\mathcal{M}}(\rho_{0})$.
\begin{itemize}
\item Classical part.
\begin{itemize}
\item Observe that classical queries to $\mathcal{S}$, yield non-colliding
$x$s. Assume that the quantum advice also yields non-colliding $x$s
and $y=f(x)$ it learnt from $\mathcal{S}$. Only for these $x$s
can the classical algorithm learn $p(x)$ with non-negligible probability.
Since these are non-colliding, $p(x)$ contains no information about
the period $s$. Thus, finding collisions are necessary for finding
$s$ with non-negligible probability.
\item Note also that the output, $c$, of the classical part cannot contain
any collisions (assuming the quantum output contained no collisions)
with non-negligible probability. This is because the quantum advice
contains no information about $f$ at points other than those $x$s
(because the quantum algorithm only learnt $f(x)$ at stochastically
chosen locations). 
\end{itemize}
\item Quantum part.
\begin{itemize}
\item Suppose the classical input revealed no collision. The only way to
learn a collision is by querying $\mathcal{S}$. Since we conditioned
on the $y$-distinct event, the $y$s that the quantum algorithm learns,
at those values $\mathcal{N}$ contains $\perp$s. Thus it could not
have learnt $s$ from these $y$s. It can give some output $q$ to
the next classical circuit but since finding collisions for a random
$2\to1$ function is hard for an efficient quantum circuit (with direct
oracle access \cite{Aaronson2001,Shi2001}; here the access is even
more restricted), the output, $q$, of the quantum part cannot contain
any collisions (with non-negligible probability). 
\end{itemize}
\end{itemize}
One can apply the two steps above repeatedly, starting with $\mathcal{A}_{1}^{\mathcal{F},\mathcal{M}}$
which receives no quantum input, to conclude that $s$ can be learnt
by $\mathcal{A}_{\tilde{n}+1}^{\mathcal{F},\mathcal{M}}\vec{U}_{\tilde{n}}^{\mathcal{\vec{G}}_{\tilde{n}},\mathcal{N}_{\tilde{n}}}\dots\vec{U}_{1}^{\vec{\mathcal{G}}_{1},\mathcal{N}_{1}}\mathcal{A}_{1}^{\mathcal{F},\mathcal{M}}(\rho_{0})$
with at most negligible probability. More formally, denote the output
of $\mathcal{A}_{j}^{\mathcal{F},\mathcal{M}}$ by a random variable
$c_{j}$ and that of $\vec{U}_{j}^{\vec{\mathcal{G}}_{j},\mathcal{N}_{j}}$
by $q_{j}$. Assume (it only potentially makes the algorithms better)
that $q_{j}$ contains $c_{j-1}$ and $c_{j}$ contains $q_{j-1}$
(i.e. all information is propagated to the very last algorithm). Let
$\mathcal{E}$ denote any arbitrary algorithm which has no access
to any oracles. Let ${\rm coll}_{f}:=\{(x_{0},x_{1}):f(x_{0})=f(x_{1})\land x_{0}\neq x_{1}\}$
be the set of all colliding pre-images of $f$. Denote the $y$-distinct
event by $E$ and recall that $\Pr[\neg E]=\negl$. The two observations
may be stated as the following.
\begin{itemize}
\item If $\max_{\mathcal{E}}\Pr[\mathcal{E}(q_{j-1})\in{\rm coll}_{f}|E]\le\negl$
then, $\max_{\mathcal{E}}\Pr[\mathcal{E}(c_{j})\in{\rm coll}_{f}|E]\le\negl$
and $\max_{\mathcal{E}}\Pr[s\leftarrow\mathcal{E}(c_{j})|E]\le\negl$. 
\item If $\max_{\mathcal{E}}\Pr[\mathcal{E}(c_{j-1})\in{\rm coll}_{f}|E]\le\negl$
then, $\max_{\mathcal{E}}\Pr[\mathcal{E}(q_{j})\in{\rm coll}_{f}|E]\le\negl$
and $\max_{\mathcal{E}}\Pr[s\leftarrow\mathcal{E}(q_{j})|E]\le\negl$. 
\end{itemize}
Since $\mathcal{A}_{1}^{\mathcal{F},\mathcal{M}}$ receives no inputs,
one can apply these repeatedly to conclude 
\[
\Pr[s\leftarrow\mathcal{A}_{\tilde{n}+1}^{\mathcal{F},\mathcal{M}}\vec{U}_{\tilde{n}}^{\mathcal{\vec{G}}_{\tilde{n}},\mathcal{N}_{\tilde{n}}}\dots\vec{U}_{1}^{\vec{\mathcal{G}}_{1},\mathcal{N}_{1}}\mathcal{A}_{1}^{\mathcal{F},\mathcal{M}}(\rho_{0})]\le\Pr(E)\Pr[s\leftarrow\mathcal{A}_{\tilde{n}+1}^{\mathcal{F},\mathcal{M}}(q_{\tilde{n}})|E]+\Pr[\neg E]\le(1-\negl)\cdot\negl+\negl
\]
yielding \Eqref{StepTwo_again-1}.

\end{proof}

\printbibliography

\newpage{}

\appendix

\section{Permutations and Combinations}
\begin{fact}
One has 
\[
\frac{^{a}P_{b}}{^{a+1}P_{b+1}}=\frac{1}{a+1}\quad\text{and}\quad\frac{^{a}C_{b}}{^{a+1}C_{b+1}}=\frac{b}{a+1}.
\]
\end{fact}

\begin{rem}
\label{rem:Pr_x_in_X_or_t}Let $M\ge N$ be an integer and fix some
element $x\in\{1,2\dots M\}$. Suppose $t$ is a \emph{tuple} of size
$N$, sampled uniformly from the collection of all size $N$ \emph{tuples}
containing distinct elements from $\{1,2\dots M\}$. Then 
\[
\Pr(x\in t)=\frac{\perm{M-1}{N-1}\cdot N}{\perm MN}=\frac{N}{M}.
\]
Similarly, suppose $X$ is a \emph{set} of size $N$, sampled uniformly
from the collection of all size $N$ \emph{subsets} of $\{1\dots M\}$.
Then, again, 
\[
\Pr(x\in X)=\frac{\comb{M-1}{N-1}}{\comb MN}=\frac{N}{M}.
\]
\end{rem}

\section{Deferred Proofs}

\subsection{Hardness of $d$-Shuffled Simon's Problem for ${\rm QNC}_{d}$\label{subsec:Hardness-of-dSS-for-QNCd}}

\branchcolor{black}{\begin{proof}[Proof of \Thmref{d-SS-is-hard-for-QNC-d}]
Suppose $((f_{i})_{i=0}^{d},s)\sim\mathbb{F}_{{\rm SS}}(d,n)$ (see
\Defref{dShuffledSimonsDistr}) and let $\mathcal{F}$ be the oracle
associated with $(f_{i})_{i=0}^{d}$. Define $f:\{0,1\}^{n}\to\{0,1\}^{n}$
to be the Simon's function encoded in the $d$-Shuffler, i.e. $f(x):=f_{d}\circ\dots\circ f_{0}(x)$.
Denote an arbitrary ${\rm QNC}_{d}^{\mathcal{F}}$ circuit, $\mathcal{A}^{\mathcal{F}}$,
by 
\[
\mathcal{A}^{\mathcal{F}}:=\Pi\circ U_{d+1}\circ\mathcal{F}\circ U_{d}\dots\mathcal{F}\circ U_{2}\circ\mathcal{F}\circ U_{1}
\]
and suppose $\Pi$ corresponds to the algorithm outputting the string
$s$. For each $i\in\{1,\dots d\}$, construct the tuples $\bar{S}_{i}$
using \Algref{SforQNCd_using_dSS}. Let $\mathcal{G}_{i}$ be the
shadow of $\mathcal{F}$ with respect to $\bar{S}_{i}$ (see \Defref{d-ShuffledSimonShadow}
and \Figref{Shadows-for-d-ShuffledSimons}). Define 
\[
\mathcal{A}^{\mathcal{G}}:=\Pi\circ U_{d+1}\circ\mathcal{G}_{d}\circ U_{d}\dots\mathcal{G}_{2}\circ U_{2}\circ\mathcal{G}_{1}\circ U_{1}.
\]
Note that $\Pr[s\leftarrow\mathcal{A}^{\mathcal{G}}]\le\frac{1}{2^{n}}$
because no $\mathcal{G}_{i}$ contains any information about $f$
as $f_{d}$ is completely blocked (see \Figref{Shadows-for-d-ShuffledSimons}).
Thus, no algorithm can do better than making a random guess. We now
show that the output distributions of $\mathcal{A}^{\mathcal{F}}$
and $\mathcal{A}^{\mathcal{G}}$ cannot be noticeably different using
the O2H lemma (see \Lemref{O2H}).

To apply the lemma, one can use the hybrid method as before to obtain
(we drop the $\circ$ symbol for brevity):

\begin{align*}
 & \left|\Pr[s\leftarrow\mathcal{A}^{\mathcal{F}}]-\Pr[s\leftarrow\mathcal{A}^{\mathcal{G}}]\right|\\
= & \left|\tr[\Pi U_{d+1}\mathcal{F}U_{d}\dots\mathcal{F}U_{2}\mathcal{F}U_{1}\rho_{0}-\Pi U_{d+1}\mathcal{G}_{d}U_{d}\dots\mathcal{G}_{2}U_{2}\mathcal{G}_{1}U_{1}\rho_{0}]\right|\\
\le & \sum_{i=1}^{d}\sqrt{2\Pr[{\rm find}:U_{i}^{\mathcal{F}\backslash\bar{S}_{i}},\rho_{i-1}]}
\end{align*}
where $\rho_{0}=\left|0\dots0\right\rangle \left\langle 0\dots0\right|$
and $\rho_{i}=\mathcal{G}_{i}\circ U_{i}\circ\dots\mathcal{G}_{1}\circ U_{1}(\rho_{0})$
for $i>0$. To bound the aforesaid, we apply \Lemref{boundPfind}.
To this end, we must ensure the following. (1) The subset of queries
at which $\mathcal{F}$ and $\mathcal{G}_{i}$ differ, i.e. $\bar{S}_{i}=(\emptyset,\dots\emptyset,X_{i},X_{i+1}\dots X_{d})$
(where $X_{i}=f_{i}\circ\dots\circ f_{0}(\{0,1\}^{n})$; see \Algref{SforQNCd_using_dSS}),
is uncorrelated to $U_{i}$ and $\rho_{i-1}$. (2) The probability
that a fixed query lands in $\bar{S}_{i}$ is at most $\mathcal{O}(2^{-n})$.
Granted these hold, since $U$ acts on $\poly$ many qubits, $q$
in the lemma can be set to $\poly$. Thus, one can bound the last
inequality by $d\cdot\poly/2^{n}$. Using the triangle inequality,
one gets 
\[
\Pr[s\leftarrow\mathcal{A}^{\mathcal{F}}]\le\frac{\poly}{2^{n}}.
\]

The following complete the proof.

(1) This readily follows from \Propref{Shadow_dSS} and the observation
that $\rho_{i-1}$ has access to only $\mathcal{G}_{1},\dots\mathcal{G}_{i-1}$.
As for $U_{i}$, that is uncorrelated to all $X_{i}$s by construction.

(2) Follows directly from \Propref{x_in_S_dSS_shadow}.
\end{proof}
}

\subsection{Technical results for $\delta$ non-uniform distributions\label{subsec:tech_res_non-uniform}}

\branchcolor{black}{\begin{proof}[Proof of \Claimref{S_k_bound_general}]
 To see this for $S_{1}$, we proceed as before and recall the lower
bound $\Pr[S_{1}\subseteq\parts(t'_{1})]>2^{\delta|S_{1}|}\Pr[S_{1}\subseteq\parts(u)]$.
The upper bound may be evaluated as 
\begin{align*}
\Pr[S_{1}\subseteq\parts(t_{1}')] & =\Pr[S_{1}\subseteq\parts(t)|S\nsubseteq\parts(t)]\\
 & =\frac{\Pr[S_{1}\subseteq\parts(t)\land S\nsubseteq\parts(t)]}{\Pr[S\nsubseteq\parts(t)]}\\
 & =\frac{\Pr[S_{1}\subseteq\parts(u)\land S\nsubseteq\parts(u)\land g(u)=r']}{\Pr[S\nsubseteq\parts(t)]\Pr[g(u)=r']}\\
 & \le\Pr[S_{1}\subseteq\parts(u)]\cdot\gamma^{-2}
\end{align*}
where we used $\alpha'_{1}=1-\Pr[S\subseteq\parts(t)]=\Pr[S\nsubseteq\parts(t)]\ge\gamma$,
and $\Pr[g(u)=r']\ge\gamma$. In the general case, suppose $t'_{i}$s,
$t_{i}$s and $S_{i}$s are as described in the proof of \Propref{sumOfDeltaNonUni_perm}.
Then, one would have 
\begin{align}
\Pr[S_{i}\subseteq\parts(t_{i}')] & =\frac{\Pr[S_{i}\subseteq\parts(u)\land S_{i-1}\nsubseteq\parts(u)\land\dots S\nsubseteq\parts(u)\land g(u)=r']}{\Pr[S_{i-1}\nsubseteq\parts(t)\land\dots S\nsubseteq\parts(t)]\Pr[g(u)=r']}\label{eq:S_i_upperbound}\\
 & \le\Pr[S_{i}\subseteq\parts(u)]\cdot\gamma^{-2}\nonumber 
\end{align}
where $\alpha'_{i}=\Pr[S_{i-1}\nsubseteq\parts(t)\land\dots S\nsubseteq\parts(t)]>\gamma$
is assumed (else there is nothing to prove).
\end{proof}
}

\begin{prop*}[\Propref{composableP_Delta_non_beta_uniform} restated with slightly
different parameters]
 Let $t\sim\mathbb{F}^{\delta'|\beta}(N)$ be sampled from a $\delta'$
non-$\beta$-uniform distribution with $N=2^{n}$. Fix any $\delta>\delta'$
and let $\gamma=2^{-m}$ be some function of $n$. Let $s=t|(h(t)=r')$
and suppose $\Pr[h(t)=r']\ge\gamma$ where $h$ is an arbitrary function
and $r'$ some string in its range. Then $s$ is ``$\gamma$-close''
to a convex combination of finitely many $(p,\delta)$ non-$\beta$-uniform
distributions, i.e. 
\[
s\equiv\sum_{i}\alpha_{i}s_{i}+\gamma's'
\]
where $s_{i}\sim\mathbb{F}_{i}^{p,\delta|\beta}$ with $p=2m/(\delta-\delta')$.
The permutation $s'$ may have an arbitrary distribution (over $\Omega(2^{n})$)
but $\gamma'\le\gamma$.
\end{prop*}
\branchcolor{black}{\begin{proof}
While redundant, we follow the proof of \Propref{sumOfDeltaNonUni_perm}
adapting it to this general setting and omitting full details this
time. 

(For comparison: We replace $t$ with $s$ and $u$ with $b$)

\textbf{Step A:} Lower bound on $\Pr[S\subseteq\parts(s)]$.

Let $b\sim\mathbb{F}^{|\beta}(N)$. Suppose $s$ is not $\delta$
non-$\beta$-uniform. Then consider the largest $S\in\Omega_{\parts}(N)$
such that

\begin{equation}
\Pr[S\subseteq\parts(s)]>2^{\delta.|S|}\cdot\Pr[S\subseteq\parts(b)].\label{eq:_S_not_delta}
\end{equation}
 
\begin{claim}
Let $S$ and $s$ be as described. The random variable $s$ conditioned
on being consistent with the paths in $S\in\Omega_{\parts}(N)$, i.e.
$s_{S}:=s|(S\subseteq\parts(s))$, is $\delta$ non-$\beta$-uniformly
distributed. 
\end{claim}

We give a proof by contradiction. Suppose $s_{S}$ is ``more than''
$\delta$ non-$\beta$-uniform. Then there exist some $S'\in\Omega_{\parts}(N,S)$
such that 
\[
\Pr[S'\subseteq\parts(s)|S\subseteq\parts(s)]>2^{\delta\cdot|S'|}\Pr[S'\subseteq\parts(b)|S\subseteq\parts(b)].
\]
Then 
\begin{align*}
\Pr[S\cup S'\subseteq\parts(s)] & =\Pr[S\subseteq\parts(s)]\Pr[S'\subseteq\parts(s)|S\subseteq\parts(s)]\\
 & >2^{\delta\cdot|S\cup S'|}\cdot\Pr[S\cup S'\subseteq\parts(b)]
\end{align*}
using \Eqref{_S_not_delta} and \Eqref{_S'_not_delta}. That's a contradiction
to $S$ being maximal.

\textbf{Step B:} Upper bound on $\Pr[S\subseteq\parts(s)]$.
\begin{claim}
One has $\left|S\right|<m/(\delta-\delta')$. 
\end{claim}

To see this, observe that

\begin{align*}
\Pr[S\subseteq\parts(s)] & =\Pr[S\subseteq\parts(t)\land h(t)=r']\cdot\Pr[h(t)=r']\\
 & \le\Pr[S\subseteq\parts(t)]\cdot\gamma^{-1}\\
 & \le2^{\delta'|S|}\Pr[S\subseteq\parts(b)]\cdot\gamma^{-1}
\end{align*}
and comparing this with the lower bound, one obtains $\left|S\right|<m/(\delta-\delta')$. 

The remaining proof \Propref{sumOfDeltaNonUni_perm} similarly generalises
by proceeding in the same vein. More concretely, suppose $S_{i}$,
$s_{i}$, $s_{i}'$ are defined analogously. Then the lower bound
goes through almost unchanged while for the upper bound, the analogue
of \Eqref{S_i_upperbound} becomes 
\begin{align*}
\Pr[S_{i}\subseteq\parts(s'_{i})] & =\frac{\Pr[S_{i}\subseteq\parts(s)\land S_{i-1}\nsubseteq\parts(s)\land\dots S\nsubseteq\parts(s)]}{\Pr[S_{i-1}\nsubseteq\parts(s)\land\dots S\nsubseteq\parts(s)]}\\
 & \le\Pr[S_{i}\subseteq\parts(t)\land S_{i-1}\nsubseteq\parts(t)\land\dots S\nsubseteq\parts(t)|h(t)=r']\cdot\gamma^{-1}\\
 & \le\frac{\Pr[S_{i}\subseteq\parts(t)]}{\Pr[h(t)=r']}\cdot\gamma^{-1}\le2^{\delta'|S_{i}|}\Pr[S_{i}\subseteq\parts(b)]\cdot\gamma^{-2}.
\end{align*}
\end{proof}
}

\section{Discussions}

\subsection{Why $\delta$ non-$\beta$-uniform doesn't work for Simon's functions\label{subsec:Delta_with_Simons_not_smart}}

We briefly discuss why the concept of $\delta$ non-$\beta$-uniform
distribution does not prove useful when applied to the distribution
over Simon's functions. Suppose an algorithm takes an oracle for a
function $f$ (sampled from an arbitrary $\beta$-uniform distribution)
and an advice as an input. We want to argue that the algorithm would
behave essentially the same if it were not given the advice. That
is clearly not true in this case where $f$ is a Simon's function.
The proposition still holds, i.e. the distribution conditioned on
the advice is still $(p,\delta)$ non $\beta$ uniform but this conditioned
distribution reveals too much information already. 

To be more concrete, suppose the algorithm is supposed to verify if
the period is $s$ and output $r=1$ if it succeeds at verifying.
We can write $\Pr[\mathcal{A}(s,u|s)=r]=\sum_{t_{i}}\Pr(t_{i})\Pr[\mathcal{A}(s,t_{i})=r]$
using the main proposition where we know $t_{i}$ are $\delta$ far
from $\mathbb{F}_{{\rm Simon}}$. Naively, one might see a contradiction.
Suppose the algorithm simply checks if $u(s)=u(0)$ and outputs $r=1$
if it is. For the LHS, the probability of outputting $r=1$ is $1$.
For the RHS, it appears that because $t_{i}$ are $\delta$ far from
$\mathbb{F}_{{\rm Simon}}^{p}$, $t_{i}(s)=u(0)$ will be much smaller
than $1$ because $\mathbb{F}_{{\rm Simon}}$ is uniformly distributed
over all Simon's functions and so a distribution $\delta$ far from
$\mathbb{F}_{{\rm Simon}}^{p}$ would also behave similarly because
$p$ fixes at most polynomially many paths. However, this reasoning
is flawed because once even a single colliding path is specified,
$\mathbb{F}_{{\rm Simon}}^{p}$ can only contain functions with period
$s$. Thus, each term in the RHS also outputs $1$ with certainty.

\end{document}